\documentclass{article}
\usepackage{amsmath,amsthm,amssymb}

\usepackage[pdftex]{graphicx}%图形宏包

\usepackage{mathrsfs}
\usepackage{cases,color}
\usepackage{cite}
\usepackage{multicol}
\usepackage{subfig}
\newtheorem{theorem}{\bf Theorem}[section]
\newtheorem{algorithm}{\bf Algorithm}
\newtheorem{lemma}{\bf Lemma}[section]

\newtheorem{proposition}{\bf Proposition}[section]

\newtheorem{remark}{\bf Remark}[section]
\renewenvironment{proof}{{\noindent\bf Proof:}}{\hfill$\Box$\vskip 0.5\baselineskip}

\def\div{\mbox{div}}

\def\det{\mbox{det}}

\def\dfrac{\displaystyle\frac}

\newcommand{\bm}[1]{\mbox{\boldmath{$#1$}}}

\begin{document}

\title{{\Large \textbf{ Geometry Method  for the Rotating Navier-Stokes
Equations With Complex Boundary and\\
the Bi-Parallel Algorithm
 }} \thanks{%
Subsidized by NSFC Project 10971165 ,10471110,10471109.}}
\author{{\ Kaitai Li \quad Demin Liu  } \\
%EndAName
{(College of Science, Xi'an Jiaotong University},\\
{\ Xi'an, 710049,P.R.China, Email: ktli@xjtu.edu.cn)}}
\date{}
\maketitle

\begin{minipage}{13cm}
\label{ablstract}

 {\bf Abstract:}\quad In this paper, a new
algorithm based on differential geometry viewpoint to solve the 3D
rotating Navier-Stokes equations with complex Boundary  is proposed,
which is called Bi-parallel algorithm. For example, it can be
applied to passage flow  between two blades in impeller and
circulation flow through aircrafts with complex geometric shape of
boundary. Assume that a domain in $R^3$ can be decomposed into a
series sub-domain, which is called ``flow layer'', by a series
smooth surface $\Im_k, k=1,\ldots,M$. Applying differential geometry
method, the 3D Navier-Stokes operator can be split into two kind of
operator: the ``Membrane Operator'' on the tangent space at the
surface $\Im_k$ and the ``Bending Operator'' along the transverse
direction. The Bending Operators are approximated by the finite
different quotients and restricted the 3D Naver-Stokes equations  on
the interface surface $\Im_k$, a Bi-Parallel algorithm can be
constructed along two directors: ``Bending'' direction and
``Membrane'' directors. The advantages of the method are that: (1)
it can improve the accuracy of approximate solution caused of
irregular mesh nearly the complex boundary; (2) it can overcome the
numerically effects of boundary layer, whic is a good boundary layer
numerical method; (3) it is sufficiency to solve a two dimensional
sub-problem without solving 3D sub-problem.   

{\bf Subject Classification(AMS):} 65N30, 76U05, 76M05

\end{minipage}

\section{Introduction}
\label{sect:one}

As well known, numerical simulation for 3D viscous flow in
turbo-machinery and circulation flow through  aircraft meet three
kinds of difficulty cause of  nonlinearity, three dimensional grid
and boundary layer effects with complex boundary.  In order to
overcome last two difficulties, a new dimensional slitting method
based on  differential geometry method is proposed.

As well known that classical domain decomposition method is that the
3D domain is made in the sum of severest overlap or non overlap
subdomain, an approximated solutions can be  established then by
solving 3D subproblem in each subdomain, see e.g., [2,4,19,24] . The
method proposals by authors here , called "bi-parallel algorithm ",
for 3D Navier-Stokes equations, is that the 3D domain $\Omega$
occupied by the fluids in $\Re^3$ is decomposed into the sum of
servery sub domains ( called layer) $\Omega^i_{i-1}$ by servery 2D
surfaces(2D manifold) $\Im_i,i=1,2,\cdots,m$.   3D Navier-Stokes
operators in the layer $\Omega^i_{i-1}\cup \Omega^{i+1}_i$ under a
new coordinate based on the surface $\Im_i$ can be represented as
the sum of "membrane" operators on tangent space at $\Im_i$ and
normal (bending) operator to $\Im$. By Eular central difference
quotient instead of bending operator and restricting 3D
Navier-Stokes equations on the $\Im_i$, a three components and two
dimensional Navier-Stokes equations (is called 3C-2D NSE) on $\Im_i$
are obtained. After successively iterations, an approximate solution
of 3D NSE can be established. It is obvious that the method is
different from the classical domain decomposition method because we
only solve a two dimensional problem in each sub domain without
solving 3D subproblems. In addition, other advantages of this method
are the followings:

(i)\quad For the complex boundary geometry, for examples, in
turbo-machinery flow with complex shape of blades of impeller, in
geophysical flow with real surface of the earth and in the
circulation flow passing through complex aircraft etc., 80 percent
of the freedoms of 3D mesh should be concentrated on a thin domain
of the boundary layer, even using different methods or finite
element method, there exists a lot of difficulty and even influence
of the accuracy ; in our method the distance between two surface can
be very small as you wanted, owing to parametrization of surface,
the domain for 2D-3C subproblem is a bounded domain in $\Re^2$, and
the mesh is 2D-mesh;

(ii)\quad Since interface surfaces are chosen such that most of flux
flows along the interface while  small flux of the fluid run cross
the interface only, iterative method is very well suitable for the
physical properties of flow;

(iii)\quad This is a better method for treatment of boundary layer
phenomenons, which can refer to [19] and references therein. Indeed
this is a good boundary layer model and associated algorithm.

This paper is organized as follows: in  section 2 we give the
preliminary for the geometry of the blade's surface and construct a
new coordinate system; in section 3, we derive the rotating
Navier-Stokes equations in the new coordinate system; in section 4,
we study The equations for the average velocity along the rotating
direction; in section 5, provided  the equations of the
G$\widehat{a}$teaux Derivatives of the solutions of Navier-stokes
equations with respect to shape of boundary; in section 6 we provide
2D-3C Navier-Stokes equations on the Surface $\Im_\xi$; in section
7, we provide a corrected equation for the pressure in order to
improvement of the accuracy of the pressure's computation; in
section 8, we present the Bi-parallel Algorithm; in section 9, we
prove the existence of the solutions for 2D-3C Navier-Stokes
equations and the solution of correcting equation for the pressure,
and discuss the dependence of the corrected pressure upon the
velocity; and finally, an  appendix is supplied.

%\section{Preliminary--The Geometry of the Blade's Surface}
%\section{\textbf{Rotating Navier-Stokes Equations  in the New Coordinate System}}

%\section{The equations For the average velocity along the Rotating Direction}

%\section{ The Equations for the G\^{a}teaux
%   Derivative of the solutions of NSE with Respect to the Shape of Boundary}

%\section{ 2D-3C Navier-Stoke Equations  on the 2D manifold $\Im_\xi$}

%\section{ Pressure Correction Equation on the Blade Surface}

%\section{Steam Layer in Domain decomposition and the Bi-parallel Algorithm}

%\section{\bf Existence of Solution of the 2D-3C Variational Problem}

%\section {\bf Finite Element Approximation Based on Approximate Inertial
%Manifold}

%\appendix

\section{Preliminary--The Geometry of the Blade's Surface}
\label{sect:two}

we use Einstein summation convention throughout this paper, in order
to distinguish, Greek indices and exponents belong to the set
$\{1,2\}$, Latin indices and exponents belong to the set
$\{1,2,3\}$, and the repeated index implies that we are summing over
all of its possible values.  The dot product and the cross product
of two vectors $\bm a,\bm b\in R^3$ are denoted by $\bm a\cdot\bm b$
, $\bm a\times \bm b$, respectively, and the Euclidean norm $|\bm
a|=\sqrt{\bm a\cdot\bm a}$.

Let $D$ be an open bounded connect subset of $R^2$, whose boundary
$\gamma$ is Lipschitz-continuous. Suppose  $D$ is locally on one
side of $\gamma$. In the following of this paper, we also suppose
that the blade of impeller  is so thin that it can be described as
an smoothing injective mapping $\Im=\bm\Re(D)$ in mathematics. In
this case, $D$ is the projection area of blade $\Im$ on the meridian
plane of impeller. Let $x=(x^\alpha)$ is an arbitrary point in the
set $\overline{D}$, and $\partial_\alpha=\partial/\partial
x^\alpha.$ If the mapping $\bm\Re(x)$ is smoothing enough, for
example, $\bm\Re(x)\in C^3(\overline{D};R^3)$, then there exist two
vectors $\bm e_\alpha(x)=\partial_\alpha \bm \Re(x)$, which are
linear independent at all point $x\in\overline{D}$. $\bm
e_\alpha(x)$ can be regards as the basis vectors of the tangential
surface at this point. Meanwhile, the unit normal vector at the same
point is
$$\bm n=\frac{\bm e_1\times \bm e_2}{|\bm e_1\times \bm
e_2|}.$$
 So  $(\bm e_\alpha,\bm n)$ constitute the covariant basis vectors at the
 point $\bm \Re(x)$. We can further define contra-variant basis vectors $(\bm
e^\alpha,\bm n)$ which satisfy the following relations
$$\bm e^\alpha\cdot\bm e_\beta=\delta^\alpha_\beta,\quad \bm e^\alpha
\cdot\bm n=0,\quad \bm n\cdot\bm n=1,$$ where $\delta^\alpha_\beta$
denotes  the Kronecker delta. It is well known that $\bm e^\alpha$
are also on the tangential plane to $\Im$ at point $\bm \Re(x)$.

The covariant components $a_{\alpha\beta}$ and contra-variant
components $a^{\alpha\beta}$ of the metric tensor of $\Im$, the
Christoffel symbols $\stackrel{\ast}{\Gamma^\alpha_{\beta\lambda}}$
, and the covariant components $b_{\alpha\beta}$ and mixed
components
 $b^\beta_\alpha$ of the curvature tensor of $\Im$ are
defined as, respectively,
$$\begin{array}{ll}
a_{\alpha\beta}:=\bm e_\alpha \cdot\bm e_\beta,\quad
a^{\alpha\beta}=\bm e^\alpha \cdot\bm e^\beta,\quad
\stackrel{\ast}{\Gamma^\alpha_{\beta\lambda}}:=\bm
e^\alpha\cdot\partial_\beta
\bm e_\lambda,\\
b_{\alpha\beta}:=\bm n\cdot\partial_\beta \bm e_\alpha,\quad
b^\beta_\alpha=a^{\beta\sigma}b_{\alpha\sigma}.
\end{array}$$
The area element along $\Im$ is $\sqrt{a}dx$, where
$a=\det{(a_{\alpha\beta})},\ \sqrt{a}=|\bm e_1\times \bm e_2|.$

\begin{figure}
  \centering
  \subfloat[Impeller]{\label{fig:gull}\includegraphics[width=0.3\textwidth]{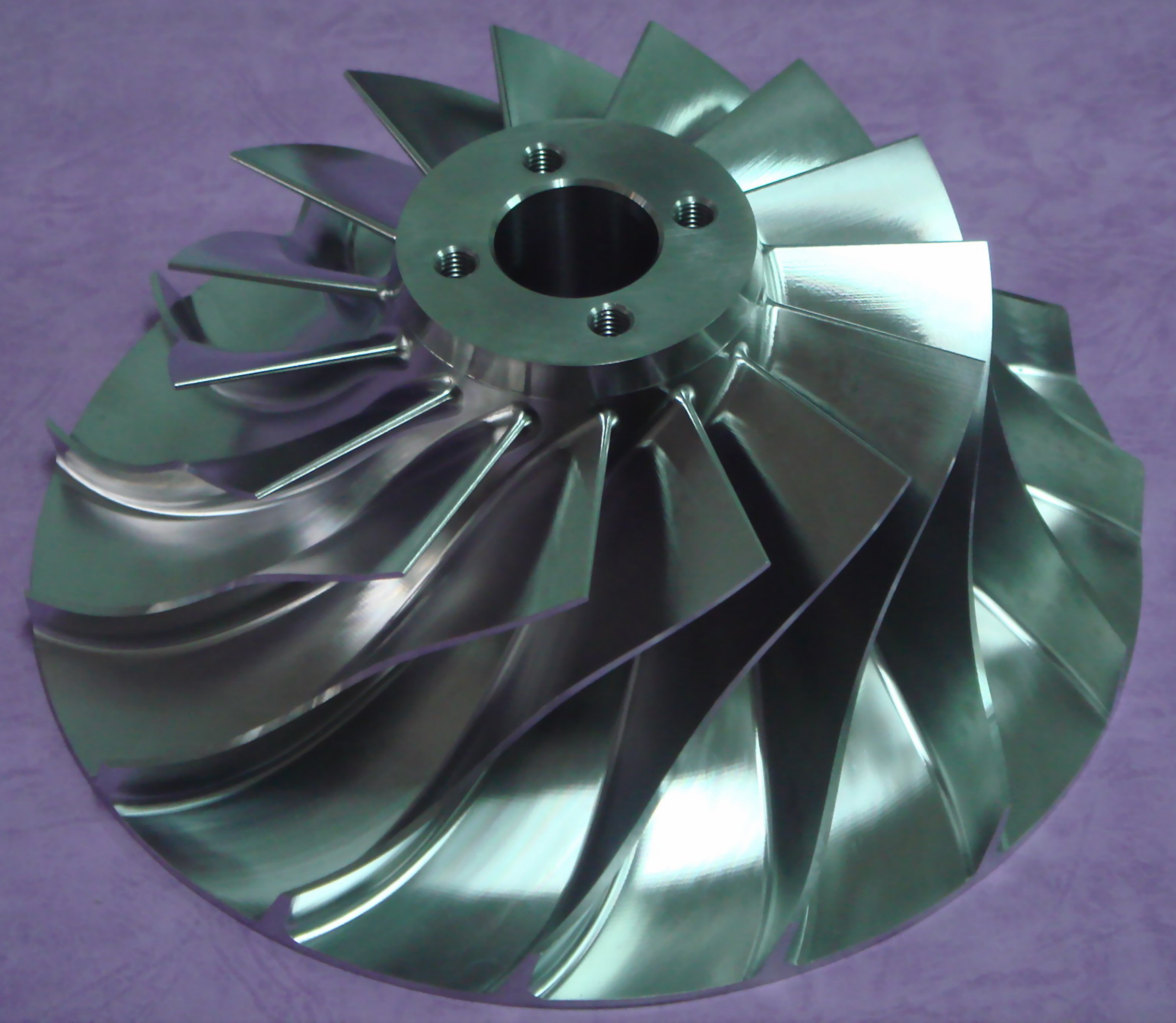}}\hspace{2cm}
  \subfloat[Flow Passage]{\label{fig:tiger}\includegraphics[width=0.3\textwidth,height=0.3\textwidth]{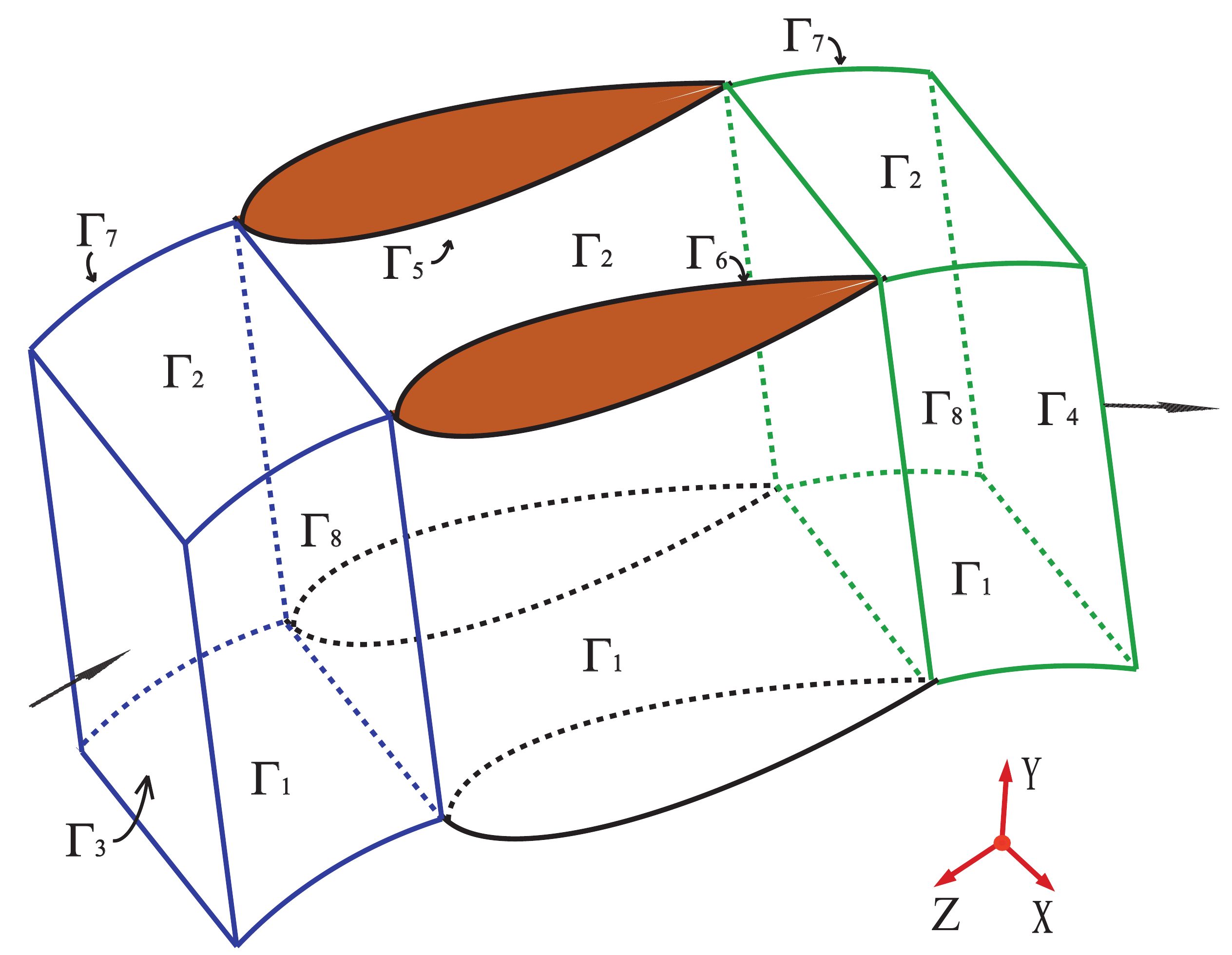}}
  \caption{Impeller and Flow Passage}
  \label{fig:ImpellerPassage}
\end{figure}

Nest, we assume that the impeller is rotating along z-axis with
angular velocity $\bm\omega=(0,0,\omega)$. Let $(\bm e_r,\bm
e_\theta,\bm k)$ be the cylindrical basis vectors established on the
impeller and rotating with the impeller(see Appendix). Let constant
N be the number of blade and $\varepsilon=\pi/N$. Then by rotating
$\frac{2\pi}{N}$ degrees, one blade is rotated to the location of
the next one. The flow passage $\Omega_\varepsilon$ of the impeller
is the inner part of the boundary $\partial
\Omega_\varepsilon=\Gamma_{in}\cup\Gamma_{out}\cup\Gamma_t\cup\Gamma_b\cup
\Im_+\cup \Im_{-}$. Further, the blade is $\Im=\bm{\Re}(D)$, and an
arbitrary point $\bm\Re(x)$ on the blade $\Im$ can be expressed as
$$\bm\Re(x)
=x^2\bm e_r+x^2\Theta(x)\bm e_\theta+x^1\bm k, \eqno{(2.1)}$$ where
$x=(x^1,x^2)\in \overset{-}{D}$ are Gauss coordinate system on
$\Im$, and $\Theta \in C^2(D,R)$ is an smoothing enough function.

 It is easy to prove that there exists a
family of surfaces $\Im_\xi$, only depending on parameter $\xi$,
which cover the domain $\Omega_\varepsilon$ by mapping
$\bm\Re(x;\xi):\ D\rightarrow \Im_\xi$, where
$$\bm\Re(x;\xi)=x^2\bm e_r+x^2(\varepsilon\xi+ \Theta(x))\bm e
_\theta+x^1\bm k.\eqno{(2.2)}
$$
It is easy to prove that the metric tensor $a_{\alpha\beta}$ of
$\Im_\xi$ is homogenous, nonsingular and independent of $\xi$, which
is given as follows
$$\begin{array}{ll}
a_{\alpha\beta}=\frac{\partial \bm\Re(x;\xi)}{\partial
x^\alpha}\frac{\partial \bm\Re(x;\xi)}{\partial x^\beta}=
\delta_{\alpha\beta}+r^2\Theta_\alpha\Theta_\beta,\quad
a^{\alpha\beta}a_{\beta\lambda}=\delta^\alpha_\lambda,\\
a=\mbox{det} (a_{\alpha\beta})=1  +r^2(\Theta_1^2+\Theta_2^2)>0.
 \end{array}\eqno(2.3)$$

Let $(x^{1^{\prime }},x^{2^{\prime }},x^{3^{\prime }})=(r,\theta,z)$
, the corresponding metric tensor in $R^3$ are $(g_{1^{\prime
}1^{\prime }}=1,g_{2^{\prime }2^{\prime }}=r^2,g_{3^{\prime
}3^{\prime }}=1,g_{i^{\prime }j^{\prime }}=0 \forall i^{\prime }\neq
j^{\prime })$. According to rule of tensor transformation under
coordinate transformation we have the following calculation formulae
$$g_{ij}=g_{i^{\prime }j^{\prime }}\frac{\partial x^{i^{\prime }}}{\partial
x^{i}} \frac{\partial x^{j^{\prime }}}{\partial x^{j}}.$$
Substituting (2.3) into the above formula, the covariant  and
contra-variant components of the metric tensor of $R^3$ in the new
curvilinear coordinate system are give by
$$\left\{\begin{array}{ll}
g_{\alpha\beta}=a_{\alpha\beta},\quad
g_{3\beta}=g_{\beta3}=\varepsilon r^2\Theta_\beta,\quad
g_{33}=\varepsilon^2 r^2,\quad g=\mbox{det}(g_{ij})=\varepsilon^2
r^2,\\
g^{\alpha\beta}=\delta^{\alpha\beta},\,
g^{3\beta}=g^{\beta3}=-\varepsilon^{-1}\Theta_\beta,\quad
g^{33}=\varepsilon^{-2}r^{-2}(1+r^2|\nabla\Theta|^2)=(r\varepsilon)^{-2}a,
\end{array}\right.\eqno{(2.5)}$$
 where
$|\nabla\Theta|^2=\Theta_1^2+\Theta^2_2,$ and
$\Theta_\alpha=\frac{\partial\Theta}{\partial x^\alpha}$.

Tensor calculations show that the following proposition is right(see
Appendix ),

\begin{proposition}
let $(\bm e_\alpha,\bm  n)$ denote the basis vectors in new
coordinate system $(x,\xi)$ , and a vector $\bm v\in R^3$ can be
wrote as $\bm v=v^\alpha \bm e_\alpha+v^3\bm n$. Further more, we
have the following representation formula in new coordinate system,

1.Angular velocity vector $\bm \omega$
$$\left\{\begin{array}{l}
\bm \omega=\omega\bm e_1-\omega\varepsilon^{-1}\Theta_1\bm e_3\\
\omega^1=\omega,
\quad\omega^2=0,\quad\omega^3=-\omega\varepsilon^{-1}\Theta_1,
\end{array}\right. \eqno(2.6)$$

2.Coriolis Force
$$\left\{\begin{array}{l}
\bm C=2 \bm \omega\times \bm w=C^i\bm e_i,\\
C^1=0,\quad C^2= -2\omega r\Pi(\bm w,\Theta),\quad C^3=
2\omega\varepsilon^{-1}(r\Theta_2\Pi(\bm w,\Theta)
+\dfrac{w^2}{r}),\\
\Pi(\bm w,\Theta)=\varepsilon w^3+w^\lambda\Theta_\lambda,
\end{array}\right.\eqno(2.7)$$

3.  Unit normal vector to $\Im$
$$\begin{array}{l}
\bm n=\dfrac{\bm e_1\times\bm e_2} {|\bm e_1\times\bm e_2|}=n^i\bm
e_i,\quad n^\lambda=-r\Theta_\lambda/\sqrt{a},\quad n^3=(\varepsilon
r)^{-1}\sqrt{a},
\end{array}\eqno(2.8)$$

 4. Second fundamental form(curvature tensors for 2D
 manifold)
 $$\begin{array}{ll}
b_{11}=\frac1{\sqrt{a}}(\Theta_2(a_{11}-1)+r\Theta_{11}),\quad
b_{12}=\frac1{\sqrt{a}}(\Theta_1(a_{12}+r\Theta_{12}),\\
b_{22}=\frac1{\sqrt{a}}(\Theta_2(a_{22}+1)+r\Theta_{22}),\quad
b=\det{(b_{\alpha\beta})}=b_{11}b_{22}-b_{12}^2.
\end{array}\eqno{(2.9)}$$

5. Mean Curvature $H$ and Gaussian Curvature $K$
$$\left\{\begin{array}{l}
H=\frac1{2a\sqrt{a}}\big(\Theta_2(a+a_{11}a_{22}+a_{11}-a_{22})+r(a_{22}\Theta_{11}+a_{11}\Theta_{22}-2a_{12}\Theta_{12})\big),
\\
K=\frac{b}{a}=\mbox{det}(b_{\alpha\beta})/a.
\end{array}\right.\eqno{(2.10)}$$
\end{proposition}

It is obvious that each $\xi=const$ corresponds to a surface
$\Im_\xi$ which has the same geometry properties with  $\Im$. It is
well known that the geometry properties of $\Im$ is completely
determined by ($a_{\alpha\beta}$) and ($b_{\alpha\beta}$) in the
following meaning.

Let ${\cal O}^3$ denote the set of all third-order orthogonal
matrices, and that ${\cal O}^3_{+}=\{Q\in {\cal O}^3, \det(Q)=1\}$
denotes the subset of all proper third-order orthogonal matrices.
Then $\textbf{J}_{+}(x)=\bm C+Q\circ\textbf{x}$ is a proper isometry
of $R^3\rightarrow R^3$, where $\bm C\in R^3, Q\in {\cal O}^3_+$. We
have

\begin{theorem}
\label{unique}
 Two immersions $\bm\Re_1\in C^1(D;R^3)$ and
$\bm\Re_2\in C^1(D;R^3)$ share the same fundamental forms
$(a_{\alpha\beta}),\ (b_{\alpha\beta})$ over an open connected
subset $D\in R^3$ if and only if
$$\bm\Re_2=\textbf{J}_+\circ\bm\Re_1,
\eqno{(2.11)}$$
 Furthermore, If two matrices fields
$(a_{\alpha\beta})\in C^2( D; {\cal S}^2),\ (b_{\alpha\beta})\in
C^2( D;{\cal S}^2)$ satisfy Gauss-Codazzi equations in $D$, i.e.,
$$\begin{array}{ll}
\partial_\beta \Gamma_{\alpha\sigma,\tau}-\partial_\sigma
\Gamma_{\alpha\beta,\tau}+\Gamma^\mu_{\alpha\beta}\Gamma_{\sigma\tau,\mu}
-\Gamma^\mu_{\alpha\sigma}\Gamma_{\beta\tau,\mu}=b_{\alpha\sigma}b_{\beta\tau}-b_{\alpha\beta}b_{\sigma\tau},\\
\partial_\beta b_{\alpha\sigma}-\partial_\sigma
b_{\alpha\beta}+\Gamma^\mu_{\alpha\sigma}b_{\beta\mu}-\Gamma^\mu_{\alpha\beta}b_{\sigma\mu}=0,
\end{array}$$ where
$\begin{array}{ll} \Gamma_{\alpha\beta,\tau}=\frac12(\partial_\alpha
a_{\alpha\tau}+\partial_\alpha a_{\beta\tau}-\partial_\tau
a_{\alpha\beta}),\
\Gamma^\sigma_{\alpha\beta}=a^{\sigma\tau}\Gamma_{\alpha\beta,\tau}
\end{array}$.
Then there exists an immersion $\bm\Re\in\,C^3(D; R^3)$ such that
$$ a_{\alpha\beta}=\partial_\alpha \bm\Re\partial_\beta\bm\Re,\quad
b_{\alpha\beta}=\partial^2_{\alpha\beta}\bm\Re\cdot\{\frac{\partial_1\bm\Re\times\partial_2\bm\Re}
{|\partial_1\bm\Re\times\partial_2\bm\Re|}\}.$$
\end{theorem}

Because the surface $\Im_\xi$ is obtained by a $\xi\varepsilon$
degree rotation of $\Im$, so by using theorem \ref{unique}, we have
the surface $\Im_\xi$ have the same geometry properties with $\Im$
for any $\xi\in[-1,1]$, i.e., their corresponding geometric
quantities $a_{\alpha\beta},b_{\alpha\beta},K,H,\cdots$ are same.

In subsequent discussion, we will often employ the third fundamental
tensor of $\Im$, i.e.,
$$c_{\alpha\beta}=a^{\lambda\sigma}b_{\alpha\lambda}b_{\beta\sigma},\eqno{(2.12)}$$
and  inverse matrix
$(\widehat{c}^{\alpha\beta})=(c_{\alpha\beta})^{-1},
(\widehat{b}^{\alpha\beta})=(b_{\alpha\beta})^{-1}$ satisfy the
following relations,
$$\widehat{b}^{\alpha\beta}b_{\beta\lambda}=\delta^\alpha_\lambda,\quad
\widehat{c}^{\alpha\beta}c_{\beta\lambda}=\delta^\alpha_\lambda.\eqno{(2.13)}$$

Furthermore,let introduce permutation tensors in Euclid space $R^3$
$$\varepsilon_{ijk}=\left\{\begin{array}{r}
    \sqrt{g},\\[0.25cm]
    -\sqrt{g},\\[0.25cm]
     0,
\end{array}\right.
\varepsilon_{ijk}=\left\{\begin {array}{rl}
\frac{1}{\sqrt{g}},& \mbox{(i,j,k) is even permutation of (1,2,3)},\\[0.25cm]
-\frac{1}{\sqrt{g}},&\mbox{(i,j,k) is odd permutation of (1,2,3)},\\[0.25cm]
0,& \mbox{otherwise},
\end {array}\right. \eqno{(2.14)}$$
where $g=\mbox{det}(g_{ij})$. Similar  permutation tensors on 2D
manifold $\Im$
$$\varepsilon_{\alpha\beta}=\left\{\begin{array}{r}
\sqrt{a},\\[0.25cm]
-\sqrt{a},\\[0.25cm]
0,
\end{array}\right.
\varepsilon_{\alpha\beta}=\left\{\begin {array}{rl}
 \frac{1}{\sqrt{a}},& (\alpha,\beta)\mbox{ is even permutation of (1,2)},\\[0.25cm]
-\frac{1}{\sqrt{a}},& (\alpha,\beta)\mbox{ is odd permutation of (1,2)},\\[0.25cm]
                  0,& \mbox{otherwise}.
\end {array}\right. \eqno(2.15)$$

\section{\textbf{Rotating Navier-Stokes Equations  in the New Coordinate System}}
\label{sect:three}

we consider the rotating impeller with rotating angular velocity
$\bm\omega=(0,0,\omega)$. Under the rotating cylindrical coordinate
established on the impeller, The motion of fluid in the flow passage
is governed by the three-dimensional rotating Navier-Stokes
equations, i.e.,
$$
\left\{
\begin{array}{ll}
\displaystyle\frac{\partial\rho}{\partial t}+\mbox{div}(\rho \bm w)=0, &  \\
\rho\bm a=\mbox{div}\sigma+\bm f, &  \\
\rho c_v(\displaystyle\frac{\partial T}{\partial t}+w^j\nabla_j
T)-\mbox{div}(\kappa\mbox{grad} T)+p\mbox{div} \bm w-\Phi=h, &  \\
p=p(\rho,T), &
\end{array}
\right.\eqno(3.1)
$$
where $\rho$ the density of the fluid, $w$ the velocity of the
fluid, $h$ the heat source, $T$ the temperature, $k$ the coefficient
of heat conductivity, $C_v$ specific heat at constant volume, and
$\mu$ viscosity. Furthermore, the strain rate tensor, stress tensor,
dissipative function and viscous tensor are respectively given by
$$\begin{array}{ll}
e_{ij}(\bm w)=\displaystyle\frac12(\nabla_i w_j+\nabla_j w_i),&
e^{ij}(\bm w)=g^{ik}g^{jm}e_{km}(w)=\displaystyle\frac
12(\nabla^i w^j+\nabla^j w^i),\\
\sigma^{ij}(\bm w,p)=A^{ijkm}e_{km}(w)- g^{ij}p,& A^{ijkm}=\lambda
g^{ij}g^{km}+\mu(g^{ik}g^{jm}+g^{im}g^{jk})\\
\Phi=A^{ijkm}e_{ij}(w)e_{ij}(\bm w),& \lambda=-\frac23\mu,
\end{array}%
\eqno{(3.2)}
$$
where $g_{ij}$, and $g^{ij}$ are the covariant and contra-variant
components of the metric tensor of three-dimensional Euclidian space
in the curvilinear coordinate $(x,\xi)$ define by (2.4),
respectively. The covariant derivatives of velocity vector and
Christoffel symbols are
$$\nabla_i w^j=\displaystyle\frac{\partial w^j}{\partial x^i}
+\Gamma^j_{ik}w^k,\quad \nabla_i w_j=\displaystyle\frac{\partial
w_j}{\partial x^i}-\Gamma^k_{ij}w_k,\quad
\Gamma^i_{jk}=g^{il}(\displaystyle\frac{\partial g_{kl}}{\partial
x^j}+ \displaystyle\frac{\partial g_{jl}}{\partial x^k}
-\displaystyle\frac{
\partial g_{jk}}{\partial x^l}).
\eqno{(3.3)}
$$
The absolute acceleration of the fluid is given by
$$\begin{array}{l}
\bm a=\displaystyle\frac{\partial\bm w}{\partial t}+(\bm w\nabla)\bm
w+2\bm \omega\times\bm w+\bm \omega\times(\bm \omega\times
\bm{\mathcal{R}}),\\
 a^i=\displaystyle\frac{\partial w^i}{\partial
t}+w^j\nabla_j w^i +2\varepsilon^{ijk}\omega_j w_k-\omega^2r^i.
\end{array} \eqno{(3.4)}$$
where $\bm{\mathcal{R}}$ is the radium vector of the fluid particle.
The flow passage occupied by the fluids is denoted by
$\Omega_\varepsilon$,  which boundary $\partial \Omega_\varepsilon$
is the union of inflow boundary $\Gamma_{in}$, outflow boundary
$\Gamma_{out}$, positive blade's surface $\Im_+$, negative blade's
surface $\Im_-$ and top wall $\Gamma_t$ and bottom wall $\Gamma_b$,
i.e.,
$$
\partial\Omega_\varepsilon=\Gamma=\Gamma_{in}\cup\Gamma_{out}\cup \Im_-\cup
\Im_+\cup\Gamma_t\cup\Gamma_b\quad \mbox{(见Fig.2)}\eqno{(3.5)}.
$$
The boundary conditions are given by
$$
\left\{
\begin{array}{lll}
\bm w|_{\Im_-\cup \Im_+}=\bm 0,& \bm w|_{\Gamma_b}=\bm 0, & \bm
w|_{\Gamma_t}=\bm 0,
\\
\sigma^{ij}(\bm w,p)n_j|_{\Gamma_{in}}=g^i_{in},& \sigma^{ij}(\bm
w,p)n_j|_{\Gamma_{out}}=g^i_{out}, &
\mbox{自然边界条件} \\
\displaystyle\frac{\partial T}{\partial n}+\lambda(T-T_0)=0,&
\mbox{$\lambda\geq 0$为常数}.&
\end{array}
\right.\eqno(3.6)
$$
We also supply the initial condition
$$ \bm w|_{t=0}=\bm w_0(x).$$
If the fluid is incompressible and flow is stationary, then the
governing equations are
$$
\left\{
\begin{array}{l}
\mbox{div} \bm w=0, \\
(\bm w\nabla)\bm w+2\bm \omega\times \bm w +\nabla
p-\nu\mbox{div}(e(\bm w))
=-(\omega)^2\bm{\mathcal{R}} +\bm f, \\
\bm w|_{\Gamma_0}=\bm 0, \\
(-p\bm n+2\nu e(\bm w))|_{\Gamma_{in}}=\bm g_{in},\\
(-pn+2\nu e(\bm w))|_{\Gamma_{out}}=\bm g_{out},
\end{array}
\right.\eqno(3.7)
$$
where $\Gamma_0=\Im_+\cup \Im_-\cup \Gamma_t\cup \Gamma_b$.
 For the polytropic ideal gas and flow is stationary, system (3.1)
 can be wrote as the conservation form
$$
\left\{\begin{array}{l}
\mbox{div}(\rho \bm w)=0,   \\
\mbox{div}(\rho \bm w\otimes \bm w)+2\rho \bm \omega\times \bm
w+R\nabla(\rho T)=\mu \Delta \bm w+(\lambda+\mu)\nabla\mbox{div} \bm
w-\rho (\omega)^2\bm{\mathcal{R}},
\\
\mbox{div}\lbrack \rho(\frac{|\bm w|^2}{2}+c_vT+RT)\bm w]
=\kappa\Delta T+\lambda\mbox{div}(\bm w\mbox{div} \bm
w)+\mu\mbox{div}\lbrack \bm w\nabla \bm w]+\frac{\mu}{2}\Delta|\bm
w|^2,
\end{array}
\right.\eqno{(3.8)}
$$
while for isentropic ideas gases, it turns
$$
\left\{
\begin{array}{l}
\mbox{div}(\rho \bm w)=0,  \\
\mbox{div}(\rho \bm w\otimes \bm w)+2\rho \bm \omega\times \bm
w+\alpha\nabla(\rho^\gamma)=2\mu \mbox{div}(\bm e)
+\lambda\nabla\mbox{div} \bm w-\rho (\omega)^2\bm{\mathcal{R}},
\end{array}
\right.\eqno{(3.9)}
$$
where $\gamma>1 $ is the specific heat radio and $\alpha$ a positive
constant.

Furthermore, we give the expressions of the  power $I(\Im,w(\Im))$
done by the impeller and global dissipative energy $J(\Im,\bm
w(\Im))$, respectively
$$I(\Im,w(\Im))=\int\int_{\Im_-\cup \Im_+}\sigma\cdot \bm n\cdot
\bm e_\theta\omega r\mbox{d} \Im,\quad J(\Im,\bm
w(\Im))=\int\int\int_{\Omega_\varepsilon}\Phi(\bm w)\mbox{d}V,
\eqno(3.10)$$

Under the new coordinate system (2.2), from the discussion in
section \ref{sect:two}, we know there exists mapping between the
fixed domain $\Omega=D\times [-1,1]$ and the flow passage
$\Omega_\varepsilon$. In the subsequent paragraph, we suppose $D\in
R^2$ is  surrounded by four arcs
$\widehat{AB},\widehat{CD},\widehat{CB},\widehat{DA}$, and
$$\partial D=\gamma_0\cup\gamma_1,\quad\gamma_0=\widehat{AB}\cup\widehat{CD},\quad
\gamma_1=\widehat{CB}\cup\widehat{DA}.$$ There exist four positive
functions
$\gamma_0(z),\widetilde{\gamma}_0(z),\gamma_1(z),\widetilde{\gamma}_1(z)$
such that
$$\begin{array}{ll}
 r:=x^2=\gamma_0(x^1)=\gamma_0(z)\quad \mbox{on}\,\widehat{AB},\quad
   x^2=\widetilde{\gamma}_0(x^1)\quad \mbox{on}\,\widehat{CD}\\
 r:=x^2=\gamma_1(x^1)=\gamma_1(z)\,\quad \mbox{on}\,\widehat{DA},\quad
   x^2=\widetilde{\gamma}_1(x^1)\quad\mbox{on}\,\widehat{BC},\\
r_0\leq \gamma_0(z)\leq r_1\quad\mbox{on}\quad \widehat{AB},\quad
r_0\leq
\widetilde{\gamma}_0(z)\leq r_1\quad\mbox{on} \,\widehat{CD}, \\
 r_0\leq
\gamma_1(z)\leq r_1\quad\mbox{on}\, \widehat{DA},\quad r_0\leq
\widetilde{\gamma}_1(z)\leq r_1\quad\mbox{on}\,
\widehat{BC}.\end{array}\eqno{(3.11)}$$

%\begin{center}
%\includegraphics[width=50mm,height=30mm]{blade1.pdf}
%\includegraphics[width=40mm]{projectedarea1.pdf}

%Fig.3,Blade\hspace{2cm}Fig.4,Projection area $D$ of blade on the
%meridian.
%\end{center}

\begin{figure}
  \centering
  \subfloat[Blade]{\label{fig:gull}\includegraphics[width=0.3\textwidth]{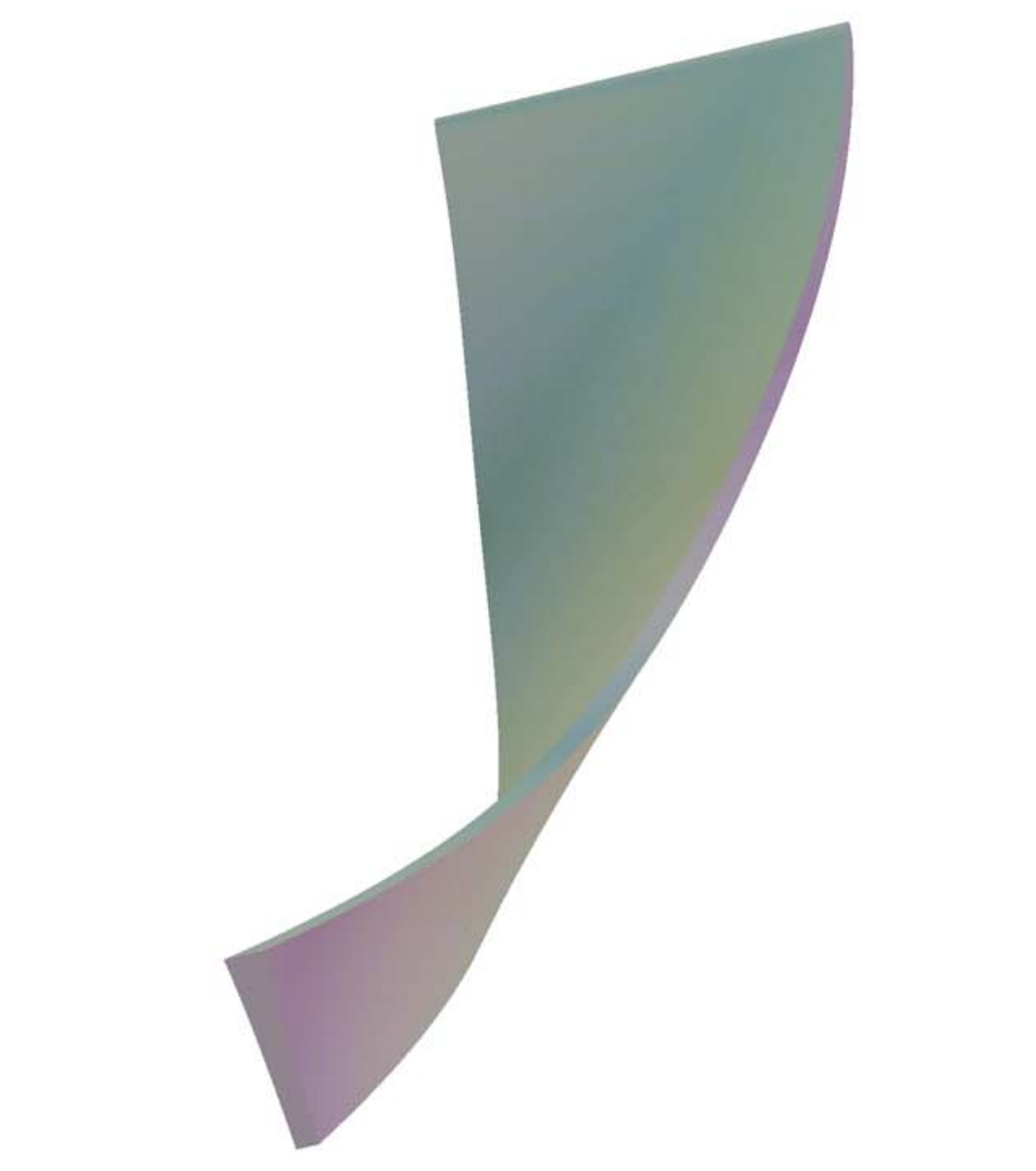}}\hspace{1cm}
  \subfloat[Projection area]{\label{fig:tiger}\includegraphics[width=0.3\textwidth,height=0.3\textwidth]{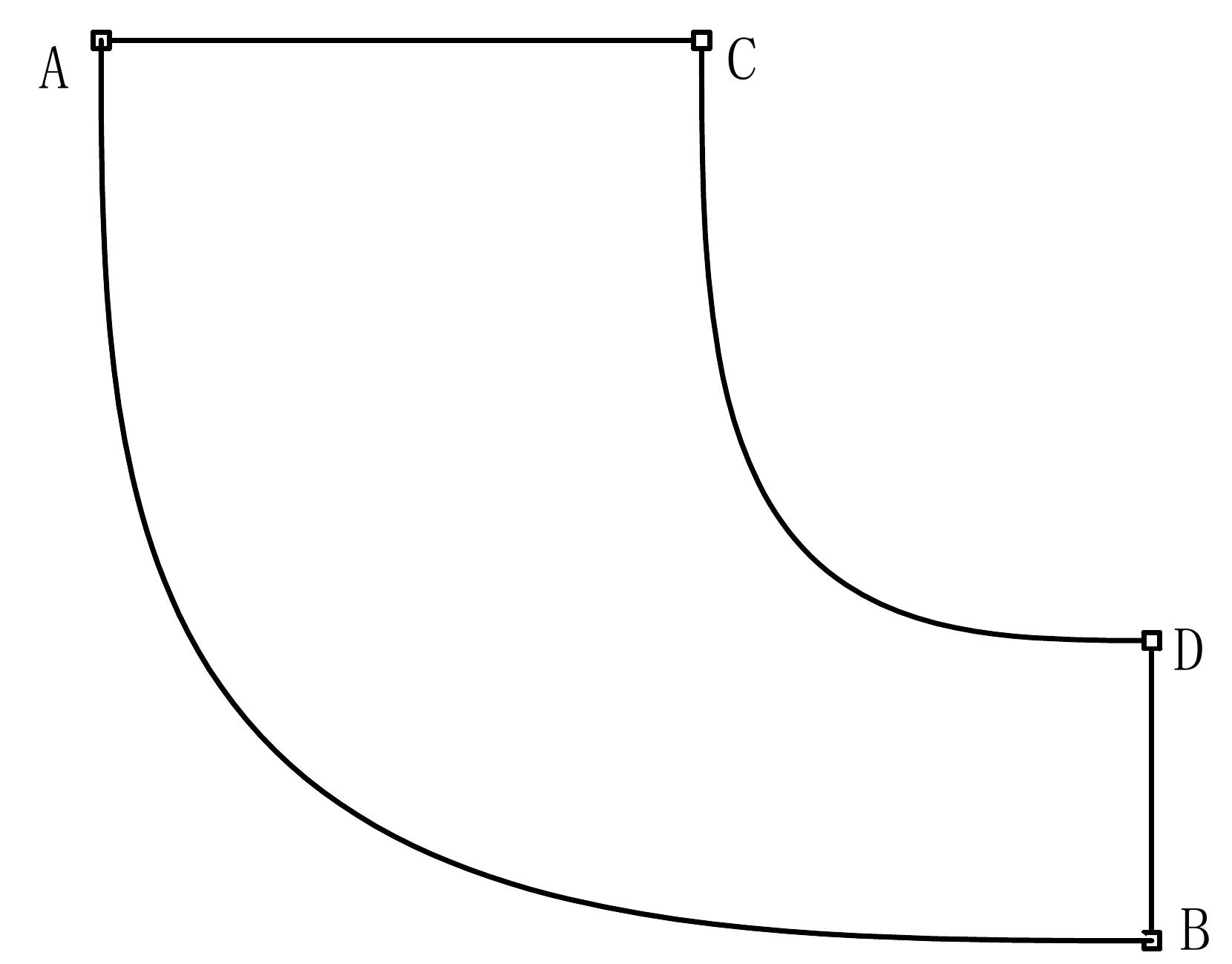}}
  \caption{Blade and Projection area $D$ of blade on the meridian.}
  \label{fig:ImpellerPassage}
\end{figure}

The corresponding three-dimensional region is expressed as
$$\begin{array}{ll}
\partial
\Omega=\widetilde{\Gamma}_0\cup\widetilde{\Gamma}_1,\quad
\widetilde{\Gamma}_1=\widetilde{\Gamma}_{out}\cup\widetilde{\Gamma}_{in},\quad
\widetilde{\Gamma}_0=\widetilde{\Gamma}_{b}\cup\widetilde{\Gamma}_{t}\cup\{\xi=1\}\cup\{\xi=-1\},\\
\widetilde{\Gamma}_{in}=\bm\Re(\Gamma_{in}),\widetilde{\Gamma}_{out}=\bm\Re(\Gamma_{out}),
\widetilde{\Gamma}_{b}=\bm\Re(\Gamma_{b}),\widetilde{\Gamma}_{t}=\bm\Re(\Gamma_{t}),\\
\end{array}\eqno{(3.12)}$$
$$\begin{array}{ll}
\partial D=\gamma_0\cup\gamma_1, \quad
\gamma_0=(D\cap\widetilde{\Gamma}_b)\cup(D\cap\widetilde{\Gamma}_t),
\quad\gamma_1=(D\cup\widetilde{\Gamma}_{out})\cup(D\cup\widetilde{\Gamma}_{in}),
\end{array}\eqno{(3.13)}$$
where $\bm\Re(\cdot)$ is defined by (2.1)

We introduce the following Sobolev space,
$$V(\Omega)=\{\bm v| \bm v\,\in
\,H^1(\Omega)^3,\,\bm v|_{\widetilde{\Gamma}_0}=0\},\quad
H^1_\Gamma(\Omega)=
\{q|,q\,\in\,H^1(\Omega),\,q|_{\widetilde{\Gamma}_0}=0\}.\eqno{(3.14)}$$
equipped with usual Sobolev norm $\|\cdot\|_{1,\Omega}$.
  Consider the variational formulation for Navier-Stokes problem (3.7) and (3.9)
$$
\left\{\begin{array}{l} \mbox{Find}\,\bm w\, \in\, V(\Omega),
p\in\,L^2(\Omega), \mbox{such that}
\\
a(\bm w,\bm v)+2(\bm \omega\times \bm w,\bm v)+b(\bm w,\bm w,\bm v)
 -(p, \mbox{div} \bm v)=<\bm F,\bm v>,\quad \forall\, \bm v\,\in\,V(\Omega),   \\
(q,\mbox{div} \bm w)=0,\quad \forall\,q\,\in\,L^2(\Omega),
\end{array}
\right.\eqno{(3.15)}$$
and
$$\left\{\begin{array}{l}
\mbox{Find}\,\bm w\,\in\,V(\Omega),\rho\in\,L^\gamma(\Omega),\mbox{such that}   \\
a(\bm w,\bm v)+2(\omega\times \bm w,\bm v)+b(\rho \bm w,\bm w,\bm v)
+(-p+\lambda\mbox{div} \bm w, \mbox{div} \bm v)=<\bm F,\bm v>,\quad
\forall\,\bm v\,\in\,V(\Omega), \\
(\nabla q,\rho \bm w))=<\rho \bm w\cdot\bm n,q>|_{\Gamma_1},\quad
\forall\,q\,\in\,H^1_\Gamma(\Omega),
\end{array}\right.\eqno{(3.16)}$$
where
$$
\begin{array}{ll}
<\bm F,\bm v>=<\bm f,\bm v> + <\widetilde{\bm g},\bm
v>_{\widetilde{\Gamma}_{1}},\quad <\widetilde{\bm g}
,\bm v>=<\bm g_{in},\bm v>|_{\widetilde{\Gamma}_{in}}+<\bm g_{out},\bm v>|_{\widetilde{\Gamma}_{out}}, &  \\
a(\bm w,\bm v)=\int_{\Omega}A^{ijkm}e_{ij}(\bm w)e_{km}(\bm v)\sqrt{g}\rm{d}x\rm{d}\xi, &  \\
b(\bm w,\bm w,\bm v)=\int_\Omega
g_{km}w^j\nabla_jw^kv^m\sqrt{g}\rm{d}x\rm{d}\xi, &
\end{array}
\eqno{(3.17)}$$

 In order to rewrite equations (3.7) and (3.9) in the new coordinate
system, we have to consider covariant derivatives of the vector
field. Therefore, we first give out the second kind of Christoffel
symbols in the new coordinate system in terms of $\Theta$(see
Appendix), i.e.,
$$
\left\{
\begin{array}{l}
\Gamma^\alpha_{\beta\gamma}=-r\delta_{2\alpha}\Theta_\beta\Theta_\gamma,
\quad\Gamma^\alpha_{3\beta}= -\varepsilon r\delta_{2\alpha}\Theta_\beta, \\
\Gamma^3_{\alpha\beta}=\varepsilon^{-1}r^{-1}(\delta_{2\alpha}\delta^
\lambda_\beta+ \delta_{2\beta}\delta^\lambda_\alpha
)\Theta_\lambda+\varepsilon^{-1}\Theta_{\alpha\beta}+\varepsilon^{-1}r
\Theta_2\Theta_\alpha\Theta_\beta, \\
\Gamma^3_{3\alpha}=\Gamma^3_{\alpha
3}=r^{-1}\delta_{2\alpha}+r\Theta_2\Theta_\alpha\quad
\Gamma^\alpha_{33}=-\varepsilon^2
r\delta_{2\alpha},\quad\Gamma^3_{33}=\varepsilon r\Theta_2,
\end{array}
\right.\eqno{(3.18)}$$
 then covariant derivatives of the velocity
field is $\nabla_iw^j=\frac{\partial w^j}{\partial
x^i}+\Gamma^j_{ik}w^k$, which can be specific expressed in the next
lemma(see Appendix),

\begin{lemma} Under the new coordinate system  $(x^1,x^2,\xi)$ defined by (2.4),
the covariant derivatives of the velocity field can  be expressed as
$$\left\{
\begin{array}{l}
\nabla_\alpha w^\beta=\frac{\partial w^\beta}{\partial x^\alpha}
-r\delta^\beta_{2}\Theta_\alpha\Pi(\bm w,\Theta), \\
\nabla_\alpha w^3=\frac{\partial w^3}{\partial x^\alpha}%
+\varepsilon^{-1}(x^2)^{-1}w^2\Theta_\alpha+
\varepsilon^{-1}w^\beta\Theta_{\alpha\beta}+(\varepsilon
x^2)^{-1}a_{2\alpha}\Pi(\bm w,\Theta), \\
\nabla_3 w^\alpha=\frac{\partial
w^\alpha}{\partial\xi}-x^2\varepsilon \delta_{2\alpha}\Pi(\bm
w,\Theta), \quad \nabla_3 w^3=\frac{\partial w^3}{\partial
\xi}+\frac{w^2}{x^2}
+x^2\Theta_2\Pi(\bm w,\Theta), \\
\mbox{div} \bm w=\frac{\partial w^\alpha}{\partial
x^\alpha}+\frac{w^2}{x^2}+ \frac{\partial w^3}{\partial\xi},\quad
\Pi(\bm w,\Theta)=\varepsilon w^3+w^\beta\Theta_\beta.
\end{array}
\right.\eqno(3.19)$$
\end{lemma}

A simple calculation show that the strain  tensor can be rewrite in
the splitting form
$$e_{ij}(\bm w)=\phi_{ij}(\bm w)+\psi_{ij}(\bm w,\Theta),\eqno{(3.20)}$$
where the first terms is independent of  $\Theta$, that is
$$\phi_{\alpha\beta}(\bm w)
=\displaystyle\frac12(\displaystyle\frac{\partial w^\alpha}{\partial
x^\beta} +\displaystyle\frac{\partial w^\beta}{\partial
x^\alpha}),\quad \phi_{3\alpha}(\bm w)
=\displaystyle\frac12(\displaystyle\frac{\partial
w^\alpha}{\partial\xi}+\varepsilon^{2}r^2\displaystyle\frac{\partial
w^3}{
\partial x^\alpha}),\quad
\phi_{33}(\bm w)  =\varepsilon^2r^2(\displaystyle\frac{\partial
w^3}{
\partial\xi}+\displaystyle\frac{w^2}{r}).\\
\eqno{(3.21)}$$
 While the second terms contains $\Theta$, $ \psi_{ij}(\bm
w,\Theta)=\psi_{ij}^\lambda(\bm
w)\Theta_\lambda+\psi_{ij}^{\lambda\sigma}(\bm
w)\Theta_\lambda\Theta_\sigma +e_{ij}^*(\bm w,\Theta)$, where
$$\left\{\begin{array}{ll}
\psi^\lambda_{\alpha\beta}(\bm w)&=\displaystyle \frac12\varepsilon
r^2(\displaystyle\frac{\partial w^3}{\partial x^\alpha}
\delta^\lambda_\beta+\displaystyle\frac{\partial w^3} {\partial
x^\beta}
\delta^\lambda_\alpha), \\
 \psi^\lambda_{3\alpha}(\bm w)&=\displaystyle\frac12
\varepsilon r^2(\displaystyle\frac{\partial w^\lambda}{\partial
x^\alpha} +\delta^\lambda_\alpha(\displaystyle\frac{\partial
w^3}{\partial\xi} + \displaystyle\frac{2}{r}w^2)), \quad
\psi^\lambda_{33}(\bm w)=\varepsilon r^2\frac{\partial
w^\lambda}{\partial\xi},\\
\psi^{\lambda\sigma}_{\alpha\beta}(\bm w) & =\displaystyle\frac12r^2
( \displaystyle\frac{\partial w^\lambda}{\partial
x^\alpha}\delta_{\beta
\sigma} +\displaystyle\frac{\partial w^\lambda}{\partial x^\beta}%
\delta_{\sigma\alpha}+\frac2{r}w^2\delta_{\alpha\lambda}\delta_{\sigma
\beta}) , \\
\psi^{\lambda\sigma}_{3\alpha}(\bm w) &
=\frac12r^2\displaystyle\frac{
\partial w^\lambda}{\partial \xi}\delta_{\alpha\sigma}, \quad
\psi^{\lambda\sigma}_{33}(\bm w)=0.
\end{array}
\right.\eqno{(3.22)}
$$
and
$$\begin{array}{ll}
e^*_{\alpha\beta}(\bm w,\Theta) &
=\frac12r^2w^\sigma\partial_\sigma(\Theta_\alpha\Theta_\beta), \quad
e^*_{3\alpha}(\bm w)=\frac12\varepsilon
r^2w^\sigma\Theta_{\sigma\alpha},\quad e^*_{33}(\bm w)=0.
\end{array}\eqno{(3.23)}$$
 The proof is omitted.

The following notations are frequently used in the later,
$$\widetilde{\Delta}=\frac{\partial^2}{\partial (x^1)^2}+\frac{\partial^2}{\partial
(x^2)^2},\quad
\widetilde{\nabla}_\alpha=\partial_\alpha=\frac{\partial}{\partial
x^\alpha},\quad \widetilde{\div} w=\frac{\partial w^\alpha}{\partial
x^\alpha}.$$

For the sake of simplicity, we just consider incompressible flow.
Taking  into account $(3.18), (3.19)$, in the new coordinate system
the Navier-Stokes equations can  be written in the form,

\begin{theorem}
 Suppose that the blade surface is smooth enough, that is $\Theta$ is smooth
 enough, for example, $\Theta\,\in\,C^3(D)$, then the rotating Navier-Stokes
equations in the new  coordinate are given by
$$\left\{\begin{array}{ll}
\frac{\partial w^\alpha}{\partial x^\alpha}+\frac{\partial
w^3}{\partial\xi}+\frac{w^2}{r}=\frac1r\frac{\partial
(rw^\alpha)}{\partial x^\alpha}+\frac{\partial
w^3}{\partial\xi}=\widetilde{\div}_2w+\frac{\partial
w^3}{\partial\xi}=0,\\
{\cal N}^k(\bm w,p,\Theta):=-\nu\widetilde{\Delta}
w^k-\nu(r\varepsilon)^{-2}a\frac{\partial^2w^k}{\partial\xi^2}-\nu
P^{k3}_{j}(\Theta)\frac{\partial
w^j}{\partial\xi}-2\nu\varepsilon^{-1}\Theta_\beta\frac{\partial^2w^k}{\partial\xi\partial
x^\beta}\\
\qquad-\nu P^{k\beta}_j(\Theta)\frac{\partial w^j}{\partial
x^\beta}-\nu q^k_j(\Theta)w^j+g^{k\beta}(\Theta)\nabla_\beta p+
g^{k3}(\Theta)\partial_\xi p\\
\qquad+C^k(\bm w,\bm \omega)+N^k(\bm w,\bm w) =f^k,\quad \forall
k=1,2,3,
\end{array}\right.\eqno{(3.24)}$$
where $\bm C(\bm w,\bm \omega)$ is Coriolis forces defined in
(A.I.4), and
$$\left\{\begin{array}{ll}
  N^\alpha(\bm w,\bm w)&=w^\beta\frac{\partial
  w^\alpha}{\partial
  x^\beta}+w^3\frac{\partial w^\alpha}{\partial\xi}-r\delta_{2\alpha}\Pi(\bm w,\Theta)\Pi(\bm w,\Theta)
 =\frac{\partial( w^3w^\alpha)}{\partial\xi}
 +\partial_\beta(w^\alpha w^\beta)\\
 & +r^{-1}w^2w^\alpha-r\delta_{2\alpha}\Pi(\bm w,\Theta)\Pi(\bm w,\Theta),\\
N^3(\bm w,\bm w)&=\partial_\xi(w^3w^3)+\partial_\beta(w^\beta
  w^3)+\varepsilon^{-1}w^\beta
  w^\lambda\Theta_{\beta\lambda}\\
  &+(r\varepsilon)^{-1}\Pi(\bm w,\Theta)(2w^2
  +r^2\Theta_2\Pi(\bm w,\Theta)),\\
  N^k(\bm w,\bm w)&=\frac{\partial( w^3w^k)}{\partial\xi}
 +\partial_\beta(w^k w^\beta)+\pi^k_{ij}w^iw^j=\frac{\partial(
 w^3w^k)}{\partial\xi}+B^k(\bm w,\bm w),\\
 B^k(\bm w,\bm w)&=\partial_\beta(w^k w^\beta)+\pi^k_{ij}w^iw^j,
\end{array}\right.\eqno{(3.25)}$$
$$\left\{\begin{array}{ll}
P^{\lambda\beta}_\alpha(\Theta)&=\frac1r\delta_{\beta2},\quad
P^{\lambda\beta}_3(\Theta)=-2r\varepsilon\delta_{2\lambda}\Theta_\beta,\\
P^{3\beta}_\alpha(\Theta)&=(r\varepsilon)^{-1}(\delta_{2\beta}\Theta_\alpha+2r\Theta_{\alpha\beta}),\quad
P^{3\beta}_3=\frac2r\delta_{\beta2},\\
P^{\alpha3}_{\lambda}(\Theta)&=-[(r\varepsilon)^{-1}(\delta_{\alpha\lambda}\Theta_2+2\delta_{2\alpha}\Theta_\lambda)
+\varepsilon^{-1}\delta_{\alpha\lambda}\widetilde{\Delta}\Theta],\quad
P^{\alpha3}_{3}=-2r^{-1}\delta_{2\alpha},\\
P^{33}_{\sigma}(\Theta)&=2\varepsilon^{-2}(r^{-3}\delta_{2\sigma}-\Theta_\beta\Theta_{\beta\sigma}),\\
P^{33}_{3}(\Theta)&=\varepsilon^{-1}(r\Theta_2|\widetilde{\nabla}\Theta|^2-\widetilde{\Delta}\Theta),
\end{array}\right.\eqno{(3.26)}$$
$$\left\{\begin{array}{ll}
q^\alpha_\sigma(\Theta)&=2\delta_{2\alpha}[\delta_{2\sigma}|\widetilde{\nabla}
\Theta|^2-a\Theta_2\Theta_\sigma+r\Theta_\sigma\widetilde{\Delta}\Theta-r\Theta_\lambda\Theta_{\lambda\sigma}]
-r^{-2}\delta_{2\sigma}\delta_{2\alpha},\\
q^\alpha_3(\Theta)&=\delta_{2\alpha}(r\widetilde{\Delta}\Theta-2a\Theta_2)\varepsilon,\\
q^3_\sigma(\Theta)&=(r\varepsilon)^{-1}[r^{-1}(1+a(a_{22}-1))\Theta_\sigma+2\Theta_{2\sigma}]
+\varepsilon^{-1}\partial_\sigma\widetilde{\Delta}\Theta,\\
q^3_3(\Theta)&=a\Theta_2\Theta_2,
\end{array}\right.\eqno{(3.27)}$$
$$\left\{\begin{array}{ll}
\Theta_\alpha=\frac{\partial\Theta}{\partial x^\alpha},\quad
\Theta_{\alpha\beta}=\frac{\partial^2\Theta}{\partial
x^\alpha\partial x^\beta},\quad \Pi(\bm w,\Theta)=\varepsilon
w^3+w^\lambda\Theta_\lambda,\\
\widetilde{\Delta}\Theta=\Theta_{\alpha\alpha}=\Theta_{11}+\Theta_{22},\quad
|\widetilde{\nabla}\Theta|^2=\Theta_1^2+\Theta_2^2,\\
\end{array}\right.\eqno{(3.28)}$$
and
$$\left\{\begin{array}{ll}
\pi^\alpha_{\lambda\sigma}(\Theta)=-r\delta_{\alpha2}\Theta_\lambda\Theta_\sigma
+r^{-1}\delta_{\lambda2}\delta_{\alpha\sigma},\\
\pi^\alpha_{\lambda3}(\Theta)=\pi^\alpha_{3\lambda}(\Theta)=-r\varepsilon\delta_{2\alpha}\Theta_\lambda,\quad
\pi^\alpha_{33}=-r\varepsilon^2 \delta_{2\alpha},\\
\pi^3_{\lambda\sigma}(\Theta)=\varepsilon^{-1}\Theta_{\lambda\sigma}+(r\varepsilon)^{-1}\Theta_\lambda(\delta_{2\sigma}
+a_{2\sigma}),\\
\pi^3_{3\lambda}(\Theta)=\pi^3_{\lambda3}(\Theta)=r^{-1}a_{2\lambda}+r^{-1}\delta_{2\lambda},\quad
\pi^3_{33}=r\varepsilon\Theta_2.
\end{array}\right.\eqno{(3.29)}$$
\end{theorem}

\begin{proof}The Proof see Appendix.\end{proof}

Let introduce the inner product in the Sobolev space $V(\Omega)$ or
$V(D)$
$$\left\{\begin{array}{ll}
(\bm w,\bm
v)=\int\limits_\Omega[g_{ij}w^iv^j\sqrt{g}\rm{d}x\rm{d}\xi=\int\limits_\Omega[
a_{\alpha\beta}w^\alpha v^\beta+r^2\varepsilon\Theta_\beta
(w^3v^\beta+w^\beta v^3)+r^2\varepsilon^2w^3v^3]r\varepsilon \rm{d}x
\rm{d}\xi,\\
(\bm w,\bm v)_D=\int\limits_D[ a_{\alpha\beta}w^\alpha
v^\beta+r^2\varepsilon\Theta_\beta (w^3v^\beta+w^\beta
v^3)+r^2\varepsilon^2w^3v^3]r\varepsilon \rm{d}x,\\
\end{array}\right.\eqno{(3.30)}$$
The subscript ``D'' will  be omitted if there is no
misunderstanding.

 Next we consider the variational formulation for (3.24) in the new
 coordinate system. Taking  into account (2.5), let set
  $$\begin{array}{ll}
 A(\bm w,\bm v)&=(g_{ij}{\cal N}^i(\bm w,p,\Theta),v^j)=(g_{\alpha\beta}{\cal N}^\alpha+\varepsilon
 r^2\Theta_\beta {\cal N}^3,v^\beta)+(\varepsilon r^2\Theta_\beta
 {\cal N}^\beta+e^2r^2{\cal N}^3,v^3)\\
 &=(A_m(\bm w,\Theta),v^m),
\end{array}\eqno{(3.31)}$$
By using index reduction, lift and descent of tensor, we get
$$\begin{array}{ll}
g_{mk}g^{k\beta}=\delta^\beta_m,\quad g_{mk}g^{k3}=\delta^3_m,\quad
C_m(\bm w,\bm \omega)=g_{mk}C^k(\bm w,\bm \omega),\\
P^l_{mj}(\Theta)=g_{mk}P^{kl}_j(\Theta),\quad
q_{mj}(\Theta)=g_{mk}q^k_j(\Theta),\quad B_m(\bm w,\bm w)=g_{mk}B^k(\bm w,\bm w).\\
\end{array}\eqno{(3.32)}$$
Let adopt the notations
$$\left\{\begin{array}{ll}
E_m(\bm w)&
=g_{mk}\widetilde{\Delta}w^k=\partial_\lambda(g_{mk}\partial_\lambda
w^k)-\partial_\lambda g_{mk}
\partial_\lambda w^k,\\
N_m(\bm w,\bm w)&=g_{m\alpha}N^\alpha(\bm w,\bm w)+g_{m3}N^3(\bm w,\bm w)\\
&=g_{mk}(\frac{\partial( w^3w^k)}{\partial\xi}
 +\partial_\beta(w^k w^\beta)+\pi^k_{ij}w^iw^j)=g_{mk}\frac{\partial(
w^3w^k)}{\partial\xi}+B_m(\bm w,\bm w),\\
B_m(\bm w,\bm w)&=g_{mk}(\partial_\beta(w^k
w^\beta)+\pi^k_{ij}w^iw^j),
\end{array}\right.\eqno{(3.33)}$$
Then
$$\left\{\begin{array}{ll}
A_m(\bm w,\Theta) &=-\nu E_m(\bm
w)-\nu(r\varepsilon)^{-2}ag_{mk}\frac{\partial^2w^k}{\partial\xi^2}-\nu
P^{3}_{mj}(\Theta)\frac{\partial
w^j}{\partial\xi}-2\nu\varepsilon^{-1}\Theta_\beta
g_{mk}\frac{\partial^2w^k}{\partial\xi\partial
x^\beta}\\
&-\nu P^{\beta}_{mj}(\Theta)\frac{\partial w^j}{\partial
x^\beta}-\nu q_{mj}(\Theta)w^j+\delta^\beta_m\nabla_\beta p+
\delta^3_m\partial_\xi p\\
&+C_m(\bm w,\omega)+B_m(\bm w,\bm w)+g_{mk}\partial_\xi(w^3w^k)
=f_m,
\end{array}\right.\eqno{(3.34)}$$

\begin{remark} Obviously we have
$$\begin{array}{rcl}
(\bm C(\bm w,\bm \omega),\bm w)&=&C_\beta(\bm w,\bm
\omega)w^\beta+C_3(\bm w,\bm \omega)w^3
=2r\omega(w^2\Theta_\beta-\delta_{2\beta}\Pi(w,\Theta))w^\beta
+2r\varepsilon \omega
w^2w^3\\
&=&2r\omega(w^2w^\beta\Theta_\beta-w^2(\varepsilon
w^3+w^\lambda\Theta_\lambda))+2r\varepsilon\omega w^2w^3=0,
\end{array}\eqno{(3.35)}$$ which is coincide with $2\bm \omega\times \bm w\cdot
\bm w=0$.
\end{remark}

Since
$$\begin{array}{ll}
-\nu g_{mk}\widetilde{\Delta} w^kv^m=-\nu
g_{mk}v^m\partial_\lambda\partial_\lambda w^k\\
\quad\quad =-\partial_\lambda(\nu g_{mk}v^m\partial_\lambda w^k)+\nu
g_{mk}\partial_\lambda w^k\partial_\lambda v^m+\nu\partial_\lambda
g_{mk}\partial_\lambda w^kv^m,\\
\int\limits_{\xi=-1}^{\xi=1}\int\limits_D[-\nu
g_{mk}\widetilde{\Delta} w^kv^m+\delta^\beta_m\partial_\beta
pv^m]\rm{d}x\rm{d}\xi\\
\quad\quad=\int\limits_{\xi=-1}^{\xi=1}\int\limits_D[-\partial_\lambda(\nu
g_{mk}v^m\partial_\lambda w^k)+\partial_\beta(v^\beta p)+\nu
g_{mk}\partial_\lambda w^k\partial_\lambda v^m+\nu\partial_\lambda
g_{mk}\partial_\lambda
w^kv^m-p\partial_\beta v^\beta]\rm{d}x\rm{d}\xi\\
\qquad=\int\limits^{\xi=1}_{\xi=-1}\int\limits_{\partial D}[-\nu
g_{mk}v^m\partial_\lambda w^kn_\lambda+p n_\beta
v^\beta]\rm{d}s\rm{d}\xi\\
\qquad\quad+\int\limits_{\xi=-1}^{\xi=1}\int\limits_D[\nu
g_{mk}\partial_\lambda w^k\partial_\lambda v^m+\nu\partial_\lambda
g_{mk}\partial_\lambda w^kv^m-p\partial_\beta
v^\beta]\rm{d}x\rm{d}\xi,\\
\delta^\beta_m\partial_\beta pv^m=\partial_\beta(v^\beta
p)-p\partial_\beta v^\beta,
\end{array}$$ where $\bm n$ is unite normal vector to $\partial D$,
then
$$\begin{array}{l}
\int\limits_{\xi=-1}^{\xi=1}\int\limits_D[-\nu
g_{mk}\widetilde{\Delta} w^kv^m+\delta^\beta_m\partial_\beta
pv^m]\rm{d}x\rm{d}\xi=\int\limits^{\xi=1}_{\xi=-1}\int\limits_{\partial D}[\sigma_{\bm nm}(\bm w,p) v^m]\rm{d}s\rm{d}\xi\\
\quad\quad\quad+\int\limits_{\xi=-1}^{\xi=1}\int\limits_D[\nu
g_{mk}\partial_\lambda w^k\partial_\lambda v^m+\nu\partial_\lambda
g_{mk}\partial_\lambda w^kv^m-p\partial_\beta
v^\beta]\rm{d}x\rm{d}\xi,
\end{array}\eqno{(3.36)}$$
where $\sigma_{\bm nm}(\bm w,p)=(-\nu g_{mk}\partial_\beta w^k+p
\delta_{m\beta})n_\beta$.

 Recall the bilinear form and trilinear form on  $V(\Omega)$
$$\left\{\begin{array}{ll}
a_0(\bm w,\bm v)&=\int\limits_\Omega(\nu g_{mk}\partial_\lambda
w^k\partial_\lambda v^m)\rm{d}x\rm{d}\xi,\\
 b(\bm w,\bm  u,\bm v)&=\int\limits_\Omega
B_m(\bm w,\bm w)v^m\sqrt{g}\rm{d}x\rm{d}\xi=\int\limits_\Omega
g_{mk}[\partial_\lambda(w^ku^\lambda)+\pi^k_{ij}w^iu^j]v^m\rm{d}x\rm{d}\xi.
\end{array}\right.\eqno{(3.37)}$$
If the following boundary conditions are satisfied
$$\left\{\begin{array}{ll}
\bm w|_{\Gamma_s}=0,&\mbox{on}\Gamma_s=\Gamma_t\cup\Gamma_b\cup\Gamma|_{\xi=\pm1},\\
\bm\sigma_{\bm n}(\bm w,p)|_{\Gamma_{1}}=\bm h,&\mbox{On}
\Gamma_1=\Gamma_{in}\cup\Gamma_{out},
\end{array}\right.\eqno{(3.38)}$$
then the variational formulation for (3.24) is given by
$$\left\{\begin{array}{ll}
\mbox{Find}\bm w\in V(\Omega),\
p\in\,L^2_0(\Omega),\ \mbox{such that}\\
a_0(\bm w,\bm v)+(\bm C(\bm w,\bm \omega),\bm v) +b(\bm w,\bm w,\bm
v)-(p,\partial_\alpha
v^\alpha)+(\frac{\partial p}{\partial\xi},v^3)-\nu((r\varepsilon)^{-2}ag_{mk}\frac{\partial^2w^k}{\partial\xi^2},v^m)\\
\qquad\quad+(\partial_\xi\Phi_m(\bm w,\Theta),\bm
v^m)+\nu(\widetilde{P}^\beta_{mj}(\Theta))\partial_\beta w^j,v^m)
-\nu(q_{mj}(\Theta)w^j,v^m)\\
\qquad=(\bm f,\bm v)+<\bm h,\bm v>:=<\bm F,\bm v>,\quad \forall\,\bm v\in V(\Omega)\\
(\frac{\partial w^\alpha}{\partial
x^\alpha}+\frac{w^2}{r}+\frac{\partial
w^3}{\partial\xi},q)=0,\quad\forall\, q\in\,L^2(\Omega),
\end{array}\right.\eqno{(3.39)}$$
where
$$
\left\{\begin{array}{ll}
 \Phi_m(\bm w,\Theta)=-\nu( P^{3}_{mj}(\Theta)
w^j+2\varepsilon^{-1}\Theta_\beta g_{mk}\frac{\partial w^k}{\partial
x^\beta}) +g_{mk}(w^3w^k),\\
\widetilde{P}^\beta_{mj}(\Theta)=(\partial_\beta
g_{mj}-P^\beta_{mj}(\Theta)),\\
\end{array}\right.\eqno{(3.40)}$$

\section{The equations For the average velocity along the Rotating Direction}

We define the average along the rotating direction for the function
$\varphi(x^1,x^2,\xi)$ in the coordinate $(x,\xi)$ in the domain
$\Omega=D\times[-1,1]\in\,R^3$
$$M(\varphi)=\frac12\int\limits^1_{-1}\varphi(x,\xi)\rm{d}\xi:=\overline{\varphi}.\quad\forall\,\varphi(x,\xi)\in\,L^2(\Omega)\eqno{(4.1)}$$

It is well known that the divergence of a vector $\bm w$ can be
written as under the coordinate $(x^1,x^2,\xi)$
$$\div\bm w=\frac{\partial w^\alpha}{\partial
x^\alpha}+\frac{w^2}{r}+\frac{\partial
w^3}{\partial\xi}=\frac1r\frac{\partial( rw^\alpha)}{\partial
x^\alpha}+\frac{\partial w^3}{\partial\xi},$$
 From this it yields
$$M(\div\bm w)=\frac12\int\limits^1_{-1}[\frac1r\frac{\partial( rw^\alpha)}{\partial
x^\alpha}+\frac{\partial w^3}{\partial\xi}]\rm{d}\xi,$$ Since
boundary conditions,
$$\begin{array}{ll}
\frac12\int\limits^1_{-1}\frac{\partial
w^3}{\partial\xi}\rm{d}\xi=\frac 12( w^3|_{\xi=1}-w^3|_{\xi=-1})=0,
\forall\, w\in\,V(\Omega),\\
\int\limits^1_{-1}\frac1r\frac{\partial( rw^\alpha)}{\partial
x^\alpha}\rm{d}\xi=\frac1r\frac{\partial}{\partial
x^\alpha}(r\overline{w}^\alpha)=\frac{\partial
\overline{w}^\alpha}{\partial
x^\alpha}+\frac{\overline{w}^2}{r}:=\widetilde{\div}_2(\overline{w}),
\end{array}$$ where
$$\widetilde{\div}_2(\bm w)=\frac{\partial w^\alpha}{\partial
x^\alpha}+\frac{w^2}{r}=\frac1r\frac{\partial}{\partial
x^\alpha}(rw^\alpha).\eqno{(4.2)}$$ Therefore we assert
$$M(\div\bm w)=\widetilde{\div}_2(\overline{\bm w}).\eqno{(4.3)}$$ and
the incompressibility  becomes
$$\widetilde{\div}_2(\overline{x^1,x^2w})=0,\eqno{(4.4)}$$
Taking  into account the boundary conditions,
$$\bm w|_{\Im_+\cup\Im_-\cup\gamma_t\cup\gamma_b}=\bm 0, \frac{\partial
w^3}{\partial\xi}|_{\xi=\pm1}=-\widetilde{\div}_2\bm
w|_{\xi=\pm1}=0.\eqno{(4.5)}$$  we get
$$M(\partial_\xi\Phi(\bm w,\Theta))=0,\eqno{(4.6)}$$
Let make notation $[\bm w]=\bm w|_{\xi=1}-\bm w|_{\xi=-1},
\overline{\bm w}=M\bm w$. Then average equations of Navier-Stokes
equations  are given by
$$\left\{\begin{array}{ll}
\widetilde{\div}_2\overline{\bm w}=0,\\
-\nu E_m(\overline{\bm w}) -\nu P^{\beta}_{mj}(\Theta)\frac{\partial
\overline{\bm w}^j}{\partial x^\beta}-\nu
q_{mj}(\Theta)\overline{\bm w}^j+\delta^\beta_m\nabla_\beta
\overline{p}\\
\qquad+C_m(\overline{\bm w},\bm \omega)+M(B_m(\bm w,\bm w)) =M(\bm
f)_m+\nu(r\varepsilon)^{-2}ag_{m\alpha}[\frac{\partial
w^\alpha}{\partial\xi}]-\delta^3_m [p],
\end{array}\right.\eqno{(4.7)}$$

Let $\bm w\bm v=a_{\lambda\sigma}w^\lambda v^\sigma+w^3v^3,
\widetilde{\bm w}=\bm w-\overline{\bm w}$, then, it is clear that
$$M(w^\lambda-\overline{w}^\lambda)=M\widetilde{\bm w}=0,\quad M(\widetilde{\bm w}\overline{\bm w})=0,\eqno{(4.8)}$$
Hence
$$\left\{\begin{array}{ll}
M(w^\lambda w^\sigma)=\overline{w}^\lambda \overline{w}^\sigma+M((\widetilde{w}^\lambda)w^\sigma),\\
M(w^\lambda \frac{\partial w^k}{\partial
x^\lambda})=\overline{w}^\lambda\frac{\partial\overline{w}^k}{\partial
x^\lambda}\\
\qquad+M((\widetilde{w}^\lambda)\frac{\partial
w^k}{\partial x^\lambda}),\\
\end{array}\right.\eqno{(4.9)}$$
thus we  conclude
$$\begin{array}{ll}
M(B_{m}(\bm w,\bm w))=M(g_{mk}(\frac{\partial w^\lambda
w^k}{\partial
x^\lambda}+\pi^k_{ij}(w^iw^j)),\\
M(B_{m}(\bm w,\bm w))=B_{m}(\overline{\bm w},\overline{\bm
w})+g_{mk}M(\partial_\lambda((\widetilde{w}^\lambda)
w^k)+\pi^k_{ij}(\widetilde{w}^i)w^j),\\
\end{array}\eqno{(4.10)}$$
Finally, by virtue of
$$M(\widetilde{\bm w}\bm w)=M(\widetilde{\bm
w}(\widetilde{\bm w}+\overline{\bm w}))=M(\widetilde{\bm
w}\widetilde{\bm w}),$$
 it yields the reduced Navier-Stokes
equations
$$\left\{\begin{array}{ll}
\widetilde{\div}_2\overline{\bm w}=0,\\
-\nu E_m(\overline{\bm w}) -\nu P^{\beta}_{mj}(\Theta)\frac{\partial
\overline{w}^j}{\partial x^\beta}-\nu q_{mj}(\Theta)\overline{\bm
w}^j+\delta^\beta_m\nabla_\beta
\overline{p}\\
\qquad+C_m(\overline{\bm w},\omega)+(B_m(\overline{\bm
w},\overline{\bm w})) =M(\bm
f)_m+\nu(r\varepsilon)^{-2}ag_{m\alpha}[\frac{\partial
w^\alpha}{\partial\xi}]-\delta^3_m
[p]\\
\qquad-g_{mk}M(\partial_\lambda((\widetilde{w}^\lambda)
\widetilde{w}^k)+\pi^k_{ij}(\widetilde{w}^i)\widetilde{w}^j),
\end{array}\right.\eqno{(4.11)}$$

We define the Sobolev spaces
$$\begin{array}{ll}
V(\Omega)=\{\textbf{u}\in \bm h^1(\Omega),\quad
\textbf{u}=0,\,\mbox{on} \Gamma_t\Gamma_b\cup\Gamma_+\Gamma_-,\},\\
V(D)=\{\textbf{u}\in\,\bm h^1, \textbf{u}=0,\mbox{on} \gamma_0
,\mbox{see}(3.13)\}, \end{array}$$ By a similar manner as (3.39) the
variational formulation for the reduced Navier-Stokes equations
(4.11) is given as
$$\left\{\begin{array}{ll}
\mbox{Find}\quad \overline{\bm w}\in V(D),\quad
\overline{p}\in\,L^2(D)\quad\mbox{such that}\, \\
a_0(\overline{\bm w},\bm v)+(\bm C(\overline{\bm w},\bm \omega),\bm
v) -\nu(\widetilde{P}^\beta_{mj}(\Theta)\partial_\beta
\overline{w}^j+q_{mj}\overline{w}^j,v^m) +b(\overline{\bm
w},\overline{\bm w},\bm v) -(\overline{p},\partial_\alpha
v^\alpha)\\
\qquad=(-g_{m k}M(\partial_\lambda
 (\widetilde{w}\widetilde{w}^k)+\pi_{k,ij}(\widetilde{w}^i\widetilde{w}^j)),v^m)\\
\hspace{1cm}+(\nu(r\varepsilon)^{-2}ag_{m\alpha}[\frac{\partial
w^\alpha}{\partial\xi}]-\delta_{3m}[p],v^m) +(Mf_m,v^m),\ \quad
\forall\,\bm v\in V(D)\\
(\widetilde{\div}_2\overline{w},q)=0, \quad \forall\,q\in\,L^2(D),
\end{array}\right.\eqno{(4.12)}$$
where
$$\left\{\begin{array}{ll}
a_0(\bm u,\bm v)=(\nu g_{mk}\partial_\lambda u^k, \partial_\lambda
v^m)
=\int\limits_D\nu g_{mk}\partial_\lambda u^k\partial_\lambda v^m\rm{d}x,\\
b(\bm u,\bm w,\bm v)=(g_{mk}(\partial_\lambda(u^\lambda
w^k)+\pi^k_{ij}(\Theta)u^iw^j),v^m),
\end{array}\right.\eqno{(4.13)}$$

\section{ The Equations for the G\^{a}teaux
   Derivative of the solutions of NSE with Respect to the Shape of Boundary}

In this section we consider the derivatives of the solution of NSE
with respective to two dimensional manifold $\Im$ which is a portion
of the solid boundary of the flow in the channel in turbo-machinery.

\begin{theorem}Assume that Surface $\Im$ is smooth enough, for example,
$\Theta\in C^3(D)$, then  there exists a $G\widehat{a}teaux$
derivatives ($\widehat{w}:=\frac{{\cal D }w}{{\cal
D}\Theta},\widehat{p}:=\frac{{\cal D }p}{{\cal D}\Theta}$) of the
solutions $(w,p)$ of Navier-Stokes equations (3.24) with respect to
$\Theta$ satisfy the following linearized Navier-Stokes equations :
$$\left\{\begin{array}{ll}
\widetilde{\div} w:=\frac{\partial \widehat{w}^\alpha}{\partial
x^\alpha}+\frac{\partial
\widehat{w}^3}{\partial\xi}+\frac{\widehat{w}^2}{r}=0,\\
-\nu\widetilde{\Delta}
\widehat{w}^k-\nu(r\varepsilon)^{-2}a\frac{\partial^2\widehat{w}^k}{\partial\xi^2}-\nu
P^{k3}_j(\Theta)\frac{\partial
\widehat{w}^k}{\partial\xi}-2\nu\varepsilon^{-1}\Theta_\beta\frac{\partial^2\widehat{w}^k}{\partial\xi\partial
x^\beta}-\nu P^{k\beta}_j(\Theta)\frac{\partial
\widehat{w}^j}{\partial x^\beta}-\nu
q^k_j(\Theta)\widehat{w}^j\\
\qquad+g^{k\beta}\partial_\beta
\widehat{p}+g^{k3}\partial_\xi\widehat{p}
+C^k(\widehat{w},\omega)+N^\alpha(w,\widehat{w})
+N^\alpha(\widehat{w},w)+R^k(w,p,\Theta)=0,\\
\end{array}\right.\eqno{(5.1)}$$

$$\left\{\begin{array}{ll}
\widehat{w}=0,\quad \mbox{on}\quad\Gamma_s\cap \{\xi=\xi_k\},\\
\nu \frac{\partial \widehat{w}}{\partial\,n}-\widehat{p}n=0,\quad
\mbox{on}\quad \Gamma_{in}\cap\Gamma_{out},
\end{array}\right.\eqno{(5.2)}$$
where

$$\left\{\begin{array}{ll}
R^k(w,p,\Theta)\eta:=-2\nu(r\varepsilon)^{-2}\Theta_\alpha\eta_\alpha\frac{\partial^2w^k}{\partial\xi^2}-\nu
\frac{DP^{k3}_j(\Theta)}{D\Theta}\eta\frac{\partial
w^k}{\partial\xi}-2\nu\varepsilon^{-1}\eta_\beta\frac{\partial^2w^k}{\partial\xi\partial
x^\beta}\\
\qquad-\nu \frac{DP^{k\beta}_j(\Theta)}{D\Theta}\eta\frac{\partial
w^j}{\partial x^\beta}-\nu \frac{Dq^k_j(\Theta)}{D\Theta}\eta w^j
+\frac{Dg^{k\beta}}{D\Theta}\eta\partial_\beta
p+\frac{Dg^{k3}}{D\Theta}\eta\partial_\xi p\\
\qquad+2\omega[-\Theta_\lambda\delta^k_2+
r\varepsilon^{-1}(\delta_{\lambda2}\Pi(w,\Theta)
+\Theta_2w^\lambda)\delta^k_3]\eta_\lambda+\frac{D\pi^k_{ij}(\Theta)}{D\Theta}\eta
w^iw^j,

\end{array}\right.\eqno{(5.3)}$$
\end{theorem}
\begin{proof}\quad The proof see Appendix.\end{proof}

Associated variational formulation for (4.3) is given by
$$\left\{\begin{array}{ll}
\mbox{Find}\quad \widehat{w}\in V(\Omega),\quad
\widehat{p}\in\,L^2(\Omega)\quad\mbox{such that}\, \forall\,v\in V(\Omega)\\
a_0(\widehat{\textbf{w}},\textbf{v})+(\textbf{C}(\widehat{w},\omega),\textbf{v})
+(\textbf{L}(\widehat{w},\Theta),\textbf{v})+b(\widehat{w},\widehat{w},v)-(\widehat{p},\partial_\alpha
v^\alpha)+(\frac{\partial \textbf{T}(\widehat{w},\widehat{w})}{\partial\xi},\textbf{v})\\
\quad
=(f,v),\\
(\frac{\partial \widehat{w}^\alpha}{\partial
x^\alpha}+\frac{\widehat{w}^2}{r}+\frac{\partial
\widehat{w}^3}{\partial\xi},q)=0,\quad\forall\quad
q\in\,L^2(\Omega),
\end{array}\right.\eqno{(5.4)}$$ where $b(\cdot,\cdot,\cdot)$ and
$T(\cdot,\cdot)$ are respectively defined by (4.17) and (3.43).

\section{ 2D-3C Navier-Stoke Equations  on the 2D manifold $\Im_\xi$}

 As mentioned previously, for any  $\xi=const$ , there will correspond to a two dimensional
 surface $\Im_\xi$. On the other hand, the three components of coordinate
 $(x,\xi)$ represent different meaning. the first two components $x^\alpha$
  are variables on the tangent plane to the surface $\Im_\xi$,
 which describe the flow direction in the channel, and the third component $\xi$ is transverse variable
 which describe transverse flow through  different manifolds. Therefor the Navier-Stokes equations (3.24)
 can be decomposed into two parts, the first is the operator on the tangent plane to the
 surface $\Im_\xi$, which will be named  ``\textbf{Membrane
 Operator}'',  meanwhile the second is the operator along the transverse
  direction, which is named ``\textbf{Bending Operator}''. Proceeding from this thinking,
 under the new coordinate the Navier-Stokes equations (3.24) can be rewritten  as
$$\left\{\begin{array}{ll}
\frac{\partial w^\alpha}{\partial x^\alpha}+\frac{\partial
w^3}{\partial\xi}+\frac{w^2}{r}=\frac1r\frac{\partial
(rw^\alpha)}{\partial x^\alpha}+\frac{\partial
w^3}{\partial\xi}=\widetilde{\div}_2w+\frac{\partial
w^3}{\partial\xi}=0,\\
{\cal N}^i(\bm w,p,\Theta):=-\nu\widetilde{\Delta}
w^i+\phi^{i\beta}\nabla_\beta p+C^i(\bm w,\bm \omega)-\nu
l^i(\bm w,\Theta)\\
\qquad+\frac{\partial}{\partial\xi}(\psi^i(\bm w,p,\Theta))+B^i(\bm
w,\bm w) =f^i,
\end{array}\right.\eqno{(6.1)}$$
where
$$\left\{\begin{array}{ll}
  B^i(\bm w,\bm w)= \partial_\beta(w^i w^\beta)+\pi^i_{lj}w^lw^j,&
  \psi^i(\bm w,p,\Theta)=w^3w^i+\eta^ip-\nu
l^i_\xi(\bm w,\Theta),\\
\phi^{i\beta}(\Theta)=\delta^{i\beta}-\delta^{3i}\varepsilon^{-1}\delta^{\beta\sigma}\Theta_\sigma,&
 \eta^\alpha=-\varepsilon^{-1}\Theta_\alpha,\quad
\eta^3=(r\varepsilon)^{-2}a,
\end{array}\right.\eqno{(6.2)}$$

Let restrict the Navier-Stokes equations (6.2) on the surface
$\Im_{\xi_k}$ and adopt the Euler center difference  to instead of
the derivative with respective to $\xi$ appearing in the bending
operator ${\cal N}_\xi$. Then we introduce several abbreviation
about jump operator and finite difference operators,
$$\begin{array}{ll}
\bm w(k):=\bm w|_{\xi=\xi_k},\quad [\bm w]_k:=\bm w({k+1})-\bm
w({k-1}), \\(\mbox{or})[\bm w]_k:=\bm w({k+1})-\bm w({k}),
(\mbox{or})[\bm w]_k:=\bm w({k})-\bm w({k-1}),\\
d^1_k(\bm w):=\frac{[\bm w]_k}{2\tau},\quad d^2_k(\bm
w):=-\frac{2\bm w(k)}{\tau^2}+\frac1{\tau^2}[\bm w]_k,\quad
\widetilde{d}^2_k(\bm w):=\frac1{\tau^2}[\bm w]_k,\quad
\tau=\xi_{k+1}-\xi_k,
\end{array}\eqno{(6.3)}$$
and the corresponding different quotient represent as
$$\left\{\begin{array}{ll}
\frac{\partial w^{\alpha}}{\partial\xi}|_{\xi_k}\cong
d^1_k(w^\alpha)=\frac1{2\tau}(w^\alpha|_{\xi=\xi_{k+1}}-w^\alpha|_{\xi=\xi_{k-1}})=\frac1{2\tau}[w^\alpha]_k,\\
d^2_k(w^\alpha)=\frac{\partial^2
w^{\alpha}}{\partial\xi^2}|_{\xi_k}\cong
\frac1{\tau^2}(w^\alpha|_{\xi=\xi_{k+1}}-2w^\alpha|_{\xi=\xi_{k}}+w^\alpha|_{\xi=\xi_{k-1}})
\end{array}\right.\eqno{(6.4)}$$
Under this notations,  we get
$$\begin{array}{ll}
\frac{\partial \psi^i}{\partial\xi}|_{\xi=\xi_k}=\alpha_\tau
w^i(k)+\frac{2\nu}{\tau\varepsilon}\Theta_\beta\frac{\partial
w^i}{\partial x^\beta}(k)-\frac\nu\tau q^i_{\xi
j}w^j(k)+\frac1\tau\eta^ip(k)+\frac1\tau w^3(k)w^i(k)\\
\quad\quad-\frac12\alpha_\tau[w^i]-\frac{2\nu}{\tau\varepsilon}\Theta_\beta\frac{\partial
w^i}{\partial x^\beta}(k-1)+\frac\nu\tau q^i_{\xi
j}w^j(k-1)-\frac1\tau\eta^ip(k-1)-\frac1\tau w^3(k)w^i(k-1)\\
\quad=(\alpha_\tau\delta^i_j-\frac\nu\tau q^i_{\xi j})w^j(k)
+\frac{2\nu}{\tau\varepsilon}\Theta_\beta\frac{\partial
w^i}{\partial x^\beta}(k)+\frac1\tau\eta^ip(k)+\frac1\tau w^3(k)w^i(k)+R_\tau(k-1),\\
R^i_\tau(k-1)=-\frac12\alpha_\tau[w^i]-\frac{2\nu}{\tau\varepsilon}\Theta_\beta\frac{\partial
w^i}{\partial x^\beta}(k-1)+\frac\nu\tau q^i_{\xi j}w^j(k-1)
-\frac1\tau\eta^ip(k-1)-\frac1\tau w^3(k)w^i(k-1),
\end{array}$$
where $\alpha_\tau=\frac{2\nu a}{r^2\varepsilon^2\tau^2}$

So we finally conclude that,

\begin{theorem}
The 2D-3C Navier-Stokes problem  restricted on a smooth 2D surface
$\Im_{\xi_k}$ is given by
$$\left\{\begin{array}{ll}
{\cal N}^i(k):=-\nu\widetilde{\Delta} w^i(k)+L^i_\tau(k)+C^i(\bm
w(k),\bm \omega)
+\phi^{i\beta}\nabla_\beta p(k)+\frac1\tau\eta^ip(k)\\
\qquad+B^i(\bm w(k),\bm w(k))+\frac1\tau w^3(k)w^i(k)=F_\tau^i(k),\\
 \div_2(\bm w(k))=-d^1_k(w^3),\quad (\mbox{where}\quad
 \div_2\bm w:=\frac1r\partial_\alpha(rw^\alpha),d_\tau(\bm w)=d^1_k(w^3)),\\
\end{array}\right.\eqno{(6.5)}$$
with boundary conditions
$$\left\{\begin{array}{ll}
\bm w|_{\gamma_s}=0,\quad
\gamma_s=\Gamma_S\cap\{\xi=\pm 1\},\\
  \sigma_{\bm n}(\bm w,p)|_{\gamma_{in}}=\bm h_{in},\quad\gamma_{in}=\Gamma_{in}\cup\{\xi=\xi_k\}\\
    \sigma_{\bm n}(\bm w,p)|_{\gamma_{out}}=\bm h_{out},\quad\gamma_{out}=\Gamma_{out}\cup\{\xi=\xi_k\}
\end{array}\right.\eqno{(6.6)}$$
where
$$\left\{\begin{array}{ll}
a=1+r^2|\widetilde{\nabla}\Theta|^2,\quad\sigma_{\bm n}(\bm
w,p)=-(\nu\frac{\partial w^\alpha}{\partial\bm n}-pn^\alpha)\bm
e_\alpha -(\nu\frac{\partial
w^3}{\partial\bm n})\bm e_3,\\
L^i_\tau(k)=(\alpha_\tau\delta^i_j-\frac\nu\tau q^i_{\xi j})
w^j(k)-\nu
l^i(w(k),\Theta)+\frac{2\nu}{\varepsilon\tau}\Theta_\beta\partial_\beta
w^i(k)\\
\qquad=(-\nu
P^{i\beta}_m+\frac{2\nu}{\tau\varepsilon}\Theta_\beta\delta^i_m)\partial_\beta
w^m+(-\nu q^i_j(\Theta)+\alpha_\tau\delta^i_j-\frac\nu\tau q^i_{\xi
j})w^j,\\
 F^m_\tau(k)=f^m(k)+R^m_\tau(k-1),
\end{array}\right.\eqno{(6.7)}$$
\end{theorem}

\begin{remark}
 There exists a second order differential operator in $\partial_\xi
 \psi^m(w,p,\Theta)$,
$$(r\varepsilon)^{-2}a\frac{\partial^2w^m}{\partial\xi^2},$$
and the term $\alpha_\tau w^m(k)$ in (6.5) is obtained by using the
different quotient of second order.
\end{remark}

By a similar manner with (3.36), the equation satisfied by the
covariant components of the Navier-Stokes equations are given as
$$\begin{array}{ll}
A_i(k)&:=g_{im}{\cal N}^m(k)=-\nu g_{im}\widetilde{\Delta} w^m(k)+
g_{im}L_\tau^m(k) +g_{im}C^m(\bm w(k),\bm
\omega)+g_{im}\phi^{m\beta}\nabla_\beta
p(k)\\
&+\frac1\tau g_{im}\eta^mp(k)
+g_{im}B^m(\bm w(k),\bm w(k))+\frac1\tau g_{im}w^3(k)w^m(k)=g_{im}F^m_\tau(k),\\
\end{array}$$
Simple calculation shows
$$\begin{array}{ll}
g_{im}\phi^{m\beta}\nabla_\beta
p(k)v^i=g_{im}(\delta^{m\beta}-\delta^{3m}\varepsilon^{-1}\delta^{\beta\sigma}\Theta_\sigma)
\partial_\beta p(k)v^i =(g_{i\beta}-g_{i3}\varepsilon^{-1}\Theta_\beta)
\partial_\beta
p(k)v^i\\
\qquad=(g_{\alpha\beta}-g_{3\alpha}\varepsilon^{-1}\Theta_\beta)\partial_\beta
p v^\alpha+(g_{3\beta}-g_{33}\varepsilon^{-1}\Theta_\beta)
\partial_\beta
p(k)v^3=\delta_{\alpha\beta}\partial_\beta
pv^\alpha\\
\qquad=\partial_\alpha(pv^\alpha)-p\partial_\alpha v^\alpha,\\
g_{im}\eta^mv^i=(-\varepsilon^{-1}\Theta_\alpha
a_{\alpha\beta}+(r\varepsilon)^{-2}ag_{3\beta})v^\beta+(-\varepsilon^{-1}\Theta_\alpha
g_{3\alpha}+(r\varepsilon)^{-2}ag_{33}v^3=v^3,\\

g_{im}\widetilde{\Delta} w^m
v^i=\partial_\lambda(g_{im}\partial_\lambda
w^mv^i)-g_{im}\partial_\lambda w^m\partial_\lambda
v^i-\partial_\lambda g_{im}\partial_\lambda w^m v^i,
\end{array}$$

$$\begin{array}{ll}
A_i(k)v^i&=-\nu\partial_\lambda (g_{im}\partial_\lambda w^mv^i)+\nu
g_{im}\partial_\lambda w^m\partial_\lambda v^i+(\nu\partial_\lambda
g_{im}\partial_\lambda w^m+ g_{im} L_\tau^m(k))v^i\\
&+g_{im}C^m(w(k),\omega)v^i+\partial_\alpha(pv^\alpha)-p\partial_\alpha
v^\alpha+pv^3\\
&+g_{im}B^m(w(k),w(k))v^i+\frac1\tau g_{im}w^3(k)w^m(k).
\end{array}\eqno{(6.8)}$$
By using Green's formula and boundary conditions, we get
$$\begin{array}{ll}
\int\limits_D[-\nu\partial_\lambda(g_{ij}\partial_\lambda
(w^i)v^j)+\partial_\alpha(pv^\alpha)]\rm{d}x=\int\limits_{\gamma_{in}\cup\gamma_{out}}
(-\nu g_{ij}\frac{\partial w^i}{\partial n}v^j+pn_\alpha
v^\alpha)\rm{d}s\\
\qquad=\int\limits_{\gamma_{in}\cup\gamma_{out}}
[g_{ij}\sigma^i_n(w,p) v^j]\rm{d}s,
\end{array}$$
where $\bm n$ is normal vector to
$\gamma_1=\gamma_{in}\cup\gamma_{out}$, i.e.,
$$\left\{\begin{array}{ll}
\bm n=n^\alpha \bm e_\alpha+0\bm e_3,\quad
n^3=n_3=0,\\
 \sigma_{\bm n}(\bm w,p)=-\nu\frac{\partial
\bm w}{\partial\bm n}+p\bm n.
\end{array}\right.\eqno{(6.9)} $$
 Hence the variational formulation associated with (6.5) and
(6.6) are expressed as
$$\left\{\begin{array}{ll}
\mbox{Find}\bm w(k)\in V(D),\
p\in\,L^2(D),\ \mbox{such that}\\
a_0(\bm w,\bm v)+(\bm L(k),\bm v)+(\bm C(k),\bm v)+b(\bm w,\bm w,\bm
v) +(p(k),v^3-\partial_\alpha
v^\alpha)\\
\qquad=(\bm G_\tau(k),\bm v),\quad \forall\,\bm v\in\,V(D),\\
({\div}_2 \bm w,q)=(d_\tau(\bm w),q),\quad \forall q\in\,L^2(D),
\end{array}\right.\eqno{(6.10)}$$
where
$$\left\{\begin{array}{ll}
a_0(\bm w,\bm v)=\int\limits_D\{(\nu
g_{ij}\widetilde{\nabla}_\lambda w^i\widetilde{\nabla}_\lambda
v^j)\rm{d}x\\
\qquad=\int\limits_D \nu[a_{\alpha\beta}\widetilde{\nabla}_\lambda
w^\alpha\widetilde{\nabla}_\lambda
v^\beta+r^2\varepsilon\Theta_\beta(\widetilde{\nabla}_\lambda
w^\beta\widetilde{\nabla}_\lambda v^3+\widetilde{\nabla}_\lambda
w^3\widetilde{\nabla}_\lambda
v^\beta)+r^2\varepsilon^2\widetilde{\nabla}_\lambda
w^3\widetilde{\nabla}_\lambda v^3]\rm{d}x,\\
 (\bm C(k),\bm v)=\int\limits_D( g_{ij}C^i(k)v^j\rm{d}x=\int\limits_D2r\omega[(\Theta_\beta w^2-\delta_{\beta2}
 \Pi(w,\Theta))v^\beta+\varepsilon w^2v^3]\rm{d}x,\\
(\bm L(w(k),\Theta),\bm v)=\int\limits_D[(\nu\partial_\lambda
g_{im}\partial_\lambda w^m+ g_{im}
L_\tau^m(k))v^i]\rm{d}x\\
\qquad=(L^\beta_{ij}(\Theta)\partial_\beta w^j
+L_{ij}(\Theta)w^j,v^i),\\
b(\bm w,\textbf{u},\bm v)=(\textbf{B}(w,u),\bm
v)=(g_{ij}B^i(w(k),w(k))+g_{ij}\frac1\tau
w^3(k)w^i(k),v^j)\\
\qquad =((\partial_\lambda(w^\lambda
w^\alpha)+\Pi^\alpha_{mk}w^mw^k,a_{\alpha\beta}v^\beta+\varepsilon r^2\Theta_\alpha v^3)\\
\qquad+((\partial_\lambda(w^\lambda w^3)+\Pi^3_{mk}w^mw^k),\varepsilon r^2\Pi(v,\Theta) ),\\
(\bm G_\tau(k),\bm v)=(\bm f_\tau(k),\bm v)+<\sigma_n(w,p),\bm h>|_{\gamma_1}\quad (\mbox{by}(6.7))\\
\end{array}\right.\eqno{(6.11)}$$
where $\gamma_1=\gamma_{in}\cup\gamma_{out}$,
$$\left\{\begin{array}{ll}
 \widetilde{\Delta}=\frac{\partial^2}{\partial
x^1}+\frac{\partial^2}{\partial x^2},\quad {\div}_2
w=\frac1r\frac{\partial(rw^\alpha)}{\partial
x^\alpha},\\
V(D):=\{v|v\in\,\bm h^1(D),\quad v=0\,\, \mbox{on}\,\, \gamma_s\},
\end{array}\right.\eqno{(6.12)}$$ and
$$\left\{\begin{array}{ll}
L^\beta_{ij}(\Theta)=\frac{2\nu}{\varepsilon\tau}\Theta_\gamma
g_{ij}-\nu g_{im}P^{m\beta}_j(\Theta)+\nu\partial_\beta g_{ij},\\
L_{ij}(\Theta)=\alpha_\tau g_{ij}-\frac\nu\tau g_{im}q^m_{\xi j}-\nu g_{im}q^m_j,\\
\Pi^3_{ij}(\Theta)=\pi^\alpha_{ij}(\Theta)+\frac1\tau\delta_{3i}\delta_{\alpha
j},\quad
\Pi^3_{ij}(\Theta)=\pi^3_{ij}(\Theta)+\frac1\tau\delta_{3i}\delta_{3
j}

\end{array}\right.\eqno{(6.13)}$$

\section{ Pressure Correction Equation on the Blade Surface}

Noting that we must give value of $[p]$ in the source term
 $\bf F_\tau(k)$ of equations (6.12), so the pressure on the surface $\Im$ must be supplied.
Therefore, we recall the Navier-Stokes equations in invariant form
$$\left\{\begin{array}{ll}
-\nu g^{jk}\nabla_j\nabla_kw^i+w^j\nabla_j
w^i+2\varepsilon^{ijk}g_{jm}g_{kl}\omega^mw^l+g^{ij}\nabla_jp=f^i,\\
\nabla_jw^j=0,
\end{array}\right.\eqno{(7.1)}$$
Let take divergence $\nabla_i$ for $(6.1)_1$ and apply the identity
$\nabla_k g_{ij}=\nabla_kg^{ij}=\nabla_k\varepsilon^{ijm}=0$, then
we have
$$\begin{array}{ll}
-\nu g^{jk}\nabla_i\nabla_j\nabla_kw^i+\div((\bm w\cdot\nabla)\bm w)
)+\div(2\bm \omega\times \bm w)+g^{ij}\nabla_i\nabla_jp=\div \bm f.\\
\end{array}\eqno{(7.2)}$$
 Because the Riemann curvature tensor  vanishes in Euclidean space
 $R^3$, therefore by exchanging the order of covariant
 derivatives, we have
$$\begin{array}{ll}
-\nu
g^{jk}\nabla_i\nabla_j\nabla_kw^i=g^{jk}\nabla_j\nabla_k(\nabla_iw^i)=0,\\
\end{array}\eqno{(7.3)}$$
In addition, a simple calculation shows that
$$\left\{\begin{array}{ll}
\div((\bm w\cdot\nabla)\bm w)=\nabla_i(w^j\nabla_j
w^k)=\nabla_iw^j\nabla_jw^i+w^j\nabla_j\nabla_iw^i=\nabla_iw^j\nabla_jw^i\\
\qquad=\nabla_\alpha w^\beta\nabla_\beta w^\alpha+2\nabla_3w^\beta
\nabla_\beta w^3+\nabla_3w^3\nabla_3w^3.
\end{array}\right.\eqno{(7.4)}$$
From (2,7) we have,
$$\begin{array}{ll}
\bm C=2\bm \omega\times \bm w=C^i\bm e_i,\quad C^1=0,\quad
C^2=-2\omega r\Pi(\bm w,\Theta),\quad C^3=2\omega
\varepsilon^{-1}(r\Theta_2\Pi(\bm w,\Theta)+\frac{w^2}{r}),\\
 \div
\bm C=\frac1{\sqrt{g}}\frac{\partial \sqrt{g}C^i}{\partial
x^i}=C^i\partial_i\ln(er)+\partial_iC^i=r^{-1}C^2+\frac{\partial}{\partial
r}(-2r\omega\Pi(\bm w,\Theta))\\
\qquad=-2\omega r^{-1}(2r\Pi(\bm w,\Theta)+
r^2\frac{\partial}{\partial
r}\Pi(\bm w,\Theta))=-\frac{2\omega}{r}\frac{\partial}{\partial r}(r^2\Pi(\bm w,\Theta)).\\
\end{array}\eqno{(7.5)}$$
On the other hand, the Laplace-Betrami operator can be expressed as
$$\begin{array}{ll}
\Delta
p&=\frac1{\sqrt{g}}\partial_j(g^{jk}\sqrt{g}\partial_kp)=(r\varepsilon)^{-1}[
\partial_\alpha(r\varepsilon g^{\alpha\beta}\partial_\beta
p)+(\partial_\xi(r\varepsilon g^{3\lambda}\partial_\lambda
p)+\partial_\lambda(r\varepsilon g^{\lambda3}\partial_\xi p))\\
&+\partial_\xi(r\varepsilon g^{33}\partial_\xi p)],
\end{array}$$
By using (2.5), we claim that
$$\begin{array}{ll}
\Delta p&=\frac1{r\varepsilon}[\frac{\partial}{\partial
x^\alpha}(r\varepsilon\frac{\partial p}{\partial
x^\alpha})-(2r\Theta_\lambda)\frac{\partial^2 p}{\partial
x^\lambda\partial\xi}-(\Theta_2+r\widetilde{\Delta}\Theta)\frac{\partial
p}{\partial\xi}+(r\varepsilon)^{-1}a\frac{\partial^2
p}{\partial\xi^2}],
\end{array}\eqno{(7.6)}$$
Assume that centrifugation force is the only exterior force, that is
$$\bm f=-\bm \omega\times\bm \omega\times R=-\omega^2\bm R,$$
then
$$\nabla_if^i=-  \omega^2\nabla_ir^i=\omega^2(\nabla_\alpha
r^\alpha+\nabla_3r^3).$$
 On the other hand, we have
$$\begin{array}{ll}
\bm R=r\bm e_r=r(\bm e_2-\varepsilon^{-1}\Theta_2\bm e_3),\quad r^2=r,\quad
r^1=0,\quad r^3=-\varepsilon^{-1}r\Theta_2,\\
\nabla_ir^i=\frac{\partial r^\alpha}{\partial
x^\alpha}-r\Theta_2\Pi(w,\Theta)+\frac{\partial
r^3}{\partial\xi}+\frac{r^2}{r}+r\Theta_2\Pi(w,\Theta)=1+0+1=2.\,\\
\div(\bm f)=-2|\bm \omega|^2.\end{array}$$
 Summing up the above conclusions, we
get
$$\left\{\begin{array}{ll}
\frac1{r\varepsilon}[\frac{\partial}{\partial
x^\alpha}(r\varepsilon\frac{\partial p}{\partial
x^\alpha})-\frac{\partial}{\partial\xi}(2r\Theta_\lambda\frac{\partial
p}{\partial
x^\lambda}+(\Theta_2+r\widetilde{\Delta}\Theta)p)+(r\varepsilon)^{-1}a\frac{\partial^2
p}{\partial\xi^2}]\\
\qquad+\nabla_\alpha w^\beta\nabla_\beta w^\alpha+2\nabla_3w^\beta
\nabla_\beta
w^3+\nabla_3w^3\nabla_3w^3\\
\qquad-\frac{2\omega}{ r}\frac{\partial}{\partial
r}(r^2\Pi(w,\Theta))=-2|\bm \omega|^2.
\end{array}\right.\eqno{(7.7)}$$

Next we consider the restriction of equation (7.7) on any surface
$\Im_{\xi_k}$.  Noting that the Laplace-Betrami operator of pressure
$p$ in the new curvilinear coordinate system $(x^\alpha,\xi)$ can be
split as the sum of two operators , membrane operator on tangent
space and the bending operator along the rotational direction
$$\left\{\begin{array}{ll}
-\Delta p=-\Delta_mp-\Delta_bp,\\
-\Delta_mp=-\frac1{r\varepsilon}\frac{\partial}{\partial
x^\alpha}(r\varepsilon\frac{\partial p}{\partial
x^\alpha})=-\frac1{r}\frac{\partial}{\partial
x^\alpha}(r\frac{\partial p}{\partial
x^\alpha}),\\
-\Delta_bp=-\frac1{r\varepsilon}[(r\varepsilon)^{-1}a\frac{\partial^2
p}{\partial\xi^2}+
\frac{\partial}{\partial\xi}(2r\Theta_\lambda\frac{\partial
p}{\partial x^\lambda}+(\Theta_2+r\widetilde{\Delta}\Theta)p)].
\end{array}\right.\eqno{(7.8)}$$
We approximate the  derivatives with respect to rotational variable
 in (7.7) by the difference quotients defined by (6.5), and then restricted it on the $\Im_{\xi_k}$,
 finally we get
$$\left\{\begin{array}{ll}
-\frac1r\frac{\partial }{\partial x^\alpha}(r\frac{\partial
p_k}{\partial x^\alpha})
+\alpha_\tau p_k=f_k(\tau),\\
p|_{\gamma_{in}}=p_0,\\
\partial_np=\tilde{f}_n,\mbox{other boundaries }.
\end{array}\right.\eqno{(7.9)}$$
where
$$\left\{\begin{array}{ll}
\alpha_\tau(x):=\frac a{\tau^2r\varepsilon}>\frac1{r\varepsilon\tau^2},\quad\forall \,x\in\,\overline{D}\\
 f_k(\tau)=- 2r\Theta_\lambda
d^1_{k}(\frac{\partial p}{\partial
x^\lambda})-(\Theta_2+r\widetilde{\Delta}\Theta)d^1_{k}(p)+\frac
a{r\varepsilon}\widetilde{d}^2_k(p)\\
\qquad-2\omega r^{-1}\frac{\partial}{\partial
r}(r^2\Pi(w(k),\Theta))-2|\omega|^2\\
+\nabla_\alpha w^\beta(k)\nabla_\beta
w^\alpha(k)+2\widetilde{\nabla}_3w^\beta(k) \nabla_\beta
w^3(k)+\widetilde{\nabla}_3w^3(k)\widetilde{\nabla}_3w^3(k),\\
\widetilde{\nabla}_3w^\lambda:=d^1_k(w^\lambda)-r\varepsilon\delta_{2\lambda}\Pi(w(k),\Theta),\\
\widetilde{\nabla}_3w^3:=d^1_k(w^3)+r^{-1}w^2(k)+r\Theta_2\Pi(w(k),\Theta).
\end{array}\right.\eqno{(7.10)}$$

Specially, we consider the restriction of (7.7) on the surface
$\Im_{\pm1}$. When $\xi=\pm1$, we have
$$\begin{array}{ll}
\bm w|_{\xi=\pm1}=0, \quad \partial_\alpha \bm w=0,\quad
\nabla_\alpha
w^i=0, \\
 \nabla_3w^3(\pm1)=\frac{\partial w^3}{\partial\xi}(\pm1)=-(\partial_\alpha w^\alpha+r^{-1}w^2)(\pm1)=0,\\
\Pi(w,\Theta)=\frac{\partial\Pi(w,\Theta)}{\partial r}=0,(\mbox{see
(3.19)}).
\end{array}$$
At present, equation (7.9) becomes
$$\left\{\begin{array}{ll}
-\frac1{r}\frac{\partial}{\partial x^\alpha}(r\frac{\partial
p}{\partial
x^\alpha})\\
=-\frac{\partial}{\partial\xi}(2r\Theta_\lambda\frac{\partial
p}{\partial
x^\lambda}+(\Theta_2+r\widetilde{\Delta}\Theta)p)+(r\varepsilon)^{-1}a\frac{\partial^2
p}{\partial\xi^2}-2(\omega)^2,\\
\end{array}\right.\eqno{(7.11)}$$
Furthermore, we replace the derivatives $\frac{\partial w}
{\partial\xi}$  by difference quotient and apply the boundary
condition $w|_{\xi=\pm1}=0$, then
$$\begin{array}{ll}
\frac{\partial\bm w}{\partial\xi}|_{\xi=-1}=\frac1{\tau}(\bm
w|_{\xi=-1+\tau}-\bm w|_{\xi=-1})=
\frac1{\tau}(\bm w|_{\xi=-1+\tau}),\\
\frac{\partial^2\bm w}{\partial\xi^2}|_{\xi=-1}=\frac1{\tau^2}(\bm
w|_{\xi=-1+2\tau}-2\bm w|_{\xi=-1+\tau}),
\end{array}\eqno{(7.12)}$$

Borrowing the notations in (6.3), when $\xi=\pm1$,
$$\left\{\begin{array}{ll}
d_{-1}(\bm w)=\frac{\partial\bm  w}{\partial\xi}(-1)= \frac1{\tau}\bm w(-1+\tau),\quad\bm (w(-1)=0),\\
d^2_{-1}(\bm w)=\frac{\partial^2\bm
w}{\partial\xi^2}(-1)=\frac1{\tau^2}(\bm w(-1+2\tau)-2\bm w(-1+\tau)),\quad (\bm w(-1)=0),\\
\end{array}\right.\eqno{(7.13)}$$
 Similarly, when $\xi=1$,
$$\left\{\begin{array}{ll}
d_1(\bm w)=\frac{\partial\bm w}{\partial\xi}(1)=- \frac1{\tau}\bm w(1-\tau),\quad (\bm w(1)=0),\\
d^2_1(\bm w)=\frac{\partial^2\bm
w}{\partial\xi^2}(1)=\frac1{\tau^2}(\bm w(1-2\tau)-2\bm w(1-\tau)),\quad (\bm w(1)=0),\\
\end{array}\right.\eqno{(7.14)}$$
In addition, Owing to $\xi=-1$ and $\xi=1$ are both sides of blade ,
hence
$$p(-1-\tau)\simeq p(1),\quad p(1+\tau)=p(-1).\eqno{(7.15)}$$
therefore, we rewrite (7.13), (7.14) as
$$\left\{\begin{array}{ll}
d_{-1}(p):=\frac{\partial p}{\partial\xi}(-1)=\frac1{\tau}(p(-1+\tau)-p(-1)),\\
d^2_{-1}(p):=\frac{\partial^2
p}{\partial\xi^2}(-1)=\frac1{\tau^2}(p(-1+\tau)-2
p(-1)+p(-1-\tau))\\
\qquad=\frac1{\tau^2}(p(-1+\tau)-2
p(-1)+p(1)),\quad (p(-1-\tau)=p(1)),\\
\widetilde{d}^2_{-1}(p):=\frac1{\tau^2}(p(-1+\tau)+p(1)),\quad (p(-1-\tau)=p(1)),\\
\end{array}\right.\eqno{(7.16)}$$
$$\left\{\begin{array}{ll}
d_1(p):=\frac{\partial p}{\partial\xi}(1)=\frac1{\tau}(p(1)-p(1-\tau)),\\
d^2_1(p):=\frac{\partial^2
p}{\partial\xi^2}(1)=\frac1{\tau^2}(p(1+\tau)-2p(1)+p(1-\tau))\\
\qquad=\frac1{\tau^2}(p(-1)-2p(1)+p(1-\tau)),\quad (p(1+\tau)=p(-1)),\\
\widetilde{d}^2_1(p):=\frac1{\tau^2}(p(-1)+p(1-\tau)).
\end{array}\right.\eqno{(7.17)}$$
Summing up and introducing $p_-=p(-1)$, from (7.5) the equation of
the pressure on the $\xi=-1$ surface is give by
$$\left\{\begin{array}{ll}
-\frac1r\frac{\partial }{\partial x^\alpha}(r\frac{\partial
p_-}{\partial x^\alpha})
+\alpha_\tau p_-=f_-(\tau),\\
f_-(\tau)= -2r\Theta_\lambda d_{-1}(\frac{\partial p}{\partial
x^\lambda})-(\Theta_2+r\widetilde{\Delta}\Theta)d_{-1}(p)+\frac
a{r\varepsilon}\widetilde{d}^2_{-1}(p)\\
\qquad+d^2_{-1}(w^3)-2|\omega|^2,
\end{array}\right.\eqno{(7.18)}$$
By a similar manner we can obtain the equation of pressure on the
surface $\xi=1$
 $$\left\{\begin{array}{ll}
-\frac1r\frac{\partial }{\partial x^\alpha}(r\frac{\partial
p_+}{\partial x^\alpha})
+\alpha_\tau p_+=f_+(\tau),\\
f_+(\tau)= -2r\Theta_\lambda d_{1}(\frac{\partial p}{\partial
x^\lambda})-(\Theta_2+r\widetilde{\Delta}\Theta)d_{1}(p)+\frac
a{\tau^2r\varepsilon}\widetilde{d}^2_{1}(p)\\
\qquad+d^2_{1}(w^3)-2|\omega|^2
\end{array}\right.\eqno{(7.19)}$$

\quad\quad  Next, in order to inspect the reliability of the method,
we make some numerical simulations of pressure field by using the
pressure correction equations on the blade surface (7.18),\ (7.19),\
and boundary conditions in (7.9). The low speed large-scale
centrifugal impeller of The NASA is used as the example([22]),
 and some comparison with FLUENT's conclusions is diagramed as below,

\begin{center}
\includegraphics[width=120mm]{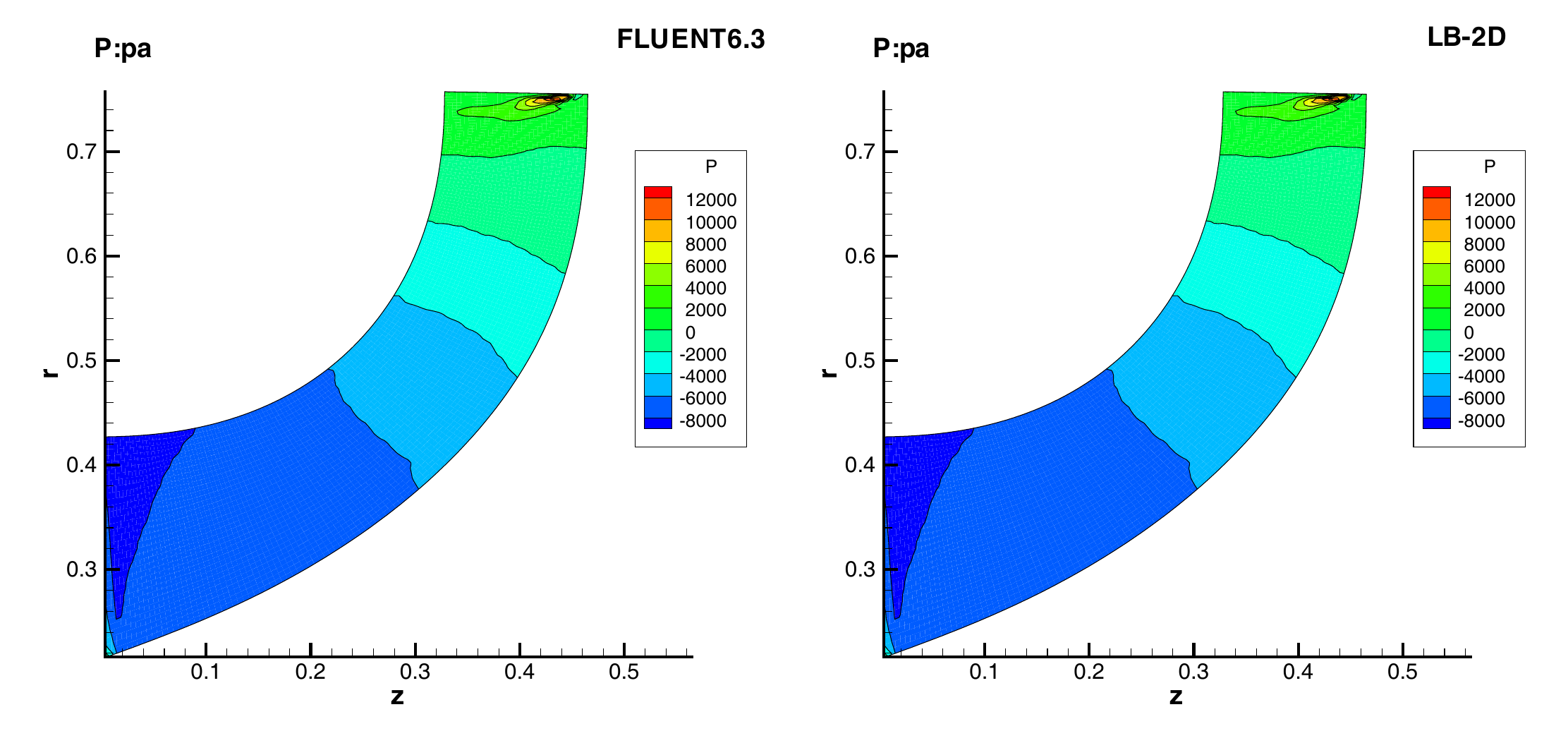}

(a)Pressure distribution on the pressure surface

\includegraphics[width=120mm]{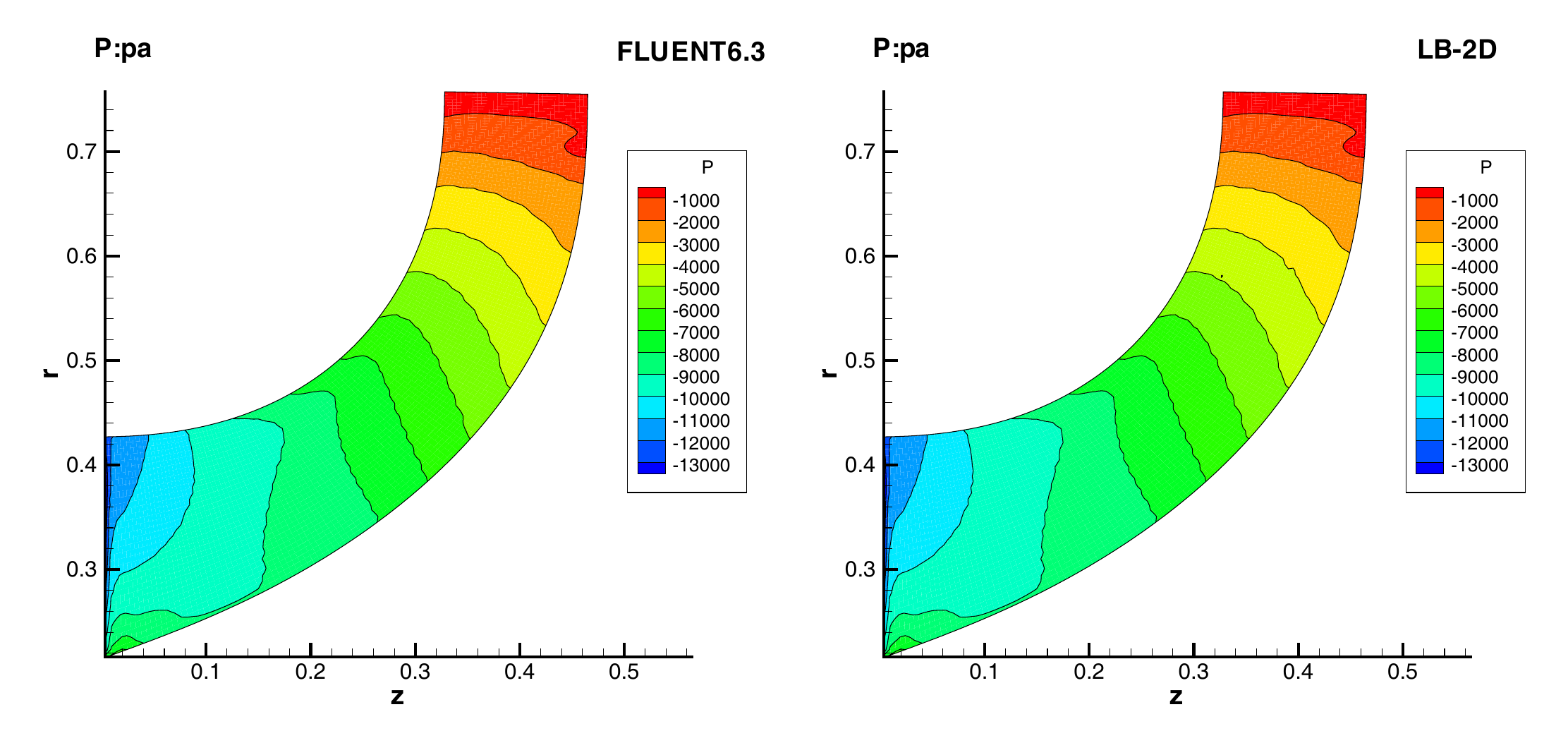}

(b)Pressure distribution on the suction surface

\includegraphics[width=120mm]{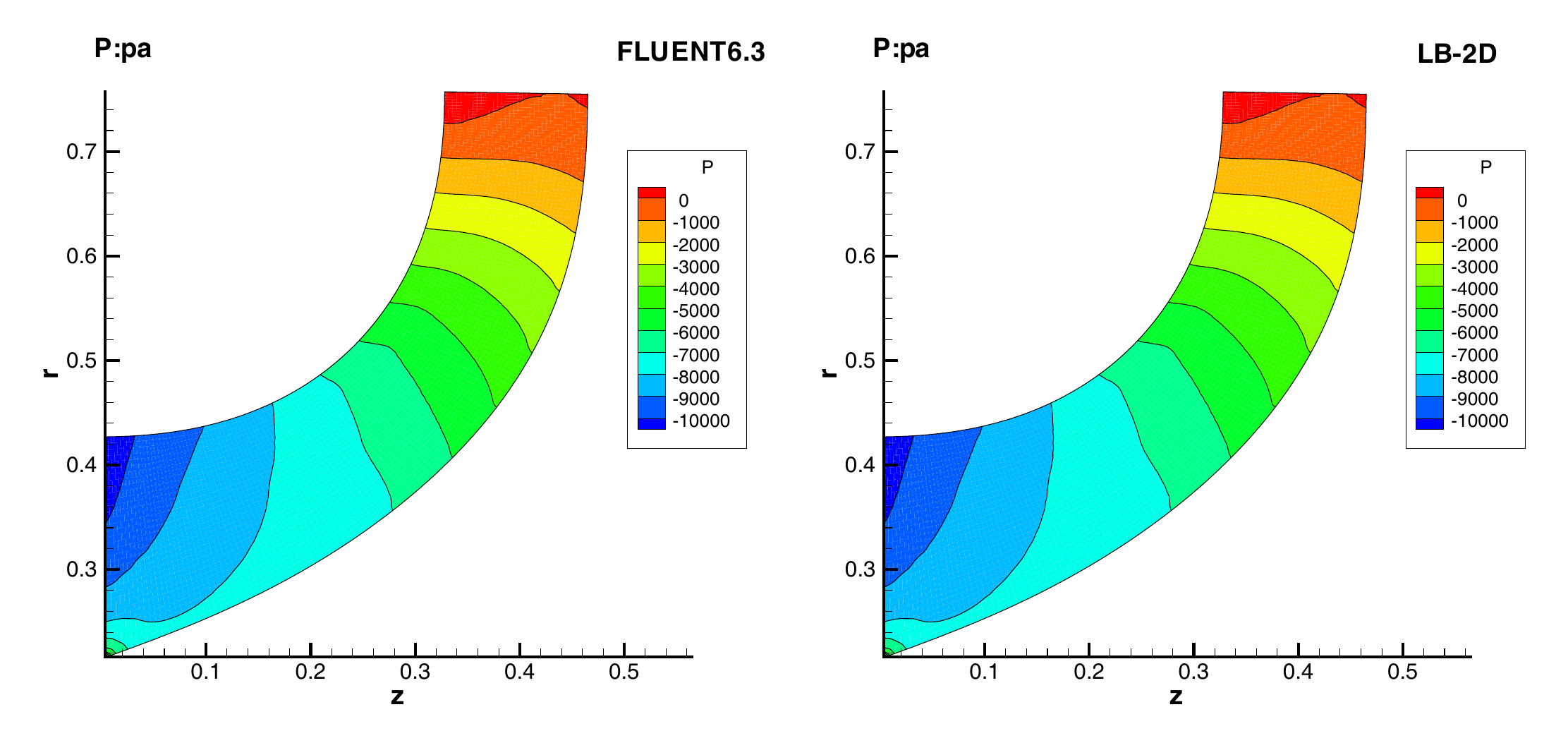}

(c)Pressure distribution on the median surface

Fig5.\ Numerical comparison of the Fluent's conclusions and the
Pressure Correction method

 [Numerical results was completed by Chen Hao, Energy and Power Engineering College of Xi'an Jiaotong University]
\end{center}
Where the LB-2D on the RHS are the results from the Pressure
Correction method.\

\quad\quad Fluid power,\ Fluent evaluates to 13973 watt,\ our method
is 13975 watt.\

\section{Steam Layer in Domain decomposition and the Bi-parallel Algorithm}

Next we we consider the decomposition of the flow passage,
$$\Omega=D\times
\{-1,1\} = \sum\limits_{k}\{D\times
[\xi_k,\xi_{k+1}]\}=\sum\limits_k\Omega_k$$ where
$-1=\xi_0<\xi_1<\cdots<\xi_m=1$. In the subsequent $\Omega_k$ is
called ``stream layer".

\begin{center}
\includegraphics[height=70mm]{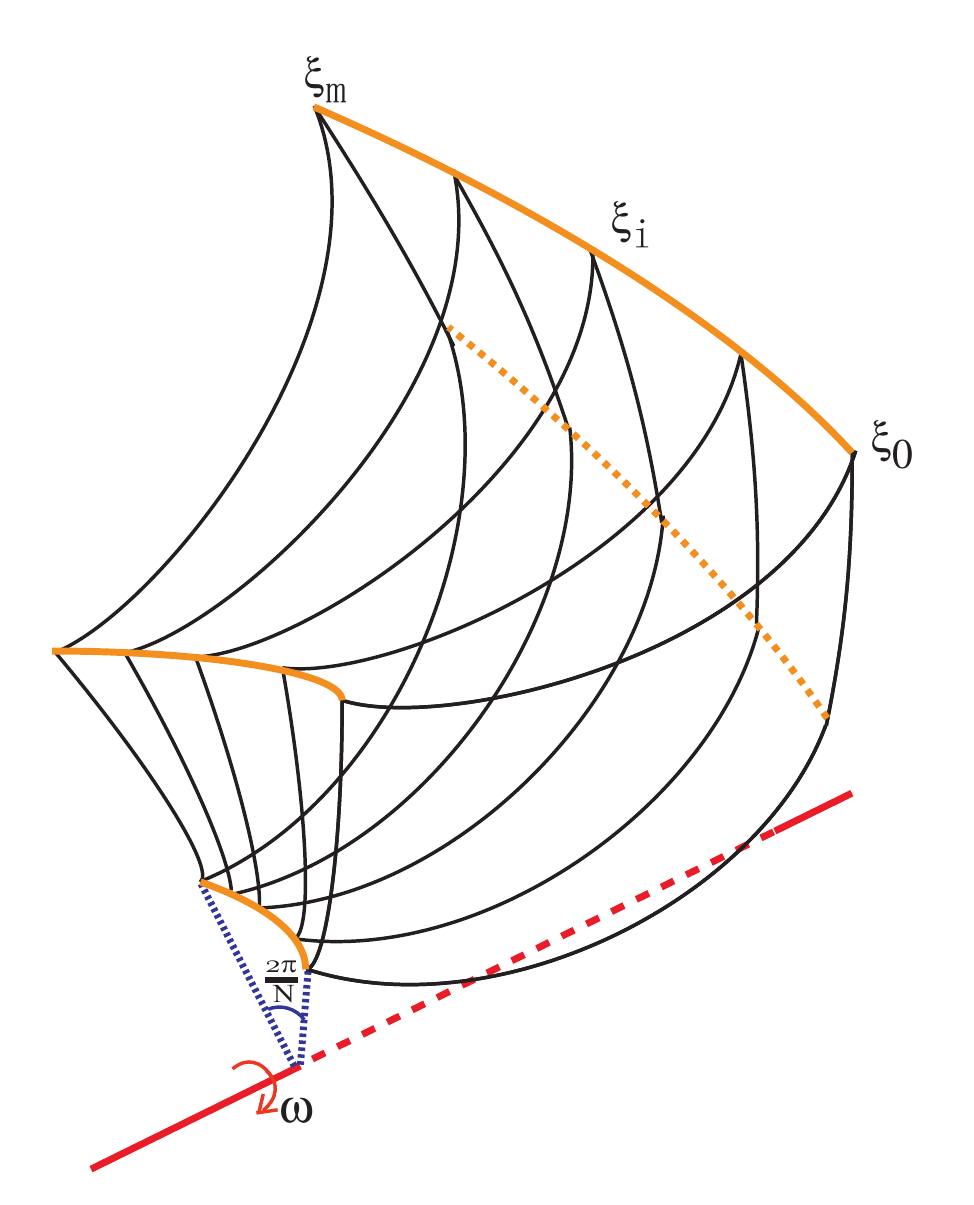}

 Fig 3:\quad decomposition of the flow passage and angular expansion
\end{center}

The new method is to solve the 2D problem (6.10) with respect to the
velocity and the pressure $(\bm w,p)$ and the pressure correction
equation (7.9) on the interface $\Im_k$ of two stream layer
$\Omega_k$ and $\Omega_k+1$, which is a 2D manifold. The general
parallel algorithm can be used to solve these 2D-problems (6.10) and
(7.9).

\vskip 0.5\baselineskip
 {\noindent\bf Bi-Parallel Algorithm:} The
Bi-Parallel Algorithm means that we adopt the parallel algorithm to
solve the problems (6.10) and (7.9) along the two directions, i.e.,
on the 2D manifold $\Im_k$ and along the direction $\xi$. On the
specified 2D manifold $\Im_k$, the general domain decomposition
method or data parallel algorithms can be used to implement the
parallel algorithm. On the other hand, the problems on the 2D
manifold $\Im_k$ corresponding to different discrete parameter
$\xi_k, k=1,\cdots,m$, can be solved at the same time, which forms
another parallel. \vskip 0.5\baselineskip

 As new method is applied to solve 3D-viscous flow in
 turbo-machinery, all interfaces $\Im_{\xi_k}$ have the same
 geometry properties, i.e., the same
 $a_{\alpha\beta},b_{\alpha\beta},\cdots$.
On the other hand , when this methods is applied to other 3D-flow,
for example, circulation flow through the aircraft, geophysical flow
around the earth, in this case the interface surface $\Im_{\xi_k}$
have different geometry properties. Suppose the next interface
$\Im_{\xi_{k+1}}$ is generated by a displacement $\bm \eta$ of the
previous interface $\Im_{\xi_k}$, then the new fundamental forms
$(a_{\alpha\beta}(\bm \eta),b_{\alpha\beta}(\bm \eta)$ can be
computed by the following formulas,

\begin{theorem} Assume that $\Im$ is a smooth surface in $\Re^3$,
 $a_{\alpha\beta},b_{\alpha\beta}$ are metric and curvature tensors
 respectively. Given a smooth displacement field
 $\bm \eta=\eta^\alpha\bm e_\alpha+\eta^3\bm n$ of $\Im$, we get the new surface $\Im(\eta)$,
  and use symbols  $a_{\alpha\beta}(\eta),b_{\alpha\beta}(\eta)$ to denote the metric tensor
  and the curvature tensors of the surface $\Im(\eta)$.
Then a simple calculation shows that they can  be expressed as(see
[3,4]),
$$\left\{\begin{array}{ll}
 a_{\alpha\beta}(\bm \eta)&=a_{\alpha\beta}+2\stackrel{0}{E}_{\alpha\beta}(\eta),\\
  b_{\alpha\beta}(\bm \eta)&=b_{\alpha\beta}+\rho_{\alpha\beta}(\eta)+Q^2_{\alpha\beta}(\eta),\\
Q^2_{\alpha\beta}(\bm
\eta)&=(b_{\alpha\beta}+\rho_{\alpha\beta}(\eta))(q(\eta)-1) +
q(\eta)[\phi_{\alpha\beta}(\eta)d(\eta)
+\phi^\sigma_{\alpha\beta}(\eta)m_\sigma(\eta)\\
&\quad-(\rho_{\alpha\beta}^\sigma(\eta)+\stackrel{\ast}{\Gamma^\lambda_{\alpha\beta}}
\stackrel{0}{\nabla}_
\lambda\eta^\sigma)\stackrel{0}{\nabla}_\sigma\eta^3],
 \end{array}\right.\eqno{(8.1)}$$where
 $$\left\{\begin{array}{ll}
 \stackrel{0}{E}_{\alpha\beta}(\bm \eta)=\gamma_{\alpha\beta}(\eta)+\frac12]a_{\lambda\sigma}\stackrel{0}{\nabla}_\alpha
 \eta^\lambda\stackrel{0}{\nabla}_\beta\eta^\sigma
 +\stackrel{0}{\nabla}_\alpha\eta^3\stackrel{0}{\nabla}_\beta\eta^3]
 ,\\
\rho_{\alpha\beta}(\bm
\eta)=\stackrel{\ast}{\nabla}_\alpha\stackrel{0}{\nabla}_\beta\eta^3
+b_{\alpha\sigma}\stackrel{0}{\nabla}_\beta\eta^\sigma,\quad
 \gamma_{\alpha\beta}(\bm \eta)=\frac12(a_{\beta\lambda}\stackrel{0}{\nabla}_\alpha\eta^\lambda
 +a_{\alpha\lambda}\stackrel{0}{\nabla}_\beta\eta^\lambda),\\
\stackrel{0}{\nabla}_\beta\eta^\sigma=\stackrel{*}{\nabla}_\beta\eta^\sigma-b^\sigma_\beta\eta^3,\quad
\stackrel{*}{\nabla}_\beta\eta^\sigma=\partial_\beta\eta^\sigma-\stackrel{\ast}{\Gamma^\sigma_{\alpha\beta}}\eta^\alpha,

\end{array}\right.\eqno{(8.2)}$$
$$\left\{\begin{array}{ll}
\rho^\sigma_{\alpha\beta}(\bm
\eta)&=\stackrel{\ast}{\nabla}_\alpha\stackrel{0}{\nabla}_\beta\eta^\sigma
-b^\sigma_\alpha\stackrel{0}{\nabla}_\beta\eta^3,\\
\phi_{\alpha\beta}(\bm
\eta)&=b_{\alpha\beta}+\rho_{\alpha\beta}(\eta)+\stackrel{\ast}{\Gamma^\lambda_{\alpha\beta}}
\stackrel{0}{\nabla}_\lambda\eta^3,\quad
\phi^\sigma_{\alpha\beta}(\bm
\eta)=\rho^\sigma_{\alpha\beta}(\eta)+(\stackrel{0}{\nabla}_\lambda\eta^\sigma
+\delta^\sigma_\lambda)\stackrel{\ast}{\Gamma^\lambda}_{\alpha\beta},\\
\end{array}\right.\eqno{(8.3)}$$
$$\left\{\begin{array}{ll}
d_\sigma(\bm
\eta)&=m_\sigma(\eta)-\stackrel{0}{\nabla}_\sigma\eta^3,\quad
d_0(\eta)=1+d(\eta),\\
d(\bm \eta)&=\gamma_0(\eta)+\det{(\stackrel{0}{\nabla}_\alpha\eta^\beta)},\\
m_\sigma(\bm
\eta)&=\varepsilon^{\nu\mu}\varepsilon_{\sigma\lambda}\stackrel{0}{\nabla}_\nu\eta^\lambda
\stackrel{0}{\nabla}_\mu\eta^3,\\
m_1(\bm \eta)&=2\left|\begin{array}{ll}
\stackrel{0}{\nabla}_1\eta^2\quad\stackrel{0}{\nabla}_2\eta^2\\
\stackrel{0}{\nabla}_1\eta^3\quad\stackrel{0}{\nabla}_2\eta^3,\\
\end{array}\right|,\quad
m_2(\bm \eta)=-2\left|\begin{array}{ll}
\stackrel{0}{\nabla}_1\eta^1\quad\stackrel{0}{\nabla}_2\eta^1\\
\stackrel{0}{\nabla}_1\eta^3\quad\stackrel{0}{\nabla}_2\eta^3,\\
\end{array}\right|,\\
\end{array}\right.\eqno{(8.4)}$$
where $Q^2_{\alpha\beta}(\bm \eta)$ is a remainder term which order
is higher than 1.
\end{theorem}

\begin{proof} The proof is omitted. \end{proof}

\section{\bf Existence of Solution of the 2D-3C Variational Problem}

 In this section we study the 2D-3C variational problem  (6.10) on the manifold
 $\Im_\xi$. First, let $\partial D=\gamma_s\cup\gamma_0$ and introduce the Sobolev space
 $V(D)$ defined by
$$V(D)=\{\bm w|,\bm w=(w^\alpha,w^3)\in\,H^1(D)\times H^1(D)\times H^1(D),
 \bm w|_{\gamma_s}=0\},
\eqno{(9.1)}$$ equipped  with the usual Sobolev norms
$$\left\{\begin{array}{ll}
|\bm w|^2_{1,D}=\sum\limits_{\alpha}\sum\limits_j\|\partial_\alpha
w^j\|^2_{0,D},\quad \|\bm w\|^2_{0,D}=\sum\limits_{i}\|w^i\|^2_{0,D}
=\sum\limits_i\int\limits_D|w^i|^2\rm{d}x,\\
 \|\bm w\|^2_{1,D}=|w|^2_{1,D}+\|w\|^2_{0,D}.
\end{array}\right.\eqno{(9.2)}$$

 It is clear that the variational problem (6.10) is a saddle point problem. In order to
 regularize it we introduce the artificial viscosity $\eta$ such that
 $$\left\{\begin{array}{ll}
\mbox{Find}\bm w\in V(D),\
p\in\,L^2(D),\ \mbox{such that}\\
(\alpha_\tau \bm w,\bm v)+a_0(\bm w,\bm v)+(\bm L(\bm w,\Theta),\bm
v) +(\bm C(\bm w,\omega),\bm v)+b(\bm w,\bm w,\bm
v)-(p,\widetilde{\div}
\bm v)\\
\quad=<\bm G_\tau,\bm v>,\quad \forall\,\bm v\in\,V(D),\\
\eta(p,q)+({\div}_2\bm w-d_\tau,q)=0,\quad \forall q\in\,L^2(D),
\end{array}\right.\eqno{(9.3)}$$
where
$$<\bm G_\tau,\bm v>=(\bm f_\tau,\bm v)+<\bm h,\bm v>|_{\gamma_{1}},
\quad \eqno{(9.4)}$$ Obviously, problem (9.3)
 is equivalent to
$$\left\{\begin{array}{ll}
\mbox{Find}\bm w\in\,V(D),\ \mbox{such that}\\
A_0(\bm w,\bm v)+\eta^{-1}({\div}_2\bm w,\widetilde{\div}\bm v))
+b(\bm w,\bm w,\bm v)=<\bm G,\bm v>,\quad \forall\bm v\in\,V(D),\\
p=\eta^{-1}[-{\div}_2 \bm w+d_\tau(\bm w)],
\end{array}\right.\eqno{( 9.5)}$$
where
$$\left\{\begin{array}{ll}
A_0(\bm w,\bm v)&=(\alpha_\tau \bm w,\bm v)+a_0(\bm w,\bm v)+(\bm
L(\bm w,\Theta),\bm v)
+(\bm C(\bm w,\bm \omega),\bm v),\\
<\bm G,\bm v>&=(\bm f_\tau,\bm v)+<\bm h_{in},\bm v>|_{\Gamma_{in}}
+(\eta^{-1}d_\tau,\widetilde{\div}\bm v)-<\eta^{-1}d_\tau,\bm v>,
\end{array}\right.\eqno{(9.6)}$$

Our first objective is to show that the bilinear form
$A_0(\cdot,\cdot)$  defined by (9.6) is $V(D)$-elliptic.

\begin{theorem} Let $D$ be a bounded domain  in $R^2$, the injective mapping
$\mathscr{\bm\Re}(x)$ defined by (2.1) satisfies
$\mathscr{\bm\Re}\in\,C^3(\stackrel{-}{D}),|D^i\mathscr{\bm\Re}|_{\infty,D}\leq
k_0$, and the two vectors $\bm e_\alpha=\partial_\alpha
\mathscr{\bm\Re}$ are linearly independent at all points of
$\stackrel{-}{D}$. Let $\gamma_0$ be a $d\gamma$-measurable subset
of $\gamma=\partial D$ and $d\gamma_0>0$. Then there exist the
constants $C(k_0), C_0(k_0)$ depend upon $x$ such that the following
equality and inequalities hold
$$\left\{\begin{array}{lll}
(i)& a_0(\bm w,\bm v)=a_0(\bm v,\bm w),& \forall\,\bm w,\bm v\in\,V(D),\\
(ii)& |a_0(\bm w,\bm v)|\leq\nu(1+r_1^2k_0^2)|\bm w|_{1,D}|\bm v|_{1,D},& \forall\,\bm w,\bm v\in\,V(D),\\
(iii)&|a_0(\bm w,\bm w)|\geq \nu |\bm w|^2_{1,D},& \forall\,\bm w,\bm v\in\,V(D),\\
\end{array}\right.\eqno{(9.7)}$$
$$\left\{\begin{array}{lcl}
(iv)&|A_0(\bm w,\bm v)|\leq & C(k_0)\|\bm w\|_{1,D}\|\bm v\|_{0,D},\quad\quad\forall\,\bm w,\bm v\in\,V(D),\\
(v)& A_0(\bm w,\bm w)\geq &
(\frac2{r^2_1\varepsilon^2\tau^2}-C_0(k_0)k_0)\|\bm w\|^2_{0,D}+(\nu-C_1(k_0)k_0)|\bm w|^2_{1,D}\\
&&+\int\limits_D\frac1r\partial_r(w^\alpha w^\alpha)\rm{d}x,\quad\quad\quad  \forall\,\bm w\,\in\,V(D),\\
\end{array}\right.\eqno{(9.8)}$$
\end{theorem}

\begin{proof} Firstly, from (6.11), we have
% $r_0\leq r\leq r_1,\varepsilon\leq 1$,
$$\begin{array}{ll}
 a_0(\bm w,\bm v)=\int\limits_D[\nu g_{ij}\widetilde{\nabla}_\lambda
 w^i\widetilde{\nabla}_\lambda v^j]\rm{d}x
=\int\limits_D[\nu a_{\alpha\beta}\widetilde{\nabla}_\lambda
 w^\alpha\widetilde{\nabla}_\lambda
 v^\beta\\
 \qquad+(\varepsilon r\widetilde{\nabla}_\lambda
 w^3)(r\Theta_\beta\widetilde{\nabla}_\lambda
 v^\beta)+(\varepsilon r\widetilde{\nabla}_\lambda
 v^3)(r\Theta_\beta\widetilde{\nabla}_\lambda
 w^\beta)+(\varepsilon r\widetilde{\nabla}_\lambda
 w^3)(r\varepsilon\widetilde{\nabla}_\lambda
 v^3)]\rm{d}x\\
 \qquad=\nu\int\limits_d[\delta_{\alpha\beta}\widetilde{\nabla}_\lambda
 w^\alpha\widetilde{\nabla}_\lambda
 v^\beta+(r\Theta_\alpha\widetilde{\nabla}_\lambda
 w^\alpha)(r\Theta_\beta\widetilde{\nabla}_\lambda
 v^\beta)\\
 \qquad+(r\Theta_\beta\widetilde{\nabla}_\lambda
 w^\beta)(r\varepsilon\widetilde{\nabla}_\lambda
 v^3)+(r\varepsilon\widetilde{\nabla}_\lambda
 w^3)(r\Theta_\beta\widetilde{\nabla}_\lambda
 v^\beta)+(r\varepsilon\widetilde{\nabla}_\lambda
 w^3)(r\varepsilon(\widetilde{\nabla}_\lambda
 v^3)]\rm{d}x,
 \end{array}$$
By combining we have
$$a_0(\bm w,\bm v)=
\nu\int\limits_D[\delta_{\alpha\beta}\widetilde{\nabla}_\lambda
 w^\alpha\widetilde{\nabla}_\lambda
 v^\beta+(r\Theta_\alpha\widetilde{\nabla}_\lambda
 w^\alpha+r\varepsilon\widetilde{\nabla}_\lambda
 w^3)(r\Theta_\beta\widetilde{\nabla}_\lambda
 v^\beta+r\varepsilon\widetilde{\nabla}_\lambda
 v^3)]\rm{d}x\eqno{(9.9)}$$
hence, we get
 $$a_0(w,v)=a_0(v,w),\quad |a_0(w,v)|\leq
 \nu(1+r^2_1k_0^2)|w|_{1,D}|v|_{1,D}.$$
$$a_0(w,w)=\nu|w|^2_{1,D}+\nu(r\Theta_\alpha\widetilde{\nabla}_\lambda
 w^\alpha+r\varepsilon\widetilde{\nabla}_\lambda
 w^3)(r\Theta_\alpha\widetilde{\nabla}_\lambda
 w^\alpha+r\varepsilon\widetilde{\nabla}_\lambda
 w^3)\geq\nu|w|^2_{1,D},$$
 Then (9.7) is proved. Next we consider
 (9.8). Obviously $(iv)$ is valid, thus we just need to prove $(v)$.
 Indeed  from (3.26), (3.27) and (3.35) we assert
$$\left\{\begin{array}{ll}
L_\sigma(\bm w,\Theta)
=(r^{-1}\delta_{2\beta}\delta_{\alpha\sigma}+r^2\Theta_\sigma\Theta_{\alpha\beta}-r^2\Theta_\alpha
\Theta_{\beta\sigma})\partial_\beta
w^\alpha-(2r\varepsilon\Theta_\beta a_{2\sigma}+\varepsilon
r^2\Theta_{\beta\sigma})\partial_\beta
w^3\\
\qquad+(a_{\alpha\sigma}q^\alpha_m+r^2\varepsilon q^3_m)w^m,\\
L_3(\bm w,\Theta)=r^2\varepsilon\Theta_{\alpha\beta} \partial_\beta
w^\alpha-2\varepsilon^2r^3\Theta_2\Theta_\beta\partial_\beta
w^3+(\varepsilon r^2\Theta_\alpha
q^\alpha_m+r^2\varepsilon^2q^3_m)w^m,\\
C_\sigma(\bm w,\bm
\omega)=2r\omega(w^2\Theta_\sigma-\delta_{2\sigma}\Pi(\bm
w,\Theta)),\quad
C_3(\bm w,\bm \omega)=2r\varepsilon\omega w^2,\\
a_{\alpha\sigma}q^\alpha_\beta+r^2\varepsilon
q^3_\beta=r^{-2}\delta_{2\beta}\delta_{2\sigma}+\delta_{2\beta}\Theta_2\Theta_\sigma
+2a_{2\sigma}(\delta_{2\beta}|\widetilde{\nabla}\Theta|^2-a\Theta_2\Theta_\beta\\
\qquad+r\Theta_\beta\widetilde{\Delta}\Theta
-r\Theta_\lambda\Theta_{\lambda\beta})+aa_{22}\Theta_\sigma\Theta_\beta+2r\Theta_\sigma\Theta_{2\beta}
+r^2\Theta_\sigma\partial_\beta\widetilde{\Delta}\Theta,\\
a_{\alpha\sigma}q^\alpha_3+r^2\varepsilon
q^3_3=ra_{2\sigma}\widetilde{\Delta}\Theta+\varepsilon a\Theta_2(r\Theta_2-2a_{2\sigma}),\\
(\varepsilon r^2\Theta_\alpha
q^\alpha_\beta+r^2\varepsilon^2q^3_\beta)=\varepsilon(\Theta_2\delta_{2\beta}+\Theta_\beta)
+2r\varepsilon\Theta_{2\beta}+r^2\varepsilon\Theta_2(2\delta_{2\beta}|\widetilde{\nabla}\Theta|^2
-a\Theta_2\Theta_\beta)\\
\qquad+\varepsilon
r^3\Theta_2\Theta_\beta\widetilde{\Delta}\Theta-r^3\varepsilon\Theta_2\Theta_\lambda
\Theta_{\lambda\beta},\\
(\varepsilon r^2\Theta_\alpha
q^\alpha_3+r^2\varepsilon^2q^3_3)=r^2\varepsilon^2(r\widetilde{\Delta}\Theta-a\Theta_2),
 \end{array}\right.\eqno {(9.10)}$$
therefore,
$$\begin{array}{lcl}
L_j(\bm w,\Theta)w^j&=&
[(r^{-1}\delta_{2\beta}\delta_{\alpha\sigma}+r^2\Theta_\sigma\Theta_{\alpha\beta}-r^2\Theta_\alpha
\Theta_{\beta\sigma})\partial_\beta
w^\alpha-(2r\varepsilon\Theta_\beta a_{2\sigma}+\varepsilon
r^2\Theta_{\beta\sigma})\partial_\beta
w^3]w^\sigma\\
&&+[(r^{-2}\delta_{2\beta}\delta_{2\sigma}+\delta_{2\beta}\Theta_2\Theta_\sigma
+2a_{2\sigma}(\delta_{2\beta}|\widetilde{\nabla}\Theta|^2-a\Theta_2\Theta_\beta\\
&&+r\Theta_\beta\widetilde{\Delta}\Theta
-r\Theta_\lambda\Theta_{\lambda\beta})+aa_{22}\Theta_\sigma\Theta_\beta+2r\Theta_\sigma\Theta_{2\beta}
+r^2\Theta_\sigma\partial_\beta\widetilde{\Delta}\Theta)w^\beta\\
&&+(ra_{2\sigma}\widetilde{\Delta}\Theta+\varepsilon a\Theta_2(r\Theta_\sigma-2a_{2\sigma}))w^3]w^\sigma\\
&&+[r^2\varepsilon\Theta_{\alpha\beta}
\partial_\beta
w^\alpha-2\varepsilon^2r^3\Theta_2\Theta_\beta\partial_\beta
w^3]w^3\\
&&+[(\varepsilon(\Theta_2\delta_{2\beta}+\Theta_\beta)
+2r\varepsilon\Theta_{2\beta}+r^2\varepsilon\Theta_2(2\delta_{2\beta}|\widetilde{\nabla}\Theta|^2
-a\Theta_2\Theta_\beta)\\
&&+\varepsilon
r^3\Theta_2\Theta_\beta\widetilde{\Delta}\Theta-r^3\varepsilon\Theta_2\Theta_\lambda
\Theta_{\lambda\beta})w^\beta
+(r^2\varepsilon^2(r\widetilde{\Delta}\Theta-a\Theta_2))w^3]w^3
\end{array}$$
By simplifying we get
$$\begin{array}{lcl}
L_j(\bm w,\Theta)w^j&=&r^{-1}\frac{\partial}{\partial r}(w^\alpha
w^\alpha)+r^{-2}w^2w^2
+P^\beta_\alpha(w,\Theta)\partial_\beta w^\alpha+P^\beta_3(w,\Theta)\partial_\beta w^3\\
&&+Q_{\alpha\beta}(\Theta)w^\beta w^\alpha
+Q_{3\beta}(\Theta)w^\beta w^3 +Q_{33}(\Theta)w^3w^3
\end{array}$$
where
$$\left\{\begin{array}{ll}
P^\beta_\alpha(\bm w,\Theta)&
=r^2\varepsilon\Theta_{\alpha\beta}w^3+r^2(\Theta_\sigma\Theta_{\alpha\beta}
-\Theta_\alpha\Theta_{\beta\sigma})w^\sigma,\\
P^\beta_3(\bm w,\Theta&
=-\varepsilon[2r^3\varepsilon\Theta_2\Theta_\beta w^3+2
r^2(\Theta_{\beta\sigma}+r\Theta\Theta_2\sigma\Theta_\beta)w^\sigma,\\
Q_{\alpha\beta}(\Theta)&=\delta_{2\beta}(\Theta_2\Theta_\alpha+2a_{2\alpha}|\widetilde{\nabla}\Theta|^2)
+a\Theta_\beta(a_{22}\Theta_\alpha-2a_{2\alpha}\Theta_2)\\
&\quad+(2r\varepsilon\delta_{2\lambda}\Theta_\alpha
+2a_{2\alpha}\Theta_\lambda)\Theta_{\lambda\beta}+r\Theta_\beta(2a_{2\alpha}+r\Theta_\alpha)\widetilde{\Delta}\Theta,\\
Q_{3\beta}(\Theta)&=\varepsilon[
r(2a_{2\beta}-\delta_{2\beta})\widetilde{\Delta}\Theta-r^3\Theta_2\Theta_\lambda\Theta_{\lambda\beta}
+3r
a\Theta_2\Theta_\beta+\Theta_\beta\\
&\quad-((1+a)\delta_{2\beta}+aa_{2\beta})\Theta_2],\\
Q_{33}(\Theta)&=\varepsilon^2r^2(r\widetilde{\Delta}\Theta-a\Theta_2),\\
\end{array}\right.\eqno{(9.11)}$$
From the assumptions of the Lemma, we clare that there exist two
constants $C_i(k_0),i=0,1$ independent of $\bm w,\Theta$ such that
$$(\bm L(\bm w,\Theta),\bm w)\geq \int\limits_D\frac1r\partial_r(w^\alpha w^\alpha)\rm{d}x-C_0(k_0)k_0\|\bm w\|_{0,D}-C_1(k_0)k_0|\bm w|_{1,D}^2$$
In addition
$$(\bm C(w,\omega),\bm w)=2\int\limits_D(\omega\times
\bm w)\bm w\rm{d}x=0,$$
 Then, from (9.6) and $\alpha_\tau=\frac{\nu
a}{\gamma^2\varepsilon^2\tau^2}\geq
\frac{\nu}{r^2_1\varepsilon^2\tau^2}$, we get
$$\begin{array}{ll}
A_0(\bm w,\bm w)&=(\alpha_\tau \bm w,\bm w)+a_0(\bm w,\bm w)+(\bm
L(\bm w,\Theta),\bm w)
+(\bm C(\bm w,\omega),\bm w)\\
&\geq (\frac{\nu}{r^2_1\varepsilon^2\tau^2}-C_0(k_0)k_0)\|\bm
w\|^2_{0,D}+(\nu-C_1(k_0)k_0)|\bm w|^2_{1,D}
+\int\limits_D\frac1r\partial_r(w^\alpha w^\alpha)\rm{d}x,
\end{array}$$the proof is ended.
\end{proof}

\begin{lemma}
Under the assumptions in Lemma 9.1, the trilinear form
$b_0(\cdot,\cdot,\cdot)$ is continuous, i.e., there exists a
constant $M(\Theta,D)$ independent of $\bm w,\textbf{u},\bm v$, such
that
$$ |b(\bm w,\textbf{u},\bm v)|\leq
M\|\bm w\|_{H^\frac56(D)}\|\bm
v\|_{H^\frac56(D)}\|\textbf{u}\|_{1,D},\forall\,\bm w,\textbf{u},
\bm v\,\in V(D) \eqno{(9.12)}$$
\end{lemma}

\begin{proof} Thanks to the H\"{o}lder inequality
$$\int_D|w^\lambda\widetilde{\nabla}_\lambda u^\alpha v^\beta|\sqrt{a}\rm{d}x\leq
\|\bm w\|_{L^4(D)}\|v^\beta\|_{L^4(D)}\|\widetilde{\nabla}
u^\alpha\|_{0,D}$$ and the Sobolev embedding theorems
$$\|\bm u\|_{L^4(D)}\leq C\|\bm u\|_{H^{\frac56}(D)},\quad
\|\bm u\|_{L^3(\gamma_1)}\leq C\|\bm
u\|_{H^\frac56(D)}.\eqno{(9.13)}$$ Further, by using the Cauchy's
inequality, we derived the conclusion. The proof is completed.
\end{proof}

\begin{remark} It is clear that we have
$$\begin{array}{ll}
a_{\alpha\beta}w^\lambda\stackrel{\ast}{\nabla}_\lambda w^\alpha
w^\beta=\stackrel{\ast}{\nabla}_\lambda(a_{\alpha\beta}w^\lambda
u^\alpha
v^\beta)-a_{\alpha\beta}w^\alpha\stackrel{\ast}{\nabla}_\lambda(w^\lambda
w^\beta)\\
\qquad=\stackrel{\ast}{\div}(|\bm w|\bm w )-|\bm
w|\stackrel{\ast}{\div}\bm w-a_{\alpha\beta}w^\beta
w^\lambda\stackrel{\ast}{\nabla}_\lambda w^\alpha,\end{array}$$
Hence
$$a_{\alpha\beta}w^\lambda\stackrel{\ast}{\nabla}_\lambda w^\alpha
w^\beta=\frac12\stackrel{\ast}{\div}(|\bm w|\bm w )-\frac12|\bm
w|\stackrel{\ast}{\div}\bm w.$$  Considering (6.6) and the boundary
conditions of element in $V(D)$, we claim
$$|b_0(\bm w_0,\bm w_0,\bm w_0)|\leq C (\|\bm w_0\|^2_{L^4(D)}\|\stackrel{\ast}{\div}\bm w\|^2_{0,D}+\|\bm w_0\|^3_{L^3(\gamma_1)})
,\eqno{(9.14)}$$where the Gauss theorem is used.
\end{remark}

\begin{theorem} Under the assumptions in Lemma 9.1, for the given $(G, d^3_0)\in
\,V^*(D)\times H^{-1}(D)$, if $\bm F$ satisfies the following
condition,
$$\begin{array}{ll}
\|\bm F\|_*\leq \frac {\nu^2\lambda^2}{MC^2},\quad  \text{with}\,
<\bm F,\bm v>=<\bm G,\bm v>-\eta^{-1}(d^3_0,\stackrel{\ast}{\div}\bm
v),
\end{array}\eqno{(9.15)}$$
then there exists one solution $w_*$ of the variational problem
(6.10) which satisfies
$$\|\bm w_*\|_{1,D}\leq \rho:=
\frac{\nu\lambda}{MC}-\sqrt{(\frac{\nu\lambda}{MC})^2-\frac{\|\bm
F\|_*}{M}}.\eqno{(9.16)}$$ Furthermore, if
$$\begin{array}{ll}
\|\bm F\|_*< \frac {\nu^2\lambda^2}{MC^2},\\
\end{array}\eqno{(9.17)}$$  then problem (6.10) has a unique
solution in $V(D)$.
\end{theorem}

\begin{proof} We begin with constructing a sequence of approximate solutions
by Galerkin's method. Since the space $V(D)$ is separable, there
exists a sequence $(\bm\varphi_m, m\geq
 1)$ in $V(D)$ such that: 1).
for all $m\geq1$, the elements $\bm\varphi_1,\cdots
 \bm\varphi_m$ are linearly independent; 2). the  finite linear combinations
$\sum\limits_ic_i\bm\varphi_i$ are dense in $V(D)$. Such a sequence
$(\bm\varphi_m, m\geq 1)$ is called a basis of the separable space
$V(D)$.

Next we use  $V_m$ to denote the subspace of $V(D)$ spanned by
finite sequence $\bm\varphi_1,\cdots, \bm\varphi_m$. Then, we can
construce the approximating problem,
$$\left\{\begin{array}{ll}
\mbox{Find}\,\bm w_m\in\,V_m\, \mbox{such that}\\
 A_0(\bm w_m,\bm v)+b_0(\bm w_m,\bm w_m,\bm v)=<\bm F,\bm v>,\quad \forall \,\bm v\in\,V_m.
\end{array}\right.\eqno{(9.18)}$$
If we set
$$\bm w_m=\sum\limits_{i=1}^mc_i\bm\varphi_i,$$
then we find that problem (9.18) amounts to solve a system of m
nonlinear equations with $m$ unknowns $c_i$. For each $m$ problem
(9.18) has at least one solution. Indeed, we can introduce the
mapping ${\cal M}_m$ : $V_m\rightarrow V_m$,\
$$({\cal
M}_m(\bm u),\bm \varphi_i)=A_0(\bm u,\bm \varphi_i)+b_0(\bm u,\bm
u,\bm \varphi_i)-<\bm F,\bm \varphi_i>,\quad 1\leq i\leq m,$$ where
$(\cdot,\cdot)$ is the scalar product in $V$. Hence, $\bm
w_m\in\,V_m$ is a solution of problem (9.18) if only if ${\cal
M}_m(\bm w_m)=0$. Since
$$({\cal M}_m(\bm u),\bm u)=A_0(\bm u,\bm u)+b_0(\bm u,\bm u,\bm u)-<\bm F,\bm u>,\quad \forall
\,\bm u\in\,V_m,$$ it follows  that
$$\begin{array}{ll}
({\cal M}_m(\bm u),\bm u)\geq (\frac{2\nu\lambda}{C}\|\bm
u\|_{1,D}-M\|\bm u\|^2_{1,D}-\|\bm F\|_*)\|\bm u\|_{1,D}
\end{array}\eqno{(9.19)}$$  Hence, if choosing
$$\rho=\frac{\nu\lambda}{MC}-\sqrt{(\frac{\nu\lambda}{MC})^2-\frac{\|\bm F\|_*}{M}},$$
then we get for all $\bm u\in\,V_m$ with $\|\bm u\|_{1,D}=\rho$,
$$({\cal M}_m(\bm u),\bm u)\geq 0.$$
 Moreover, ${\cal M}_m$ is continuous in
$V_m$, and the space $V_m$ is finite dimensional, we can apply
Corollary 1.1 in [17], there exists at least one solution $\bm
w_m\in\,V_m$ of problem (9.18).

Furthermore, we have for any solution $\bm w_m$ to (9.18)
$$0=({\cal M}_m(\bm w_m),\bm w_m)\geq
(\frac{2\nu\lambda}{C}\|\bm w_m\|_{1,D}-M\|\bm w_m\|^2_{1,D}-\|\bm
F\|_*)\|\bm w_m\|_{1,D},$$ therefore,
$$\frac{2\nu\lambda}{C}\|\bm w_m\|_{1,D}-M\|\bm w_m\|^2_{1,D}-\|\bm F\|_*\leq 0.$$
It follows that, when denoting by $y=\|\bm w_m\|_{1,D}$,
$$(y-\frac{\nu\lambda}{MC})^2\geq
(\frac{\nu\lambda}{MC})^2-\frac{\|\bm F\|_*}{M}\Rightarrow
y-\frac{\nu\lambda}{MC}\leq-\sqrt{(\frac{\nu\lambda}{MC})^2-\frac{\|\bm
F\|_*}{M}},$$ i.e.,
$$\|\bm w_m\|_{1,D}\leq\frac{\nu\lambda}{MC}-\sqrt{(\frac{\nu\lambda}{MC})^2-\frac{\|\bm F\|_*}{M}}.\eqno{(9.20)}$$
This shows that the sequence $(\bm w_m)$   is uniformly bounded in
$V$. Therefore, we can extract a subsequence, still denoted by $\bm
w_{m})$, such that
$$\bm w_{m}\rightharpoonup (\mbox{weak})\,\bm w_*
\,\mbox{in}\,V(D)\,\mbox{as}\,m\rightarrow+\infty.$$ Then, the
compactness of the embedding of $V(D)$ into $L^2(D)^3$ implies that
$$\bm w_m\rightarrow (\mbox{strong})\, w_*\,\mbox{in}\,
L^2(D)^3\,\mbox{as}\,m\rightarrow\,+\infty,$$

the remainder is to prove $b(\cdot,\cdot,\cdot)$ is weakly sequence
continuous, i.e., $b_0(\bm w_m,\bm w_m,\bm v)\rightarrow b_0(\bm
w_*,\bm w_*,\bm v)$. To do this, we recall
$${\cal V}=\{\bm u\in\,C^\infty(D) \,\mbox{satisfy the boundary conditon (3.36) }\}
$$ is dense in $V(D)$ and
$$b_0(\bm w_m,\bm w_m,\bm v)=\int_{\gamma_1}a_{\alpha\beta}w_m^\alpha v^\beta
a_{\lambda\sigma}w^\lambda_mn^\sigma dl-b_0(\bm w_m,\bm v,\bm
w_m).$$ For any  $\bm v\in\,{\cal V}$, then $\bm
v\in\,L^\infty(D)\,L^\infty(\gamma_1)$,
$\partial_{x^\alpha}v^\beta\,\in\,L^\infty(D)$, and the convergence
relations $\lim\limits_{m\rightarrow \infty}w_m^\lambda
w^\sigma_m=w^\lambda_*w^\sigma_*$ are satisfied in $L^1(D)$ and
$L^1(\gamma_1)$ respectively, therefore
$$\begin{array}{ll}
\lim\limits_{m\rightarrow\infty}b_0(\bm w_m,\bm w_m,\bm v)
=\int_{\gamma_1}a_{\alpha\beta}w_*^\alpha v^\beta
a_{\lambda\sigma}w^\lambda_*n^\sigma dl-b_0(\bm w_*,\bm v,\bm
w_*)=b_0(\bm w_*,\bm w_*,\bm v) ,\quad \forall  \bm v\in\,{\cal V}.
\end{array}$$
Next  for all $\bm v\,\in\,V(D)$, by virtue of the density of ${\cal
V}$, and taking the limitation of  both sides of (9.18) implies
$$\begin{array}{ll}
 A_0(\bm w_*,\bm v)+b_0(\bm w_*,\bm w_*,\bm v)=<\bm F,\bm v>,\quad \forall \,\bm
 v\in\,V(D),
\end{array}\eqno{(9.21)}$$
that means $\bm w_*$ is a solution of problem (6.10).

In order to prove (9.16), by a similar manner we get
$$|\|\bm w_*\|_{1,D}-\frac{\nu\lambda}{MC}|\geq
\sqrt{(\frac{\nu\lambda}{MC})^2-\frac{\|\bm F\|_*}{M}},$$因此
$$\|\bm w_*\|_{1,D}\geq\frac{\nu\lambda}{MC}+\sqrt{(\frac{\nu\lambda}{MC})^2-\frac{\|\bm F\|_*}{M}},\,
\mbox{or}\quad \|\bm
w_*\|_{1,D}\leq\frac{\nu\lambda}{MC}-\sqrt{(\frac{\nu\lambda}{MC})^2-\frac{\|\bm
F\|_*}{M}}.$$ Obviously, the first one of the above inequalities  is
contract to (9.20), thus only the second one is true, that is
(9.16).

Next we prove the uniqueness. In fact, if there exist two solutions
$\bm w_*$ and $\widetilde{\bm w}_*$  of (6.10). Let $\bm e_*=\bm
w_*- \widetilde{\bm w}_*$, then
$$A_0(\bm e_*,\bm e_*)+b_0(\bm e_*,\bm w_*,\bm e_*)+b_0(\widetilde{\bm w}_*,\bm e_*,\bm e_*)=0.$$
Owing to condition satisfied by $w_*$ and $\widetilde{w}_*$ and
(6.17), we get
$$0\geq (\frac{2\nu\lambda}{C}-2M(\frac{\nu\lambda}{MC}-\sqrt{(\frac{\nu\lambda}{MC})^2
-\frac{\|\bm F\|_*}{M}})\|\bm
e_*\|_{1,D}^2=2\sqrt{(\frac{\nu\lambda}{MC})^2 -\frac{\|\bm
F\|_*}{M}}.$$ This yields $\|\bm e_*\|_{1,D}=0$. Therefor, the
solutions is unique. The proof is complete.
\end{proof}

\begin{theorem} Let $(\bm w_0,p_0)$  and $(\bm w_\eta,p_\eta)$ be the solutions of
(6.10) and (9.3), respectively. If $\bm F$ and
$H_*=\sup\limits_D|H|$ satisfy the condition
$$\begin{array}{ll}
C_0-2M\rho-2\eta^{-1}H_*^2\geq C_2>0,
\end{array}\eqno{(9.22)}$$
then the following estimates are valid
$$\|\bm w_0-\bm w_\eta\|_{1,D}+\|p_0-p_\eta\|_{0,D}\leq \max(C_3,C_4)\eta,\eqno{(9.23)}$$
where
$$C_3=\frac{C+2M\rho}{C_2\beta_0}\|p_0\|_{0,D},\quad
C_4=\frac{(C+2M\rho)^2}{C_2\beta_0^2}\|p_0\|_{0,D}.\eqno{(9.24)}$$
and $\beta_0$ is the constant in the $\inf-\sup$ condition.
\end{theorem}

\begin{proof} From the assumption we have
$$\left\{\begin{array}{ll}
\mbox{Find}\bm w_0\in\,V(D),p_0\in\,L^2(D),\ \mbox{such that}\\
a_0(\bm w_0,v)-(p_0,\stackrel{\ast}{\div}\bm v)+b_0(\bm w_0,\bm w_0,\bm v)+(l(\bm w_0),\bm v)=<\bm G,\bm v>,\quad \forall \bm v\in\,V(D),\\
(\stackrel{\ast}{\div}\bm w_0-2Hw^3_0+d^3_0,q)=0,\quad
\forall\,q\in\,L^2(D),
\end{array}\right.\eqno{(9.25)}$$
and
$$\left\{\begin{array}{ll}
\mbox{Find}\bm w_\eta\in\,V(D),p_\eta\in\,L^2(D),\ \mbox{such that} ,\\
a_0(\bm w_\eta,v)-(p_\eta,\stackrel{\ast}{\div}\bm v)+b_0(\bm w_\eta,\bm w_\eta,\bm v)+(l(\bm w_\eta),\bm v)=<\bm G,\bm v>,\quad \forall \bm v\in\,V(D)\\
\eta(p_\eta,q)+(\stackrel{\ast}{\div}\bm
w_\eta-2Hw^3_\eta+d^3_0,q)=0.\forall\,q\in\,L^2(D).
\end{array}\right.\eqno{(9.26)}$$
Next, we denote $e_*=w_0-w_\eta, s_*=p_0-p_\eta$. subtracting (9.25)
from (9.26)then yields
$$\left\{\begin{array}{ll}
a_0(\bm e_*,v)+b_0(\bm e_*,\bm w_0,\bm v)+b_0(\bm w_\eta,\bm e_*,\bm v)-(s_*,\stackrel{\ast}{\div}\bm v)=0,\forall\,\bm v\in\,V(D),\\
(\stackrel{\ast}{\div}\bm
e_*,q)+(-2He_*^3,q)+\eta(s_*,q)-\eta(p_0,q)=0,\forall\,q\in\,L^2(D).
\end{array}\right.\eqno{(9.27)}$$
choosing $\bm v=\bm e_*, q=s_*$ in (9.27) and summing the above two
equations, we obtain
$$a_0(\bm e_*,\bm e_*)+b_0(\bm e_*,\bm w_0,\bm e_*)+b_0(\bm w_\eta,\bm
e_*,\bm e_*)+\eta(s_*,s_*)-\eta(p_0,s_*)-( 2H\bm e_*^3,s_*) =0.
\eqno{(9.28)}$$ Noting that (9.7), (9.8), and (9.12), we have
$$\begin{array}{ll}
(C_0-2M\rho)\|e_*\|^2_{1,D}+\eta\|s_*\|^2_{0,D}\leq
(\eta\|p_0\|_{0,D}+2\sup\limits_D|H|\|e_*\|_{0,D})\|s_*\|_{0,D}.
\end{array}\eqno{(9.29)}$$
Furthermore by using Young's inequality we get
$$2\sup\limits_D|H|\|e_*\|_{0,D}\|s_*\|_{0,D}\leq
\frac12\eta\|s_*\|_{0,D}^2+2\eta^{-1}H_*^2\|e_*\|_{1,D}^2,\eqno{(9.30)}$$
therefore from (9.29),(9.30)
$$\begin{array}{ll}
(C_0-2M\rho-2\eta^{-1}H_*^2)\|e_*\|^2_{1,D}+\frac12\eta\|s_*\|^2_{0,D}\leq
\eta\|p_0\|_{0,D}\|s_*\|_{0,D},\\
\end{array}$$
Noting that the condition (9.22) are satisfied, hence
$$\|\bm e_*\|_{1,D}^2\leq\eta \frac{\|p_0\|_{0,D}}{C_2}\|s_*\|_{0,D}.\eqno{(9.31)}$$
On the other hand, the $\inf-\sup$ condition means that
$$\begin{array}{ll}
\beta_0\|s_*\|_{0,D}&\leq \sup\limits_{\bm
v\in\,V(D)}\frac{|(s_*,\div
\bm v)|}{\|\bm v\|_{1,D}}\\
&\leq\sup\limits_{\bm v\in\,V(D)}(\|\bm v\|^{-1}_{1,D}|[a_0(\bm e_*,\bm v)+b_0(\bm e_*,\bm w_0,\bm v)+b_0(\bm w_\eta,\bm e_*,\bm v)]|\\
&\leq (C+2M\rho)\|\bm e_*\|_{1,D}.
\end{array}\eqno{(9.32)}$$
Finally from (9.31), (9.32) we have
$$\begin{array}{ll}
\|\bm e_*\|_{1,D}\leq C_3\eta^2,\quad \|s_*\|_{0,D}\leq C_4\eta,
\end{array}$$ where
$$C_3=\frac{C+2M\rho}{C_2\beta_0}\|p_0\|_{0,D},\quad
C_4=\frac{(C+2M\rho)^2}{C_2\beta_0^2}\|p_0\|_{0,D}.$$ Thus the
theorem is proved
\end{proof}

\section {\bf Finite Element Approximation Based on Approximate Inertial
Manifold}

In this section, we focus on the variational problem for the 2D-3C
problem (6.10), called 2D-3C variational problem,
$$\left\{\begin{array}{ll}
\mbox{Find}\quad \bm w_0\in\,V(D),\ \mbox{such that}\\
A_0(\bm w_0,\bm v)+b_0(\bm w_0,\bm w_0,\bm v)=<\bm G_\eta,\bm
v>,\quad \forall \bm v\in\,V(D),
\end{array}\right.\eqno{(10.1)}$$
where
$$\left\{\begin{array}{ll}
A_0(\bm w_0,\bm v)&=a_0(\bm w_0,\bm
v)+\eta^{-1}(\stackrel{\ast}{\div}\bm w_0,\stackrel{\ast}{\div}\bm
v)-\eta^{-1}
(2Hw^3_0,\stackrel{\ast}{\div}\bm v) +(l_0(\bm w_0),\bm v)\\
&=a_0(w_0,v)+\eta^{-1}(\stackrel{\ast}{\div}w_0,\stackrel{\ast}{\div}v)+(\nu-\eta^{-1})
(2Hw^3_0,\stackrel{\ast}{\div}\bm v)\\
&+(a_{\alpha\beta}(C^\alpha_\lambda
w^\lambda_0+C^\alpha_3w^3_0),v^\beta)+(C^3_\beta w^\beta_0,v^3),\\
a_0(\bm w_0,\bm
v)&=2\nu(a^{\alpha\lambda}a^{\beta\sigma}\gamma_{\lambda\sigma}(\bm
w_0),\gamma_{\alpha\beta}(\bm v))
+\nu(a^{\alpha\beta}\stackrel{\ast}{\nabla}_\alpha
w^3_0,\stackrel{\ast}{\nabla}_\beta
v^3)\\
&+\nu h^{-2}[(a_{\alpha\beta}w^\alpha_0,v^\beta)
+(w^3_0,v^3)],\\
b_0(\bm w_0,\bm w_0,\bm
v)&=(a_{\alpha\beta}w^\lambda_0\stackrel{\ast}{\nabla}_\lambda
w^\alpha_0-2b_{\alpha\beta} w^\alpha_0w^3_0 ,v^\beta)
+(w^\beta\stackrel{\ast}{\nabla}_\beta
w^3_0+b_{\alpha\beta}w^\alpha_0w^\beta_0,v^3),\\
(l(\bm w_0),\bm v)&=(a_{\alpha\beta}l^\alpha(w_0),v^\beta)+(l^3(w_0),v^3)\\
          &=(2\nu H\stackrel{\ast}{\nabla}_\beta w^3_0 ,v^\beta)+(a_{\alpha\beta}(C^\alpha_\lambda
          w^\lambda_0+C^\alpha_3w^3_0),v^\beta)+(C^3_\beta
          w^\beta_0,v^3),\\
<\bm G_\eta,\bm v>&=<\bm F_h,\bm v>
+\int_{\gamma_0}[\sigma_{n\alpha} v^\alpha+\sigma_{n3}v^3
]d\gamma+(a_{\alpha\beta}d^\alpha_0m^3_0,v^\beta)-\eta^{-1}(d^3_0,\stackrel{\ast}{\div}v)
\end{array}\right.\eqno{(10.2)}$$

We now consider the finite element approximation of the 2D-3C
variational problem (10.1). Assume that $V_h$ and $M_h$ are finite
element subspaces of $V(D)$ and $L^2(D)$ respectively. Introduce the
product space $Y_h=V_h\times M_h$, obviously which is a subspace of
$Y=V(D)\times L^2(D)$.

Then the standard Galerkin finite element approximation of (10.1) is
defined by
$$\left\{\begin{array}{l}
\mbox{Find} \bm w_h\in V_h,\  \mbox{such that}\\
A_0(\bm w_h,\bm v)+b_0(\bm w_h,\bm w_h,\bm v)=(\bm G_h,\bm v)\quad
   \forall \bm v\in V_h \end{array}\right.\eqno{(10.3)}$$

As usual, we make the following standard assumptions on the finite
element subspace $Y_h$
\begin{description}
\item[{\bf (H1)}] Approximation property
$$  \inf\limits_{(\bm v_h,q_h)\in Y_h}\{h\|\bm u-\bm v_h\|_{1,D}+\|\bm u-\bm v_h\|_{0,D}+h\|p-q_h\|_{0,D}\}\leq
 C\,h^{k+1}\{\|\bm u\|_{k+1,D}+\|p\|_{k,D}\}$$
for any $(\bm u, p)\in Y\cap (H^{k+1}(\Omega)^d\times
H^k(\Omega)),1\leq k\leq l$.

\item[{\bf (H2)}] Interpolation property
$$\|\bm v-I_h\bm v\|_{1,D}  +\|q-J_hq\|_{0,D}\leq C\,h^k(\|\bm v\|_{k+1,D}+\|q\|_{k,D})$$
for any $(\bm v,q)\in Y\cap (H^{k+1}(\Omega)^d\times
H^k(\Omega)),1\leq k\leq l$, where $I_h$ and $J_h$ are some
interpolation operators from $V(D)$ and $L^2(D)$ into $X_h$ and
$M_h$, respectively.

\item[{\bf (H3)}] Inverse inequality
        $$\|\bm v_h\|_{1,D}\leq C\,h^{-1}\|\bm v_h\|_{0,D}\quad \forall\, \bm v_h\in V_h. $$

\item[{\bf (H4)}] $(V_h,M_h)$ satisfies LBB-condition
$$\inf\limits_{q\in M_h}\sup\limits_{\bm v\in X_h}\frac{(q,\div
\bm v)}{\|q\|_{0,D}\,\|\bm v\|_{1,D}} \geq\beta>0,$$ where $\beta$
is a constant independent of $h$.
\end{description}

The following optimal error estimates of the Galerkin finite element
approximation are well-known (cf.[17]),

\begin{theorem} Suppose $\bm w_0\in V(D)\cap
H^{k+1}(D)^3$ is a nonsingular solution of (10.1) and the finite
element subspace $V_h$ satisfies assumptions $(H1)\sim (H4)$. Then
there exists a solution $\bm w_h$ satisfying (10.3) such that
$$\begin{array}{ll}
   h\|\bm w_0-\bm w_h\|_{1,D}+\|\bm w_0-\bm w_h\|_{0,D}\leq C\,h^{k+1}\|\bm w_0\|_{k+1,D}.
\end{array}\eqno{(10.4)}$$
\end{theorem}

Next we try to improve the error estimation. To do that we rewrite
(10.1) in operator form. Let  ${\cal F}:\ V(D)\longrightarrow
V^*(D)$ denote the 2D-3C Navier-Stokes operator  on the manifold
$\Im$ via
$$<{\cal F}(\bm w_0),\bm v>:=A_0(\bm w_0,\bm v)+b_0(\bm w_0,\bm w_0,\bm v)-<\bm G_\eta,\bm v>,\quad
\forall\,\bm v\in\,V(D).$$
 It is obvious that ${\cal F}(\bm w_0)=0$  is equivalent to (10.1). The operator form of finite element
approximation (10.3) is
$$<{\cal F}_h(\bm w_h),\bm v>:=A_0(\bm w_h,\bm v)+b_0(\bm w_h,\bm w_h,\bm v)-<\bm G_\eta,\bm v>,\quad
\forall\,\bm v\in\,V_h(D).$$ Therefore, ${\cal F}_h(w_h)=0$ is
equivalent to (10.3). Furthermore, it is easy to show that ${\cal
F}(\bm w_0)$ and ${\cal F}_h(\bm w_h)$  are Fr\'{e}chet
differentiable and the Fr\'{e}chet derivatives at $\bm w_0$ and $\bm
w_h$ along direction $\bm u$ are given by, respectively
$$\begin{array}{ll}
{\cal A}_{\bm w_0}(\bm u,\bm v):=(D{\cal
F}(\bm w_0)\bm u,\bm v)=A_0(\bm u,\bm v)+b_0(\bm u,\bm w_0,\bm v)+b_0(\bm w_0,\bm u,\bm v),\quad \forall\, \bm u, \bm v\in\,V(D),\\
{\cal A}_{\bm w_h}(\bm u,\bm v):=(D{\cal F}_h(\bm w_h)\bm u,\bm
v)=A_0(\bm u,\bm v)+b_0(\bm u,\bm w_h,\bm v)+b_0(\bm w_h,\bm u,\bm
v),\quad \forall\,\bm u,\bm v\in\, V_h(D).
\end{array}$$
It is well known that $\bm w_0$ is a nonsingular solution of  (10.1)
if and only if $D{\cal F}(\bm w_0)$ is an isomorphism on $V(D)$,
furthermore, equivalent to ${\cal A}_{\bm w_0}(\cdot,\cdot)$
satisfies the $\inf-\sup $ condition(weak coerciveness), i.e.,
$$\begin{array}{ll}
\inf\limits_{\bm u\in\,V(D)}\sup\limits_{\bm v\in\,V(D)}\frac{{\cal
A}_{\bm w_0}(\bm u,\bm v)}{\|\bm u\|_{1,D}\|\bm v\|_{1,D}}\geq
\alpha_0>0,\quad \inf\limits_{\bm v\in\,V(D)}\sup\limits_{\bm
u\in\,V(D)}\frac{{\cal A}_{\bm w_0}(\bm u,\bm v)}{\|\bm
u\|_{1,D}\|\bm v\|_{1,D}}\geq \alpha_0>0.
\end{array}\eqno{(10.5)}$$
 In this case, for any $\bm f\in\,V^*(D)$, the variational problem
$$\left\{\begin{array}{ll}
\mbox{Find }\,\bm u\in\,V(D)\,\mbox{such that}\\
{\cal A}_{\bm w_0}(\bm u,\bm v)=<\bm f,\bm v>,\quad\forall\,\bm
v\in\,V(D),
\end{array}\right.\eqno{(10.6)}$$
has a one and only one solution. Similarly, $\bm  w_h$ is a
nonsingular solution of (10.3) if and only if $D{\cal F}_h(\bm w_h)$
is an isomorphism on $V_h(D)$, equivalent to ${\cal A}_{\bm
w_h}(\cdot,\cdot)$ satisfies the $\inf-\sup $ condition(weak
coerciveness)
$$\begin{array}{ll}
\inf\limits_{\bm u\in\,V_h(D)}\sup\limits_{\bm
v\in\,V_h(D)}\frac{{\cal A}_{\bm w_h}(\bm u,\bm v)}{\|\bm
u\|_{1,D}\|\bm v\|_{1,D}}\geq \alpha_h>0,\quad \inf\limits_{\bm
v\in\,V_h(D)}\sup\limits_{\bm u\in\,V_h(D)}\frac{{\cal A}_{\bm
w_h}(\bm u,\bm v)}{\|\bm u\|_{1,D}\|\bm v\|_{1,D}}\geq \alpha_h>0,
\end{array}\eqno{(10.7)}$$
 In this case, the variational problem
$$\left\{\begin{array}{ll}
\mbox{Find}\,\bm u_h\in\,V_h(D)\,\mbox{such that}\\
{\cal A}_{\bm w_h}(\bm u_h,\bm v)=<\bm f,\bm v>,\quad\forall\,\bm
v_h\in\,V_h(D),
\end{array}\right.\eqno{(10.8)}$$
has a unique solution $w_h$ for any $\bm f\in\,V_h^*(D)$.
 Condition (10.5) is equivalent to
  $$\|D{\cal F}(\bm w_0)\|_{{\cal L}(V,V)}\leq \alpha^{-1}_0.\eqno{(10.9)}$$

The next theorem shows the uniqueness condition to insure the finite
element solution $\bm w_h$ of (10.3).

\begin{theorem} Assume that the assumptions $(H1)\sim (H4)$ are valid, and $\bm
w_0$ is a nonsingular solution of (10.1). If the finite element mesh
$h$ is small enough such that
$$2MC\alpha_0^{-1}\|\bm w_0\|_{2,D}h<1.\eqno{(10.10)}$$ Then, solution
$\bm w_h$ of the finite element approximation problem (10.3) is
nonsingular.
\end{theorem}

\begin{proof} In fact, from the above explanation  in this section, it is enough to
prove that
$$\|D{\cal F}_h(\bm w_h)\|_{{\cal L}(V_h,V_h)}\leq \beta^{-1}_0.\eqno{(10.11)}$$
Therefore, noting that
$$\begin{array}{ll}
\varepsilon:&=\|D{\cal F}(\bm w_0)-D{\cal F}_h(\bm w_h)\|_{{\cal
L}(V,V)}\\
&=\sup\limits_{\bm u,\bm v\in\,V_h}\frac{((D{\cal F}(\bm w_0)-D{\cal
F}_h(\bm w_h))\bm u,\bm v)}{\|\bm u\|_{1,D}\|\bm v\|_{1,D}}\\
&=\sup\limits_{\bm u,\bm v\in\,V}\frac{b_0(\bm w_0-\bm w_h,\bm u,\bm v)+b_0(\bm u,\bm w_0-\bm w_h,\bm v)}{\|\bm u\|_{1,D}\|\bm v\|_{1,D}}\\
&\leq 2M\|\bm w_0-\bm w_h\|_{1,D}\leq2MC\|\bm w_0\|_{2,D}h.
\end{array}\eqno{(10.12)}$$
Set $B=\{D{\cal F}(\bm w_0)\}^{-1}\{D{\cal F}(\bm w_0)-D{\cal
F}_h(\bm w_h)\}.$ Then from (10.9) and (10.12), we have
$$\begin{array}{ll}
D{\cal F}_h(\bm w_h)=D{\cal F}(\bm w_0)(I-B),\\
\|B\|_{{\cal L}(V,V)}\leq \alpha_0^{-1}2MC\|\bm w_0\|_{2,D}h,\quad
\|(I-B)^{-1}\|_{{\cal
L}(V,V)}\leq\frac1{1-2MC\alpha_0^{-1}\|\bm w_0\|_{2,D}h},\\
\|D{\cal F}_h(\bm w_h)\|_{{\cal
L}(V,V)}\leq\frac{1}{\alpha_0}\frac1{1-2MC\alpha_0^{-1}\|w_0\|_{2,D}h}.
\end{array}$$
substituting (10.10) into the above inequality, then we get
 $D{\cal F}_h(\bm w_h)$ is an isomorphism on $V_h$, hence $\bm w_h$ is a nonsingular solution of
(10.3).
\end{proof}

Theorem 10.2 shows if mesh size $h$ is small enough, then we have
$$\begin{array}{ll}
\inf\limits_{\bm u\in\,V_h(D)}\sup\limits_{\bm
v\in\,V_h(D)}\displaystyle\frac{{\cal A}_{\bm w_h}(\bm u,\bm
v)}{\|\bm u\|_{1,D}\|\bm v\|_{1,D}}\geq \frac12\alpha_0>0,\quad
\inf\limits_{\bm v\in\,V_h(D)}\sup\limits_{\bm
u\in\,V_h(D)}\frac{{\cal A}_{\bm w_h}(\bm u,\bm v)}{\|\bm
u\|_{1,D}\|\bm v\|_{1,D}}\geq \frac12\alpha_0>0.
\end{array}\eqno{(10.13)}$$
Next, assume $\bm w_h$ is a nonsingular solution (10.3). We define a
projection  $P_h:V(D)\rightarrow V_h(D), \forall \bm w\in\,V(D)$
through
$${\cal A}_{\bm w_h}(\bm w-P_h\bm w,\bm v)=0,\forall\,\bm
v\in\,V_h(D).\eqno{(10.14)}$$
 Since $\bm w_h$ is a nonsingular
solution, then there exists a unique solution of (10.14).
Consequently, $V$ can be decomposed into the direct sum of two
subspaces:
$$V(D)=V_h(D)\oplus\widehat{V}_h(D).$$ This meas that for any $\bm
w\in\,V(D)$, we have
$$\bm w=P_h\bm w+P_h^\bot\bm w=\bm w_p+\bm w_q, \quad \bm w_p\in\,V_h(D),\quad
\bm w_q\in\,\widehat{V}_h(D).$$  It is straightforward to show that
$$\left\{\begin{array}{ll}
{\cal A}_{\bm w_h}(\bm w_q,\bm v_p)=0,\quad\forall \bm
v_p\in\,V_h(D),\quad {\cal
A}_{\bm w_h}(\bm w,\bm v_p)={\cal A}_{\bm w_h}(\bm w_p,\bm v_p),\\
\|\bm w_q\|_{1,D}\leq Ch^k\|\bm w\|_{k+1,D},\forall \,\bm
w\in\,V\cap H^{k+1}(D)^2.
\end{array}\right.\eqno{(10.15)}$$

Next, we present some technical lemmas.

\begin{lemma} there exists a constant independent of $\bm u,\bm v,\bm w$ such that
$$|b_0(\bm u,\bm w,\bm v)|\leq
C\|\bm u\|^{\frac12}_{0,D}\|\bm u\|^{\frac12}_{1,D}\|\bm
w\|_{1,D}\|v\|^{\frac12}_{0,D}\|\bm v\|^{\frac12}_{1,D}, \quad
\forall\,\bm u,\bm w,\bm v\in\,V(D).\eqno{(10.16)}$$
\end{lemma}

\begin{proof}By virtue of H\"{o}lder inequality we get
$$|b_0(\bm u,\bm w,\bm v)|\leq C\|\bm u\|_{0,4,D}\|\bm w\|_{1,D}\|\bm v\|_{0,4,D},$$
furthermore  the Ladyzhenskaya inequality shows that
$$\|\bm u\|_{0,4,D}\leq
C\|\bm u\|^{\frac12}_{0,2,D}\|\bm u\|^{\frac12}_{1,2,D}.$$
 Therefore we prove the lemma immediately.
\end{proof}

Since $\bm w_0\in V(D)$, it can be decomposed into $\bm w_0=\bm
w_{0p}+\bm w_{0q}$ with $\bm w_{0p}\in\,V_h(D),\ \bm
w_{0q}\in\,\widehat{V}_h(D)$.

\begin{lemma} The following estimation is valid,
$$\|\bm w_{0p}-\bm w_h\|_{1,D}\leq \frac{2M}{\alpha_0}\|\bm w_0-\bm
w_h\|^{\varepsilon_1}_{1,D}\|\bm w_0-\bm
w_h\|^{\varepsilon_0}_{0,D},\eqno{(10.17)}$$
where
$$\begin{array}{l}
 \varepsilon_1=\left\{\begin{array}{ll}
 1, &\mbox{for the Homogenous Dirichelet B.C.}=(B.C.I),\\
\frac32,& \mbox{for Mixed B.C.}=(B.C.II),\\
\end{array}\right.\\
\varepsilon_0=\left\{\begin{array}{ll}
 1, & \mbox{for the Homogenous Dirichelet B.C.}=(B.C.I),\\
\frac12,& \mbox{for Mixed B.C.}=(B.C.II).
\end{array}\right.
\end{array}$$
\end{lemma}

\begin{proof} Firstly, equation ((10.1) can be rewritten as
$${\cal A}_{\bm w_h}(\bm w_{0q},\bm v)+{\cal
A}_{\bm w_h}(\bm w_{0p},\bm v)+b_0(\bm w_0-\bm w_h,\bm w_0-\bm
w_h,\bm v)-b_0(\bm w_h,\bm w_h,\bm v)=<\bm G_\eta,\bm v>,\forall
\,\bm v\in\,V(D),\eqno{(10.18)}$$ meanwhile (10.3) as
$${\cal A}_{\bm w_h}(\bm w_h,\bm v)-b_0(\bm w_h,\bm w_h,\bm v)=<\bm G_\eta,\bm v>,\forall
\,\bm v\in\,V_h(D).\eqno{(10.19)}$$
 Let $\bm v\in\,V_h$, subtracting
(10.18) from (10.19) and using (10.15) with $\bm w=\bm w_0$, we
derive
$${\cal A}_{\bm w_h}(\bm w_{0p}-\bm w_h,v)=-b_0(\bm w_0-\bm w_h,\bm w_0-\bm w_h,\bm v).\eqno{(10.20)}$$
Since $\bm w_h$ is nonsingular, (10.13) shows
$$\begin{array}{ll}
\frac12\alpha_0\leq \frac1{\|\bm w_{0p}-\bm
w_h\|_{1,D}}\sup\limits_{\bm v\in\,V_h(D)}\frac{{\cal A}_{\bm
w_h}(\bm w_{0p}-\bm w_h,v)}{\|\bm v\|_{1,D}}=\frac1{\|\bm w_{0p}-\bm
w_h\|_{1,D}}\sup\limits_{\bm v\in\,V_h(D)}\frac{-b_0(\bm w_0-\bm
w_h, \bm w_0-\bm w_h,v)}{\|\bm v\|_{1,D}},
\end{array}$$
that is,
$$\begin{array}{ll}
\frac12\alpha_0\|\bm w_p-\bm w_h\|_{1,D}\leq \sup\limits_{\bm
v\in\,V_h(D)}\frac{|b_0(\bm w_0-\bm w_h, \bm w_0-\bm w_h,\bm
v)|}{\|\bm v\|_{1,D}}.
\end{array}$$
By using (10.16), for any $\bm v\in V_h(D)$ we get
$$\begin{array}{ll}
|b_0(\bm w_0-\bm w_h,\bm w_0-\bm w_h,\bm v)|=|b_0(\bm w_0-\bm w_h,\bm v,\bm w_0-\bm w_h)|\\
\quad\quad \leq
M\|\bm w_0-\bm w_h\|_{0,D}\|\bm w_0-\bm w_h\|_{1,D}\|\bm v\|_{1,D},\quad \mbox{for}\,B.C.I,\\
|b_0(\bm w_0-\bm w_h,\bm w_0-\bm w_h,\bm v)|=|b_0(\bm w_0-\bm w_h,\bm v,\bm w_0-\bm w_h)|\\
 \quad\quad\leq
M\|\bm w_0-\bm w_h\|^{\frac12}_{0,D}\|\bm w_0-\bm
w_h\|^{\frac32}_{1,D}\|\bm v\|_{1,D},\quad  \mbox{for}B.C.II,
\end{array}\eqno{(10.21)}$$
Summing up the above relations we draw the conclusion (10.17).
\end{proof}

 Next, let the mapping $\bm\phi(\cdot): V_h(D)\rightarrow
 \widehat{V}_h(D)$.
 we define the manifold ${\cal M}$ as the the graph of a function $\bm\phi$, that is,
 ${\cal M}=\mbox{Graph}{\bm\phi}$, then problem (10.1) can be
rewriten as
$$\left\{\begin{array}{ll}
\mbox{Find}\bm\phi(\bm w)\in\,\widehat{V}_h(D),\ \mbox{such that}\\
{\cal A}_{\bm w_h}(\bm\phi(\bm w),\bm v)=b_0(\bm w,\bm w,\bm
v)-{\cal A}_{\bm w_h}(\bm w,\bm v)-{\cal A}_{\bm w_h}(\bm v,\bm
w)+<\bm G_\eta,\bm v>,\quad\forall\,\bm v\in\,\widehat{V}_h.
\end{array}\right.\eqno{(10.22)}$$
We first give an approximation property of the solution of (10.22).

\begin{theorem} Suppose that the finite element space $V_h$ satisfies the
assumptions $(H1)\sim(H4)$. Then there exists a mapping $\bm\phi(\bm
w)$ defined by (10.22) which is a Lipschitz continuous function with
the Lipschitz constant $l=l(\rho)$, and $\bm\phi(\cdot)$ attracts
any solution $\bm w_0$ of (10.1), i.e.,
$$\begin{array}{ll}
(\textbf{H5})\qquad \|\bm\phi(\bm w_1)-\bm\phi(\bm w_2)\|_{1,D}\leq
l(\rho)\|\bm w_1-\bm w_2\|_{1,D},\quad\forall \,\bm w_1,\bm
w_2\in\,V_h(D)\cap B_\rho,\\
(\textbf{H6})\qquad \mbox{dist}(\bm w_0,{\cal M})\leq
\delta=C(1+\|\bm w_{0p}\|_{1,D}+\|\bm w_h\|_{1,D})\|\bm
w_0\|_{1,D}h^{2k+\frac12},
\end{array}$$
where $B_\rho=\{\bm w|\bm w\in\,V(D),\|\bm w\|_{1,D}\leq \rho\}.$
\end{theorem}

\begin{proof} Let $\bm w_i\in\,V_h,
\bm\phi_i=\bm\phi(\bm w_i)$ for $i=1,2$, and
$\bm\phi=\bm\phi_1-\bm\phi_2$.  If  setting $\bm w=\bm w_1, \bm w_2$
in (10.22) respectively and making subtraction, then we get
$$\left\{\begin{array}{ll}
{\cal A}_{\bm w_h}(\bm\phi,\bm v)=b_0(\bm w_1-\bm w_2,\bm w_1,\bm
v)+b_0(\bm w_2,\bm w_1-\bm w_2,\bm v)-{\cal
A}_{\bm w_h}(\bm w_1-\bm w_2,\bm v)\\
\qquad-{\cal A}_{\bm w_h}(\bm v,\bm w_1-\bm w_2),\quad\forall\,\bm
v\in\,\widehat{V}_h.
\end{array}\right.\eqno{(10.23)}$$
 Owing to  $w_h$ is nonsingular and (10.13), the following
 inequality is valid,
$$\begin{array}{ll}
\frac12\alpha_0\|\bm\phi\|_{1,D}&\leq\sup\limits_{\bm
v\in\,V_h(D)}\frac{{\cal A}_{\bm w_h}(\bm\phi,v)}{\|\bm
v\|_{1,D}}\leq\sup\limits_{\bm v\in\,V(D)}\frac{{\cal
A}_{\bm w_h}(\bm\phi,\bm v)}{\|\bm v\|_{1,D}}\\
&\leq M(\|\bm w_1\|_{1,D}+\|\bm w_2\|_{1,D}+\|\bm
w_h\|_{1,D}+1)\|\bm w_1-\bm w_2\|_{1,D}.
\end{array}$$ Furthermore by using the triangle inequality
$\|\bm w_h\|_{1,D}\leq\|\bm w_0-\bm w_h\|_{1,D}+\|\bm w_0\|_{1,D}$,
we derive the $(\textbf{H}5)$.

Our task is now to prove $(\textbf{H}6)$. We first know that
$$\begin{array}{ll}\text{dist}(\bm w_0,{\cal M})&=\inf\limits_{\bm w\in\,{\cal
M}}\|\bm w_0-\bm w\|_{1,D}\\
&\leq\|\bm w_0-(\bm w_{0p}+\bm\phi(\bm w_{0p}))\|_{1,D}\\
&=\|\bm w_0-\bm w_{0p}-\bm\phi(\bm w_{0p})\|_{1,D}\\
&=\|\bm w_{0q}-\bm\phi(\bm w_{0p})\|_{1,D}.\end{array}$$
 For any $\bm v\in\,V(D)$, Equation (10.1) can be rewritten as,
$${\cal A}_{\bm w_h}(\bm w_{0q},\bm v)+{\cal
A}_{\bm w_h}(\bm w_{0p},\bm v)+b_0(\bm w_0-\bm w_h,\bm w_0-\bm
w_h,\bm v)-b_0(\bm w_h,\bm w_h,\bm v)=<\bm G_\eta,\bm
v>.\eqno{(10.24)}$$ Choosing $\bm w=w_{0p}$ in (10.22) gives
$$\begin{array}{ll}
{\cal A}_{\bm w_h}(\bm\phi(\bm w_{0p}),\bm v)=b_0(\bm w_{0p},\bm
w_{0p},\bm v)-{\cal A}_{\bm w_h}(\bm w_{0p},\bm v)-{\cal A}_{\bm
w_h}(\bm v,\bm w_{0p})+<\bm G_\eta,\bm v>,\,\forall\,\bm
v\in\,\widehat{V}_h(D).
\end{array}\eqno{(10.25)}$$
Setting $\bm v\in\,\widehat{V}_h(D)$ in (10.24) and Subtracting with
(10.25). Taking  into account (7.15), then we have
$$\begin{array}{ll}
{\cal A}_{\bm w_h}(\bm w_{0q}-\bm\phi(\bm w_{0p}),\bm v)=-b_0(\bm
w_0-\bm w_h,
\bm w_0-\bm w_h,\bm v)+b_0(\bm w_h,\bm w_h,\bm v)-b_0(\bm w_{0p},\bm w_{0p},\bm v)\\
\qquad=-b_0(\bm w_0-\bm w_h, \bm w_0-\bm w_h,\bm v)+b_0(\bm w_h-\bm
w_{0p},\bm w_h,\bm v)+b_0(\bm w_{0p},\bm w_h-\bm w_{0p},\bm v).
\end{array}$$
 Since $\bm w_0$ is a nonsingular solution of (7.1), and noting that (10.5) and
 (10.15), we get
$$\begin{array}{ll}
\alpha_0\|\bm w_{0q}-\bm\phi(\bm w_{0p})\|_{1,D}\leq\sup\limits_{\bm
v\in\,V(D)}\frac{{\cal A}_{\bm w_0}((\bm w_{0q}-\bm\phi(\bm
w_{0p}),\bm v)}{\|\bm v\|_{1,D}}\leq I+II.
\end{array}$$
Our problem reduce to the estimation of $I, II$. It is easy to show
that
$$\begin{array}{ll}
I&=\sup\limits_{\bm v\in\,V(D)}\frac{{\cal
A}_{\bm w_h}(\bm w_{0q}-\bm\phi(\bm w_{0p}),\bm v)}{\|\bm v\|_{1,D}}\\
&=\sup\limits_{(\bm v_{0p}+\bm v_{0q})\in\,V(D)}\frac{{\cal
A}_{\bm w_h}(\bm w_{0q}-\bm\phi(\bm w_{0p}),\bm v_{0p}+\bm v_{0q})}{\|\bm v_{0p}+\bm v_{0q}\|_{1,D}}\\
&\leq\sup\limits_{\bm v_{0q}\in\,\widehat{V}_h(D)}\frac{{\cal
A}_{\bm w_h}(\bm w_{0q}-\bm\phi(\bm w_{0p}),\bm v_{0q})}{\|\bm v_{0q}\|_{1,D}}\\
&=\sup\limits_{\bm v_{0q}\in\,\widehat{V}_h(D)}\frac{ -b_0(\bm
w_0-\bm w_h,
\bm w_0-\bm w_h,\bm v_{0q})+b_0(\bm w_h-\bm w_{0p},\bm w_h,v_{0q})+b_0(\bm w_p,\bm w_h-\bm w_{0p},\bm v_{0q})}{\|\bm
v_{0q}\|_{1,D}}\\
 &\leq
M(\|\bm w_0-\bm w_h\|^{\frac12}_{0,D}\|\bm w_0-\bm w_h\|^{\frac32}_{1,D}+(\|\bm w_p\|_{1,D}+\|\bm w_h\|_{1,D})\|\bm w_{0p}-\bm w_h\|_{1,D})\\
&\leq
M(1+\|\bm w_{0p}\|_{1,D}+\|\bm w_h\|_{1,D})\|\bm w_0-\bm w_h\|^{\frac12}_{0,D}\|\bm w_0-\bm w_h\|^{\frac32}_{1,D},\\
II&=\sup\limits_{\bm v\in\,V(D)}\frac{{\cal A}_{\bm w_0}(\bm
w_{0q}-\bm\phi(\bm w_{0p}),\bm v)-{\cal
A}_{\bm w_h}(\bm w_{0q}-\bm\phi(\bm w_{0p}),\bm v)}{\|\bm v\|_{1,D}}\\
&=\sup\limits_{\bm v\in\,V(D)}\frac{b_0(\bm w_0-\bm w_h,\bm w_{0q}-\bm\phi(\bm w_{0p}),\bm v)+b_0(\bm w_{0q}-\bm\phi(\bm w_{0p}),\bm w_0-\bm w_h,\bm v)}{\|\bm v\|_{1,D}}\\
&\leq 2M\|\bm w_0-\bm w_h\|_{1,D}\|\bm w_{0q}-\bm\phi(\bm w_{0p})\|_{1,D}\\
&\leq 2MCh^k\|\bm w_{0q}-\bm\phi(\bm w_{0p})\|_{1,D}.
\end{array}$$
where in the fifth step estimation of $I$ we adopt $(10.16)$.
Finally, noting that $(\textbf{ H1})$, we have
$$\begin{array}{ll}
\|\bm w_{0q}-\bm\phi(\bm w_{0p})\|_{1,D}&\leq \frac{M(1+\|\bm
w_{0p}\|_{1,D}+\|\bm w_h\|_{1,D})}{\alpha_0-2MCh^k}
\|\bm w_0-\bm w_h\|^{\frac12}_{0,D}\|\bm w_0-\bm w_h\|^{\frac32}_{1,D} \\
&\leq Ch^{2k+\frac12}. \end{array}$$
 So we prove $(\textbf{H6})$.
\end{proof}

\begin{theorem} Assume that the assumptions $(H1)\sim(H4)$ for finite element space $V_h$ are
satisfied and $\bm w_h$ is the nonsingular solution of (10.3). Then
variational problem  (one step Newtonian iteration)
$$\left\{\begin{array}{ll}
\mbox{Find}\,\bm w_*\in\,V(D) \mbox{such that}\\
{\cal A}_{\bm w_h}(\bm w_*,\bm v)=<\bm G_\eta,\bm v>+b_0(\bm w_h,\bm
w_h,\bm v),\quad\forall\,\bm v\in\,V(D),
\end{array}\right.\eqno{(10.26)}$$
has a unique  solution $\bm w_*$ and  the following estimation are
valid
$$\|\bm w-\bm w_*\|_{1.D}\leq Ch^{2k+\varepsilon},\eqno{(10.27)}$$
where
$$\varepsilon=\left\{\begin{array}{ll}
1,\quad \mbox{for }B.C.I,\\
\frac12,\quad \mbox{for}\,B.C.II
\end{array}\right.$$
\end{theorem}

\begin{proof} Navier-Stokes equations (10.1) equals to
$${\cal A}_{\bm w_h}(\bm w_0,\bm v)+b_0(\bm w_0-\bm w_h,\bm w_0-\bm w_h,\bm v)-b_0(\bm w_h,\bm w_h,\bm v)=<\bm G_\eta,\bm v>,\forall
\bm v\in\,V(D).\eqno{(10.28)}$$
 Subtracting (10.28) from (10.26) leads to
$${\cal A}_{\bm w_h}(\bm w_0-\bm w_*,\bm v)+b_0(\bm w_0-\bm w_h,\bm w_0-\bm w_h,\bm v)=0,\forall
\bm v\in\,V(D).\eqno{(10.29)}$$
 By applying (10.13) and lemma 10.1 we assert
$$\begin{array}{ll}
\frac12\alpha_0\|\bm w_0-\bm w_*\|_{1,D}&\leq\sup\limits_{\bm
v\in\,V_h(D)}\frac{{\cal
A}_{\bm w_h}(\bm w_0-\bm w_*,v)}{\|\bm v\|_{1,D}}\\
&\leq\sup\limits_{\bm v\in\,V(D)}\frac{{\cal
A}_{\bm w_h}(\bm w_0-\bm w_*,\bm v)}{\|\bm v\|_{1,D}}\\
&\leq\sup\limits_{\bm v\in\,V_h(D)}\frac{-b_0(\bm w_0-\bm w_h,\bm
w_0-\bm w_*,\bm v)}{\|\bm v\|_{1,D}}.
\end{array}$$
 For B.C.I., we have
$$\begin{array}{ll}
|b_0(\bm w_0-\bm w_h,\bm w_0-\bm w_*,\bm v)|&=|-b_0(\bm w_0-\bm w_h,\bm v,\bm w_0-\bm w_*)|\\
&\leq
M\|\bm w_0-\bm w_h\|_{0,D}\|\bm w_0-\bm w_h\|_{1,D}\\
&\leq \|\bm v\|_{1,D}MC\|\bm w_0\|^2_{k+1,D}h^{2k+1}.\end{array}$$
Similarly, For B.C. II.,
$$|b_0(\bm w_0-\bm w_h,\bm w_0-\bm w_*,\bm v)|\leq M\|\bm v\|_{1,D}\|\bm w_0-\bm w_h\|^{\frac12}_{0,D}\|\bm w_0-\bm w_h\|^{\frac32}_{1,D}
\leq MC\|\bm v\|_{1,D}\|\bm w_0\|^2_{k+1,D}h^{2k+\frac12}.$$ Thus
the proof is completed.
\end{proof}

\begin{remark} First, it is simple to show that
$$w_*=w_h+\phi(w_h).$$
Second, the variational problem (10.26)  is still an infinite
dimensional problem. We can apply the standard two-level finite
element method on this problem(see Layton et al.[21-23] and
references therein).
\end{remark}

\begin{theorem}
Suppose that the assumptions in theorem 10.4 are satisfied.
$V_{h^*}$ is a finite element subspace with mesh parameter $h^* \leq
h$ and satisfies assumptions $(H1)\sim(H4)$ with integer $m\leq k$.
If $(\bm w_{h^*})$ is a Galerkin finite element approximation
solution to (10.26), that is
$$
\left\{\begin{array}{l}
\mbox{Find} (\bm w_{h^*}) \in V_{h^*} \mbox{such that}\\
{\cal A}_{\bm w_h}(\bm w_{h^*},\bm v)=(\bm G_\eta,\bm v) + b_0(\bm
w_h,\bm w_h,\bm v)
       \quad\forall \bm v\in V_{h^*}.
\end{array}\right.\eqno{(10.30)}$$
Then the following error estimation holds
$$\|\bm w_*-\bm w_{h^*}\|_{1,D}  \leq C h^{*(m+1)}(\|\bm w_*\|_{m+1}).
   \eqno{(10.31)}$$
\end{theorem}
\begin{proof}The proof is omitted.
\end{proof}

Combining theorem 10.4 and 10.5 leads to our final conclusion,

\begin{theorem} Suppose that the assumptions in theorem 10.4 and 10.5 are
satisfied. Then we have following estimation
$$ \|(\bm w_0-\bm w_{h^*})\|_{1,D} \leq C(h^{2k+\varepsilon} + h^{*(m+1)}).$$
In particular, if choosing $h^* = h^{(2k+1)/(m+1)}$, then we have
$$ \|(\bm w_0-\bm w_{h^*})\|_{1,D} \leq C(h^{2k+\varepsilon} ).$$
\end{theorem}
\begin{proof}The proof is omitted.
\end{proof}

\begin{algorithm}
Here we  present the finite element approximation algorithm based on
the approximate inertial manifold, i.e.,
\begin{itemize}
\item {\bf Step1:} \quad Solve the nonlinear problem (10.3) on the coarse grid with mesh size
$h$,
\item {\bf Step2:} \quad Solve the linear problem (10.30) on the fine grid with mesh size
$h^*$.
\end{itemize}
\end{algorithm}

\begin{remark}
If we use the linear finite element method for (10.3) and (10.30),
respectively, then on the two dimensional problem the following
estimates hold
$$\|(\bm w_0-\bm w_{h^{*}})\|_{1,D}\leq ch^{3}\approx ch^{*^{2}},$$
where $h^{*}\approx h^{\frac{3}{2}}$.
 As we know that in Layton[19, Theorem 2], the result is
$$\|(\bm w_0-\bm w_{h^{*}})\|_{1,D}\leq ch^{2}\approx ch^{*},$$
with $h^{*}\approx h^{2}.$
 This shows that our results is much better than that in [19].
\end{remark}

\appendix
\section{Appendix}

This  section gathers most of the  preliminary knowledge that will
be required in this article. In  subsection 1, we focus on the the
expressions of some physical and geometrical quantities in the new
coordinates system; then in subsection 2, The Navier-Stokes Equation
in the new coordinate system is derived. Finally, we consider the
G\^{a}teaux derivative of the solutions of the Navier-Stokes
equations with respect to the shape of blade.

\subsection{Some Physical and Geometrical Quantities}

In order to simplicity, we consider the 3D fluid flow in an flow
passage in an impeller with  rotating angular velocity
$\bm\omega=(0,0,\omega)$ around its axis, and the thickness of the
blade is uniform. Let $(x,y,z)$ be the cartesian coordinate system
out of the impeller in the Euclidean space $R^3$, and three
coordinate basis vector are $\bm i,\bm j,\bm k$ respectively.
Furthermore, $(r,\theta,z)$ be the cylindrical coordinate system
attached to and fixed on the impeller, and $\bm e_r,\bm e_\theta,\bm
e_z$ are three basis vector of this system respectively. Next we
define a new coordinate system $(x^1,x^2,x^3)$ through the following
relations
$$\left\{\begin{array}{lll}
x^1=z, & x^2=r,&x^3=\xi=\varepsilon^{-1}(\theta-\Theta(x^1,x^2)),\\
r=x^2, &  z=x^1,&\theta=\varepsilon\xi+\Theta(x^1,x^2),
\end{array}\right.\eqno{(A.I.1)}$$
where $\Theta(x^1,x^2)$ be a smooth mapping from a bounded smooth
enough subset $D\subset R^2$ into $R$, especially, $(x^1,x^2,
\Theta(x^1,x^2))$ denote an arbitrary point on the blade surface.
Let $\bm e_i$ be the basic vectors of this  new coordinate system.
The parameter $\xi$ satisfies $0\leq\xi\leq1$ and obviously,
$\xi=const$ represent a surface $\Im_\xi$ in $R^3$, which can be
obtained by a rotation of the blade through an angle of
$\varepsilon\xi$ degree.

Next, we present the following proposition.
\begin{proposition}
The covariant components $a_{\alpha\beta}$ and contra-variant
components $a^{\alpha\beta}$ of the metric tensor of the surface
$\Im_\xi$ and the covariant components $g_{ij}$ and the
contra-variant components $g^{ij}$ of the metric tensor of 3D
Euclidean space $\bm R^3$ are  given by, respectively,
$$\left\{\begin{array}{ll}
 a_{\alpha\beta}=\delta_{\alpha\beta}+(x^2)^2\Theta_\alpha\Theta_\beta,\quad
a=\det(a_{\alpha\beta})=1+(x^2)^2(\Theta_1^2+\Theta_2^2),\quad
\Theta_\alpha=\frac{\partial \Theta}{\partial
x^\alpha},\\
a^{\alpha\beta}a_{\beta\sigma}=\delta^\alpha_\sigma,\quad
a^{11}=\frac{a_{22}}{\sqrt{a}},\quad
a^{22}=\frac{a_{11}}{\sqrt{a}},\quad
a^{12}=a^{21}=-\frac{a_{12}}{\sqrt{a}},\\
g_{\alpha\beta}=a_{\alpha\beta},\quad
g_{3\alpha}=g_{\alpha3}=\varepsilon r^2\Theta_\alpha,\quad
g_{33}=\varepsilon^2r^2,\quad g=\det{(g_{ij})}=\varepsilon^2r^2,\\
g^{ij}g_{jk}=\delta^i_k,\quad g_{\alpha\beta}=a_{\alpha\beta},\quad
g_{3\alpha}=g_{\alpha3}=\varepsilon r^2\Theta_\alpha,\quad
g_{33}=\varepsilon^2r^2,
\end{array}\right.\eqno{(A.I.2)}$$ where $\delta^i_j$ is the Kronecker symbol.
 Furthermore, we have the followings conclusions

 {\bf 1. Rotating Angular Velocity $\bm\omega$}
$$
\left\{\begin{array}{l}
\bm\omega=\omega \bm e_1-\omega\varepsilon^{-1}\Theta_1 \bm e_3\\
\omega^1=\omega,\quad \omega^2=0,\quad
\omega^3=-\omega\varepsilon^{-1}\Theta_1,
\end{array}\right.\eqno(A.I.3)$$

{\bf 2. Coriolis Forces}
 $$\left\{
\begin{array}{l}
2\bm\omega\times\bm w=C^1\bm e_1+C^2\bm e_2+C^3\bm e_3\\
C^1=0,\quad C^2=-2r\omega \Pi(\bm w,\Theta),\\
C^3= 2\omega\varepsilon^{-1}(r\Theta_2\Pi(\bm
 w,\Theta)+\frac{w^2}{r}),
\end{array}\right.\eqno(A.I.4)$$

{\bf 3. Unite Normal vector to the Surface $\Im_\xi$}
 $$\left\{\begin{array}{l} \textbf{n}=-x^2\Theta_\alpha/\sqrt{a}
\bm e_\alpha+(\varepsilon
 x^2)^{-1}\frac{1+r^2\Theta_2^2}{\sqrt{a}}\bm e_3,\\
 n^\alpha=-x^2\Theta_\alpha/\sqrt{a},\quad n^3=(\varepsilon
 x^2)^{-1}\frac{1+r^2\Theta_2^2}{\sqrt{a}}.
\end{array}\right.\eqno(A.I.5)$$

 {\bf 4.  Curvature Tensor of the Surface $\Im_\xi$ (Second Fundamental Form)}
 $$\begin{array}{ll}
b_{11}=\frac1{\sqrt{a}}(\Theta_2(a_{11}-1)+x^2\Theta_{11}).\quad
b_{12}=b_{21}=\frac1{\sqrt{a}}(\Theta_1a_{12}+x^2\Theta_{12}),\\
b_{22}=\frac{1}{\sqrt{a}}(\Theta_2(a_{22}+1)+x^2\Theta_{22}),
\quad b=\det{(b_{\alpha\beta})}=b_{11}b_{22}-b_{12}^2,\\
\end{array}\eqno{(A.I.6)}$$
where
$$\begin{array}{ll}
|\widetilde{\nabla}\Theta|^2=\Theta_1^2+\Theta_2^2,\quad \Delta
\Theta=a^{\alpha\beta}\Theta_{\alpha\beta},\quad
\widetilde{\Delta}\Theta=\Theta_{11}+\Theta_{22},\quad
\Theta_{\alpha\beta}=
\partial_\alpha\partial_\beta\Theta,
\end{array}\eqno{(A.I.7)}$$

{\bf 5. Mean Curvature and Gaussian Curvature of the Surface
$\Im_\xi$}
$$\left\{\begin{array}{l}
\,\,\,K=b/a,\\
2H=\frac1{a\sqrt{a}}[x^2(a_{22}\Theta_{11}+a_{11}\Theta_{22})-2a_{12}\Theta_{12})
\Theta_2(2a_{11}a_{22}+a_{11}-a_{22})-2\Theta_1a^2_{12},
\end{array}\right.\eqno{(A.I.8)}$$
\end{proposition}

\begin{proof} Firstly, from (A.I.1) we have
  $$\left\{\begin{array}{l}
x=x(x^1,x^2,\xi)=r\cos\theta=x^2\cos(\varepsilon\xi+\Theta(x^1,x^2))\\
y=y(x^1,x^2,\xi)=r\sin\theta=x^2\sin(\varepsilon\xi+\Theta(x^1,x^2))\\
z=z(x^1,z^2,\xi)=x^1
\end{array}\right.\eqno(A.I.9)$$ Therefore,
  $$
\left\{\begin{array}{lll} \dfrac{\partial x}{\partial
x^1}=-x^2\sin\theta \Theta_1,& \dfrac{\partial x}{\partial
x^2}=\cos\theta-x^2\sin\theta \Theta_2,&
\dfrac{\partial x}{\partial \xi}=-x^2\sin\theta\varepsilon,\\
\dfrac{\partial y}{\partial x^1}=x^2\cos\theta\Theta_1, &
\dfrac{\partial y}{\partial x^2}=\sin\theta+x^2\cos\theta\Theta_2,&
\dfrac{\partial y}{\partial\xi}=x^2\cos\theta\varepsilon,\\
\frac{\partial z}{\partial x^1}=1,& \frac{\partial z}{\partial
x^2}=\frac{\partial z}{\partial x^3}=0,&
\end{array}\right.\eqno(A.I.10)$$
 where $\theta=\varepsilon\xi+\Theta(x^1,x^2)$.
From (A.I.10) we  get
 $$\dfrac{\partial(x,y,z)}{\partial(x^1,x^2,x^3)}=\varepsilon x^2.\eqno(A.I.11)$$
It is well known that
$$ \left\{\begin{array}{ll}
\bm e_r=\cos\theta \bm i+\sin\theta \bm j,&
\bm e_\theta =-\sin\theta \bm i+\cos\theta \bm j,\\
\bm i=\cos\theta \bm e_r-\sin\theta \bm e_\theta,& \bm j=\sin\theta
\bm e_r+\cos\theta \bm e_\theta,
\end{array}
\right. \eqno(A.I.12)$$

Let $\bm\Re(x^1,x^2,\xi)=x(x^1,x^2,\xi)\bm i+y(x^1,x^2,\xi)\bm
j+z(x^1,x^2,\xi)\bm k$ denote the radial vector at the point
$P=(x^1,x^2,\xi)$. Then the covariant basic vectors $(\bm e_\alpha,
\bm e_3)$ and the contra-variant basic vectors $(\bm e^\alpha, \bm
e^3)$ in the new curvilinear coordinate system $(x^\alpha,\xi)$ are
given by, respectively,
$$\left\{\begin{array}{l}
\bm e_\alpha=\partial_\alpha \bm\Re=\partial_\alpha x\bm
i+\partial_\alpha y\bm j+\partial_\alpha z\bm k,\quad \bm
e_3=\frac{\partial}{\partial\xi}\bm\Re=\frac{\partial
x}{\partial\xi}\bm i+\frac{\partial
y}{\partial\xi}\bm j+\frac{\partial z}{\partial\xi}\bm k,\\
\bm e_1=x^2\Theta_1\bm e_\theta+\bm k=-x^2\sin\theta\Theta_1 \bm i
+x^2\cos\theta\Theta_1\bm j+\bm k,\\
\bm e_2=\Theta_2 x^2\bm e_\theta+\bm
e_r=(\cos\theta-x^2\sin\theta\Theta_2)
\bm i+(\sin\theta+x^2\cos\theta\Theta_2)\bm j,\\
\bm e_3=x^2\varepsilon\bm e_\theta=-\varepsilon
x^2\sin\theta\bm i +\varepsilon x^2\cos\theta\bm j,\\
\bm e^i=g^{ij}\bm e_j,\quad \bm e^\alpha=\bm
e_\alpha-\varepsilon^{-1}\Theta_\alpha \bm e_3,\quad
\bm e^3=-\varepsilon^{-1}\Theta_\alpha\bm e_\alpha+(r\varepsilon)^{-2}a\bm e_3,\\
\bm e^1=\bm k,\quad \bm e^2=\cos\theta
\bm i+\sin\theta \bm j,\\
\bm e^3=-(r\varepsilon)^{-1}(\sin\theta+r\Theta_2\cos\theta)\bm i
+(r\varepsilon)^{-1}(\cos\theta-r\Theta_2\sin\theta)\bm j-\varepsilon^{-1}\Theta_1\bm k.\\
\end{array}
\right. \eqno(A.I.13)$$ Inversely, we have
 $$\left\{
\begin{array}{l}
\bm e_r=\bm e_2-\varepsilon^{-1}\Theta_2\bm e_3,\quad \bm
e_\theta=(\varepsilon x^2)^{-1}\bm e_3,\quad
\bm k=\bm e_1-\varepsilon^{-1}\Theta_1\bm e_3,\\
\bm i=\cos\theta\bm
e_2-(\varepsilon^{-1}\cos\theta\Theta_2+(\varepsilon
x^2)^{-1}\sin\theta)\bm e_3,\\
\bm j=\sin\theta\bm e_2+((\varepsilon
x^2)^{-1}\cos\theta-\varepsilon^{-1}\Theta_2\sin\theta
)\bm e_3,\\
\end{array}
\right. \eqno(A.I.14)$$ For any fixed $\xi$, the mapping
$$\bm\Re(x^1,x^2;\xi)=x(x^1,x^2,\xi)\bm i+y(x^1,x^2,\xi)\bm j+z(x^1,x^2,\xi)\bm k$$
define a 2-dimensional surface $\Im_\xi$ with single parameter $\xi$
and the covariant components of metric tensor of $\Im_\xi$ is
expressed as
$$a_{\alpha\beta}=\bm e_\alpha \bm e_\beta=\delta_{\alpha\beta}+(x^2)^2\Theta_\alpha\Theta_\beta,\quad
a=\det(a_{\alpha\beta})=1+(x^2)^2(\Theta_1^2+\Theta_2^2)=1+|\widetilde{\nabla}\Theta|^2.$$
 Meantime, the covariant components of the metric tensor of 3D
 Euclidean space $R^3$ in the coordinate system $(x^1,x^2,\xi)$
 are given by
$$g_{ij}=\bm e_i\cdot\bm e_j$$
From this and the expression (A.I.13) it is easy to derive (A.I.1).

Next, we consider the angular velocity vector. Obviously,
$$\left\{\begin{array}{ll}
\bm\omega=\omega\bm k=\omega\bm e_1
-\omega\varepsilon^{-1}\Theta_1\bm e_3,\\
\omega^1=\omega,\quad \omega^2=0,\quad
\omega^3=-\varepsilon^{-1}\omega\Theta_1,
\end{array}\right.\eqno(A.I.15)$$
which norm can calculated as
 $$\begin{array}{ll}
 |\bm\omega|^2&=\bm\omega\cdot\bm\omega=g_{ij}\omega^i\omega^j=g_{11}\omega^1\omega^1+g_{33}\omega^3\omega^3
 +2g_{13}\omega^1\omega^3\\
 &=a_{11}(\omega)^2+r^2\varepsilon^2(\omega)^2\varepsilon^{-2}\Theta^2_1+2\varepsilon
 r^2\Theta_1(\omega)^2(-\varepsilon^{-1}\Theta_1)\\
 &=(\omega)^2(a_{11}+r^2\Theta_1^2-2r^2\Theta_1^2)=(\omega)^2.
 \end{array}$$

Similarly, from the coordinate relation (A.I.1) and the definition
of $\varepsilon_{ijk}$, the Coriolis force is formulated as
$$\begin{array}{ll}
\bm C&=2\bm\omega\times \bm w=2\varepsilon_{ijk} \omega^j w^k \bm
e^i= (2\varepsilon_{ijk}\omega^j w^k)g^{im}
\bm e_m\\
&=2\varepsilon_{ijk}\omega^j w^k(g^{i1}\bm e_1 +g^{i2}\bm e_2+g^{i3}\bm e_3)=C^i\bm e_i,\\
 C^1&=2(g^{i1}\varepsilon_{i1k}\omega
w^k+g^{i1}\varepsilon_{i3k}(-\omega\varepsilon^{-1}\Theta_1)w^k)\\
&=2\omega(g^{\alpha1}\varepsilon_{\alpha13}
w^3+g^{31}\varepsilon_{312} w^2
-g^{\alpha1}\varepsilon_{\alpha3\beta}\varepsilon^{-1}\Theta_1
w^\beta)\\
&=2\omega(0-\varepsilon^{-1}\Theta_1
w^2\sqrt{g}-\varepsilon_{132}\varepsilon^{-1}\Theta_1 w^{2})=0,\\

 C^2&=2(g^{i2}\varepsilon_{i1k}\omega
w^k+g^{i2}\varepsilon_{i3k}(-\omega\varepsilon^{-1}\Theta_1)w^k)\\
&=2\omega(g^{\alpha2}\varepsilon_{\alpha13}
w^3+g^{32}\varepsilon_{312} w^2
-\varepsilon^{-1}\Theta_1 g^{\alpha2}\varepsilon_{\alpha3\beta} w^\beta)\\
&=2\omega(\varepsilon_{213} w^3-\varepsilon^{-1}\Theta_2
w^2\varepsilon_{312}-\varepsilon^{-1}\Theta_1
\varepsilon_{231} w^1)\\
&=2\omega\sqrt{g}(-w^3-\varepsilon^{-1}\Theta_2w^2-\varepsilon^{-1}\Theta_1w^1)
=-2r\omega\Pi(w,\Theta),\\

C^3 &=2(g^{i3}\varepsilon_{i1k}\omega
w^k+g^{i3}\varepsilon_{i3k}(-\omega\varepsilon^{-1}\Theta_1)w^k)\\
&=2(g^{\alpha3}\varepsilon_{\alpha13}\omega
w^3+g^{33}\varepsilon_{312}\omega w^2
+ g^{\alpha3}\varepsilon_{\alpha3\beta}(-\omega\varepsilon^{-1}\Theta_1) w^\beta)\\
&=2\omega(-\varepsilon^{-1}\Theta_2\varepsilon_{213}
w^3+g^{33}\varepsilon_{312} w^2
-\varepsilon^{-1}\Theta_1(-\varepsilon^{-1}\Theta_2\varepsilon_{231}
w^1
-\varepsilon^{-1}\Theta_1\varepsilon_{132} w^2))\\
&=2\omega\sqrt{g}(\varepsilon^{-1}\Theta_2 w^3+g^{33} w^2
+\varepsilon^{-2}\Theta_1(\Theta_2 w^1-\Theta_1 w^2)).
\end{array}$$
Owing to the identity
$$g^{33}=(r\varepsilon)^{-2}w^2+\varepsilon^{-2}(\Theta_1^2+\Theta_2^2)w^2,$$
we have
$$C^3=2r\omega\varepsilon^{-1}\Theta_2\Pi(w,\Theta)+2\omega(r\varepsilon)^{-1} w^2.$$
where
$$\Pi(w,\Theta)=\varepsilon w^3+\Theta_\alpha w^\alpha.$$
Therefore,
$$\begin{array}{ll}
2\bm\omega\times \bm w&=-2r\omega\Pi(w,\Theta)\bm e_2
+(2r\omega\varepsilon^{-1}\Theta_2\Pi(w,\Theta)+2\omega(r\varepsilon)^{-1}w^2)\bm e_3]\\
\end{array}$$
Finally, the contravariant components of Coriolis force in the new
coordinate system is given by
$$\begin{array}{l}
C^1=0,\quad C^2=-2\omega r\Pi(w,\Theta),\quad C^3=2\omega
(\varepsilon)^{-1}(r\Theta_2\Pi(w.\Theta)+\frac{w^2}{r}).\end{array}\eqno{(A.I.16)}$$

Next we consider the unite normal vector $\bm n$ of $\Im_\xi$. At
first, it is well know that,
$$\bm n=\frac{\bm e_1
\times \bm e_2}{|\bm e_1\times \bm e_2|} =\frac{1}{\sqrt{a}}(\bm
e_1\times \bm e_2)=n_x\bm i+n_y\bm j+n_z\bm k=n^i\bm e_i$$ By virtue
of (A.I.13) and (A.I.14), the above expression shows that
 $$\begin{array}{ll}\bm e_1\times
 \bm e_2&=\left|\begin{array}{lll}
 \bm i & \bm j & \bm k\\
(\bm e_1)_x&(\bm e_1)_y& (\bm e_1)_z\\
(\bm e_2)_x& (\bm e_2)_y& (\bm e_2)_z
\end{array}\right|\\[20pt]
&=\left|\begin{array}{lll}
 \bm i & \bm j & \bm k\\
-r\Theta_1\sin\theta& r\Theta_1\cos\theta& 1\\
\cos\theta-r\Theta_2\sin\theta& \sin\theta+r\Theta_2\cos\theta& 0
\end{array}\right|\\[20pt]
&=-(\sin\theta+r\Theta_2\cos\theta)\bm
i+(\cos\theta-r\Theta_2\sin\theta) \bm j-r\Theta_1\bm k
\end{array}$$
 From this we obtain the contra-variant components of  $\bm n$ in the cartesian and the new coordinate
 system,
$$\left\{\begin{array}{ll}
n_x=-\frac1{\sqrt{a}}(\sin\theta+x^2\Theta_2\cos\theta),\quad
n_y=\frac1{\sqrt{a}}(\cos\theta-x^2\Theta_2\sin\theta),\\
n_z=-x^2\Theta_1/\sqrt{a}.\\
 n^\alpha=-x^2\Theta_\alpha/\sqrt{a},\quad
 n^3=(r\varepsilon)^{-1}\sqrt{a},
\end{array}\right.\eqno{(A.I.17)}$$

Now we calculate the curvature tensor of the surface $\Im_\xi$.
Noting that
 $$b_{\alpha\beta}=-\frac12(\bm n_\alpha\bm e_\beta+
 \bm n_\beta\bm e_\alpha)=\bm n\bm e_{\alpha\beta}=\frac1{\sqrt{a}}
 \bm e_1\times \bm e_2\cdot \bm e_{\alpha\beta},\eqno{(A.I.18)}$$
If the radial vector at the point $P$ on $\Im_\xi$ is denoted by
$\bm\Re$, then
 $$\bm e_{\alpha\beta}=\frac{\partial^2\bm\Re}{\partial
 x^\alpha\partial x^\beta}=x_{\alpha\beta}\bm i+y_{\alpha\beta}\bm j+z_{\alpha\beta}\bm k,$$
where $x_{\alpha\beta}=\partial_\alpha\partial_\beta x$. Therefore
$$\begin{array}{ll}
\sqrt{a}b_{\alpha\beta}&=\left |\begin{array}{lll}
x_{\alpha\beta}& y_{\alpha\beta}& z_{\alpha\beta}\\
(\bm e_1)_x&  (\bm e_1)_y & (\bm e_1)_z\\
(\bm e_2)_x& (\bm e_2)_y & (\bm e_2)_z
\end{array}\right|
=\left |\begin{array}{lll}
x_{\alpha\beta}& y_{\alpha\beta}& 0\\
-r\Theta_1\sin\theta& r\Theta_1\cos\theta & 1\\
\cos\theta-r\Theta_2\sin\theta& \sin\theta+r\Theta_2\cos\theta & 0
\end{array}\right|\\[15pt]
&=-[(x_{\alpha\beta}\sin\theta-y_{\alpha\beta}\cos\theta)+r\Theta_2(x_{\alpha\beta}\cos\theta+y_{\alpha\beta}\sin\theta)]\\
&=-[(x_{\alpha\beta}+r\Theta_2y_{\alpha\beta})\sin\theta+(r\Theta_2x_{\alpha\beta}-y_{\alpha\beta})\cos\theta],
\end{array} \eqno{(A.I.19)}$$ Simply calculation from (A.I.8-A.I.10) shows that
$$\begin{array}{ll}
x_{11}=\frac{\partial^2x}{\partial
(x^1)^2}=-x^2(\Theta_{11}\sin\theta+\Theta^2_1\cos\theta),\\
x_{12}=-\Theta_1\sin\theta-x^2(\Theta_{12}\sin\theta+\Theta_1\Theta_2\cos\theta),\\
x_{22}=-2\Theta_2\sin\theta-x^2(\Theta_{22}\sin\theta+\Theta^2_2\cos\theta),\\
y_{11}=x^2(\Theta_{11}\cos\theta-\Theta_1^2\sin\theta),\quad
y_{12}=\Theta_1\cos\theta+x^2(\Theta_{12}\cos\theta-\Theta_1\Theta_2\sin\theta),\\
y_{22}=2\Theta_2\cos\theta+x^2(\Theta_{22}\cos\theta-\Theta^2_2\sin\theta),\\
z_{\alpha\beta}=0,
\end{array}\eqno{(A.I.20)}$$ Substituting (A.I.20) into (A.I.19) leads to

$$\begin{array}{ll}
b_{11}=\frac1{\sqrt{a}}(\Theta_2(a_{11}-1)+x^2\Theta_{11}).\quad
b_{12}=b_{21}=\frac1{\sqrt{a}}(\Theta_1a_{12}+x^2\Theta_{12}),\\
b_{22}=\frac{1}{\sqrt{a}}(\Theta_2(a_{22}+1)+x^2\Theta_{22}),\quad
b=\det(b_{\alpha\beta})=b_{11}b_{22}-b_{12}^2,
\end{array}\eqno{(A.I.21)}$$

Finally, the mean curvature and the Gaussian curvature are
calculated as
$$\left\{\begin{array}{ll}
K&=\frac ba,\\[5pt]
2H&=\frac1{a\sqrt{a}}[x^2(a_{22}\Theta_{11}+a_{11}\Theta_{22})-2a_{12}\Theta_{12})
\Theta_2(2a_{11}a_{22}+a_{11}-a_{22})-2\Theta_1a^2_{12},
\end{array}\right.\eqno{(A.I.22)}$$ Those are (A.I.4)(A.I.5).
\end{proof}

In the next place, we consider the Christoffel symbols and the
covariant derivatives under the new coordinate system.

\begin{proposition} Under the new curvilinear coordinate system
$(x^\alpha,\xi)$, the Christoffel symbols and covariant derivatives
are respectively given by
$$
\left\{%
\begin{array}{l}
\Gamma^\alpha_{\beta\gamma}=-r\delta_{2\alpha}\Theta_\beta\Theta_\gamma,
\quad\Gamma^\alpha_{3\beta}= -\varepsilon r\delta_{2\alpha}\Theta_\beta, \\
\Gamma^3_{\alpha\beta}=\varepsilon^{-1}r^{-1}(\delta_{2\alpha}\delta^%
\lambda_\beta+ \delta_{2\beta}\delta^\lambda_\alpha
)\Theta_\lambda+\varepsilon^{-1}\Theta_{\alpha\beta}+\varepsilon^{-1}r%
\Theta_2\Theta_\alpha\Theta_\beta, \\
\Gamma^3_{3\alpha}=\Gamma^3_{\alpha
3}=r^{-1}\delta_{2\alpha}+r\Theta_2\Theta_\alpha\quad
\Gamma^\alpha_{33}=-\varepsilon^2
r\delta_{2\alpha},\quad\Gamma^3_{33}=\varepsilon r\Theta_2,%
\end{array}%
\right.\eqno{(A.I.23)}
$$ and
$$
\left\{%
\begin{array}{l}
\nabla_\alpha w^\beta=\frac{\partial w^\beta}{\partial x^\alpha}
-r\delta^\beta_{2}\Theta_\alpha\Pi(w,\Theta), \\
\nabla_\alpha w^3=\frac{\partial w^3}{\partial x^\alpha}%
+\varepsilon^{-1}(x^2)^{-1}w^2\Theta_\alpha+
\varepsilon^{-1}w^\beta\Theta_{\alpha\beta}+(\varepsilon
x^2)^{-1}a_{2\alpha}\Pi(w,\Theta), \\
\nabla_3 w^\alpha=\frac{\partial
w^\alpha}{\partial\xi}-x^2\varepsilon \delta_{2\alpha}\Pi(w,\Theta),
\quad
\nabla_3 w^3=\frac{\partial w^3}{\partial \xi}+\frac{w^2}{x^2}%
+x^2\Theta_2\Pi(w,\Theta), \\
\mbox{div} \bm w=\frac{\partial w^\alpha}{\partial x^\alpha}+\frac{w^2}{x^2}+%
\frac{\partial w^3}{\partial\xi},\quad \Pi(w,\Theta)=\varepsilon
w^3+w^\beta\Theta_\beta.%
\end{array}%
\right.\eqno{(A.I.24)}$$
\end{proposition}

\begin{proof} From (A.I.13) it  follows that
$$\begin{array}{ll}
\bm e_{ij}=\partial_i\bm e_j,\\
\bm e_{11}=-x^2(\cos\theta\Theta_1^2+\sin\theta\Theta_{11})\bm
i+x^2(-\sin\theta\Theta_1^2
+\cos\theta\Theta_{11})\bm j,\\
\bm e_{12}=\bm
e_{21}=(-\sin\theta\Theta_1-x^2(\cos\theta\Theta_1\Theta_2+\sin\theta\Theta_{12}))
\bm i.\\
\quad\quad\quad+(\cos\theta\Theta_1+x^2(-\sin\theta\Theta_1\Theta_2+
+\cos\theta\Theta_{12}))\bm j,\\
\bm
e_{22}=(-2\sin\theta\Theta_2-x^2(\cos\theta\Theta_2\Theta_2+\sin\theta\Theta_{22}))
\bm i+(2\cos\theta\Theta_1+x^2(-\sin\theta\Theta_2\Theta_2+
+\cos\theta\Theta_{22}))\bm j,\\
\bm e_{13}=\bm e_{31}=-(r\varepsilon)\Theta_1(\cos\theta \bm
i+\sin\theta \bm j),\quad \bm e_{33}=-r\varepsilon^2(\cos\theta \bm
i+\sin\theta
\bm j),\\
\bm e_{23}=\bm
e_{32}=-\varepsilon(\sin\theta+r\Theta_2\cos\theta)\bm i
+\varepsilon(\cos\theta-r\Theta_2\sin\theta)\bm i,\\
\Gamma^i_{jk}=\bm e^i\bm e_{jk},\\
 \Gamma^1_{11}=\bm e^1\bm e_{11}=\bm k\bm e_{11}=0,\quad
\Gamma^2_{11}=-x^2\Theta^2_1,\quad
\Gamma^3_{11}=(\varepsilon)^{-1}(r\Theta_2\Theta^2_1+\Theta_{11}),\\
\Gamma^1_{12}=\Gamma^1_{21}=0,\quad
\Gamma^2_{12}=\Gamma^1_{21}=-r\Theta_1\Theta_2,\quad
\Gamma^3_{12}=\Gamma^3_{21}=(r\varepsilon)^{-1}(\Theta_1a_{22}+r\Theta_{12}),\\
\Gamma^1_{22}=0,\quad \Gamma^2_{22}=-r\Theta^2_2,\quad
\Gamma^3_{22}=(r\varepsilon)^{-1}[2\Theta_2+r(\Theta_{22}+r\Theta_2^3)]=(r\varepsilon)^{-1}
[\Theta_2(1+a_{22})+r\Theta_{22}],\\
\Gamma^1_{13}=\Gamma^1_{31}=0,\quad
\Gamma^2_{13}=\Gamma^2_{31}=-(r\varepsilon)^{-1}\Theta_1,\quad
\Gamma^3_{13}=\Gamma^3_{31}=r\Theta_1\Theta_2,\\
\Gamma^1_{23}=\Gamma^1_{32}=0,\quad
\Gamma^2_{23}=\Gamma^2_{32}=-r\varepsilon\Theta_2,\quad \Gamma^3_{23}=\Gamma^2_{32}=r^{-1}a_{22},\\
\Gamma^1_{33}=0,\quad \Gamma^2_{33}=-r\varepsilon^2,\quad
\Gamma^3_{33}=-r\varepsilon\Theta_2,\\
\end{array}$$ This yields (A.I.23). The (A.I.24) can be obtain from (
A.I.23) and $$\nabla_iw^j=\partial_iw^j+\Gamma^j_{ik}w^k.$$ This
ends the proof.
\end{proof}

\subsection{The Navier-Stokes Equation In the New Coordinate System}

\begin{proposition} The Rotating Navier-stokes equations in the new coordinate
system $(x^\alpha, \xi)$ can be written as
$$\left\{\begin{array}{ll}
\frac{\partial w^\alpha}{\partial x^\alpha}+\frac{\partial
w^3}{\partial\xi}+\frac{w^2}{r}=0,\\
{\cal N}^\alpha(\bm w,p,\Theta):=-\nu\widetilde{\Delta}
w^\alpha+\nabla_\alpha p+C^\alpha(\bm w,\bm \omega)-\nu l^\alpha(\bm
w,\Theta)+\frac{\partial}{\partial\xi}(-\nu
l^\alpha_\xi(\bm w,\Theta)\\
\qquad-\varepsilon^{-1}\Theta_\alpha p)+N^\alpha_x(\bm w,p)+{\cal
N}_\xi^\alpha(\bm w,p)
  =f^\alpha,\\
  {\cal N}^3(\bm w,p,\Theta):=-\nu\widetilde{\Delta}
w^3-\varepsilon^{-1}\Theta_\alpha\frac{\partial p}{\partial
x^\alpha}+C^3(\bm w,\bm \omega)-\nu
l^3(\bm w,\Theta)\\
\qquad+\frac{\partial}{\partial\xi}(-\nu l^3_\xi(\bm w,\Theta)
  +(r\varepsilon)^{-2}ap)
  +N^3_x(\bm w,p)+{\cal N}^3_\xi(\bm w,p)=f^3,
\end{array}\right.\eqno{(A.II.1)}$$
where $\bm C(\bm w,\bm \omega)$ is Coriolis forces defined in
(A.I.4), and the other shortening symbols are definitely expressed
as, respectively,
$$\left\{\begin{array}{ll}
  N^\alpha_x(\bm w,\bm w)&=w^\beta\frac{\partial
  w^\alpha}{\partial
  x^\beta}-r\delta_{2\alpha}\Pi(\bm w,\Theta)\Pi(\bm w,\Theta),\quad
 {\cal N} ^\alpha_\xi(\bm w,p)=w^3\frac{\partial w^\alpha}{\partial\xi},\\
N^3_x(\bm w,\bm w)&=w^\beta\frac{\partial
  w^3}{\partial x^\beta}+\varepsilon^{-1}w^\beta
  w^\lambda\Theta_{\beta\lambda}+(r\varepsilon)^{-1}\Pi(\bm w,\Theta)(2w^2
  +r^2\Theta_2\Pi(\bm w,\Theta)),\\
 {\cal N}^3_\xi(\bm w,p)&=w^3\frac{\partial
  w^3}{\partial\xi},\\
\end{array}\right.\eqno{(A.II.2)}$$
and
$$\left\{\begin{array}{ll}
l^\alpha(\bm
w,\Theta)&=-2r\varepsilon\delta_{2\alpha}\Theta_\lambda\frac{\partial
w^3}{\partial x^\lambda}+\frac1r\frac{\partial w^\alpha}{\partial
r}-\frac{w^2}{r^2}\delta_{2\alpha}
+\delta_{2\alpha}(r\widetilde{\Delta}\Theta-2a\Theta_2)\Pi(\bm w,\Theta)\\
&+B^\alpha_\sigma(\Theta)w^\sigma,\\
l^\alpha_\xi(\bm w,\Theta)&=(r\varepsilon)^{-2}a\frac{\partial
w^\alpha}{\partial\xi}-2\varepsilon^{-1}\Theta_\beta\frac{\partial
w^\alpha}{\partial x^\beta}\\
&-[(r\varepsilon)^{-1}(\delta_{\alpha\lambda}\Theta_2+2\delta_{2\alpha}\Theta_\lambda)
+\varepsilon^{-1}\delta_{\alpha\lambda}\widetilde{\Delta}\Theta]w^\lambda
-2r^{-1}\delta_{2\alpha}w^3,\\
B^\alpha_\sigma(\Theta)&=\delta_{2\alpha}[2(\delta_{2\sigma}|\widetilde{\nabla}
\Theta|^2-r\Theta_\lambda\Theta_{\lambda\sigma})],\\
\end{array}\right.\eqno{(A.II.3)}$$
$$\left\{\begin{array}{ll}
l^3(\bm
w,\Theta)&=(r\varepsilon)^{-1}(\delta_{2\beta}\Theta_\lambda+r\Theta_{\beta\lambda})\frac{\partial
w^\lambda}{\partial x^\beta}+\frac2r \frac{\partial w^3}{\partial
r}+\frac{\partial}{\partial
x^\beta}(\varepsilon^{-1}\Theta_{\beta\sigma}w^\sigma)\\
&+a\Theta_2\Theta_2w^3
+B^3_\sigma(\Theta)w^\sigma,\\
l^3_\xi(\bm w,\Theta)&=(r\varepsilon)^{-2}a\frac{\partial
w^3}{\partial\xi}-2\varepsilon^{-1}\Theta_\beta\frac{\partial
w^3}{\partial
x^\beta}+2\varepsilon^{-2}(r^{-3}\delta_{2\sigma}-\Theta_\beta\Theta_{\beta\sigma})w^\sigma\\
&+\varepsilon^{-1}(r\Theta_2|\widetilde{\nabla}\Theta|^2-\widetilde{\Delta}\Theta)w^3,\\
B^3_\sigma(\Theta)&=(r\varepsilon)^{-1}[(r^{-1}+ra\Theta_2\Theta_2)\Theta_\sigma+2\Theta_{2\sigma}),
\end{array}\right.\eqno{(A.II.4)}$$

$$\left\{\begin{array}{ll}
\Theta_\alpha=\frac{\partial\Theta}{\partial x^\alpha},\quad
\Theta_{\alpha\beta}=\frac{\partial^2\Theta}{\partial
x^\alpha\partial x^\beta},\quad \Pi(\bm w,\Theta)=\varepsilon
w^3+w^\lambda\Theta_\lambda,\\
\widetilde{\Delta}\Theta=\Theta_{\alpha\alpha}=\Theta_{11}+\Theta_{22},\quad
|\widetilde{\nabla}\Theta|^2=\Theta_1^2+\Theta_2^2,\\
\end{array}\right.\eqno{(A.II.5)}$$
\end{proposition}

\begin{proof} Firstly, we derive the Trace Laplacian operator in the New
Coordinate System, that is,
$$\left\{\begin{array}{ll}
\Delta w^\alpha
=g^{ij}\nabla_i\nabla_jw^\alpha=\widetilde{\Delta}w^\alpha
+\frac{\partial}{\partial\xi}l^\alpha_\xi(\bm w,\Theta)+l^\alpha(\bm w,\Theta),\\
 \Delta w^3=g^{ij}\nabla_i\nabla_jw^3=\widetilde{\Delta} w^3+\partial_\xi
l^3_\xi(\bm w,\Theta)+l^3(\bm w,\Theta).
\end{array}\right.\eqno{(A.II.6)}$$
In fact, we have
$$\Delta
w^\alpha=g^{ij}\nabla_i\nabla_jw^\alpha=\nabla_\beta\nabla_\beta
w^\alpha-\varepsilon^{-1}\Theta_\beta(\nabla_3\nabla_\beta
w^\alpha+\nabla_\beta\nabla_3w^\alpha)
+(r\varepsilon)^{-2}a\nabla_3\nabla_3w^\alpha.$$ Employing (A.I.23)
and (A.I.24) we claim that
$$\begin{array}{rcl}
\nabla_\beta\nabla_\beta w^\alpha&=&\partial_\beta(\nabla_\beta
w^\alpha)-\Gamma^m_{\beta\beta}\nabla_mw^\alpha+\Gamma^\alpha_{\beta
m}\nabla_\beta w^m\\[5pt]
&=&\partial_\beta(\frac{\partial w^\alpha}{\partial
x^\beta}-r\delta_{2\alpha}\Theta_\beta\Pi(w,\Theta))+(\Gamma^\alpha_{\beta\sigma}\delta^\lambda_{\beta}
-\Gamma^\lambda_{\beta\beta}\delta_{\alpha\sigma})
\nabla_\lambda w^\sigma\\[5pt]
&&+\Gamma^\alpha_{3\beta}\nabla_\beta
w^3-\Gamma^3_{\beta\beta}\nabla_3w^\alpha,\\[5pt]
\nabla_\beta\nabla_3w^\alpha&=&\partial_\beta\nabla_3w^\alpha-\Gamma^\lambda_{3\beta}\nabla_\lambda
w^\alpha+(\Gamma^\alpha_{\lambda\beta}-\Gamma^3_{3\beta}\delta_{\alpha\lambda})\nabla_3w^\lambda
+\Gamma^\alpha_{3\beta}\nabla_3w^3,\\[5pt]
\nabla_3\nabla_\beta w^\alpha&=&\partial_\xi\nabla_\beta
w^\alpha+(\Gamma^\alpha_{3\sigma}\delta_{\beta\lambda}-\Gamma^\lambda_{3\beta}\delta_{\alpha\sigma})\nabla_\lambda
w^\sigma-\Gamma^3_{3\beta}\nabla_3w^\alpha
+\Gamma^\alpha_{33}\nabla_\beta w^3,\\[5pt]
\nabla_3\nabla_3w^\alpha&=&\partial_\xi\nabla_3w^\alpha+(\Gamma^\alpha_{3\lambda}-\Gamma^3_{33}\delta_{\alpha\lambda})
\nabla_3 w^\lambda+\Gamma^\alpha_{33}\nabla_3w^3
-\Gamma^\beta_{33}\nabla_\beta w^\alpha,\\[5pt]
&&-\varepsilon^{-1}\Theta_\beta(\nabla_3\nabla_\beta
w^\alpha+\nabla_\beta\nabla_3w^\alpha)
+(r\varepsilon)^{-2}a\nabla_3\nabla_3w^\alpha\\[5pt]
&=&(r\varepsilon)^{-2}a\partial_\xi\nabla_3w^\alpha
-\varepsilon^{-1}\Theta_\beta(\partial_\beta\nabla_3w^\alpha+\partial_\xi\nabla_\beta
w^\alpha)\\[5pt]
&&+[-(r\varepsilon)^{-2}a\Gamma^\lambda_{33}\delta_{\alpha\sigma}+\varepsilon^{-1}\Theta_\beta
(2\Gamma^\lambda_{3\beta}\delta_{\alpha\sigma}-\Gamma^\alpha_{3\sigma}\delta_{\beta\lambda})]\nabla_\lambda
w^\sigma\\[5pt]
&&+[(r\varepsilon)^{-2}a(\Gamma^\alpha_{3\lambda}-\Gamma^3_{33}\delta_{\alpha\lambda})
+\varepsilon^{-1}\Theta_\beta(2\Gamma^3_{3\beta}\delta_{\alpha\lambda}-\Gamma^\alpha_{\lambda\beta})]\nabla_3w^\lambda\\[5pt]
&&-\varepsilon^{-1}\Theta_\beta\Gamma^\alpha_{33}\nabla_\beta w^3
+[(r\varepsilon)^{-2}a\Gamma^\alpha_{33}
-\varepsilon^{-1}\Theta_\beta\Gamma^\alpha_{3\beta}]\nabla_3w^3
\end{array}$$
and
$$\begin{array}{ll}
(r\varepsilon)^{-2}a\partial_\xi\nabla_3w^\alpha
-\varepsilon^{-1}\Theta_\beta(\partial_\beta\nabla_3w^\alpha+\partial_\xi\nabla_\beta
w^\alpha)\\
\hspace{2cm}=\frac{\partial}{\partial\xi}[(r\varepsilon)^{-2}a\frac{\partial
w^\alpha}{\partial\xi}-(r\varepsilon)^{-1}\delta_{2\alpha}\Pi(\bm
w,\Theta)-2\varepsilon^{-1}\Theta_\beta\frac{\partial
w^\alpha}{\partial
x^\beta}]+\Theta_\beta\delta_{2\alpha}\frac{\partial}{\partial x^\beta}(r\Pi(\bm w,\Theta)),\\
\end{array}$$
Summing up the above results, we get
$$\begin{array}{rcl}
\Delta w^\alpha &=&\partial_\beta\partial_\beta
w^\alpha+\frac{\partial}{\partial\xi}[(r\varepsilon)^{-2}a\frac{\partial
w^\alpha}{\partial\xi}-(r\varepsilon)^{-1}\delta_{2\alpha}\Pi(\bm
w,\Theta)-2\varepsilon^{-1}\Theta_\beta\frac{\partial
w^\alpha}{\partial
x^\beta}]\\
&&-\delta_{2\alpha}r\widetilde{\Delta}\Theta\Pi(\bm w,\Theta)
+I_1\nabla_\lambda w^\sigma+I_2\nabla_3w^\lambda+I_3\nabla_\lambda
w^3+I_4\nabla_3w^3\\
&=&\widetilde{\Delta}w^\alpha+\frac{\partial}{\partial\xi}[(r\varepsilon)^{-2}a\frac{\partial
w^\alpha}{\partial\xi}-(r\varepsilon)^{-1}\delta_{2\alpha}\Pi(\bm
w,\Theta)-2\varepsilon^{-1}\Theta_\beta\frac{\partial
w^\alpha}{\partial
x^\beta}+I_2w^\lambda+I_4w^3]\\
&=&+I_1\frac{\partial w^\sigma}{\partial
x^\lambda}+I_3\frac{\partial w^3}{\partial
x^\lambda}+[-r\delta_{2\sigma}\Theta_\lambda
I_1-r\varepsilon\delta_{2\lambda}I_2+(r\varepsilon)^{-1}a_{2\lambda}I_3+r\Theta_2I_4]\Pi(\bm w,\Theta)\\
&=&+[\varepsilon^{-1}(\Theta_{\lambda\sigma}-r^{-1}\Theta_\lambda\delta_{2\sigma})I_3+r^{-1}\delta_{2\sigma}I_4]w^\sigma,
\end{array}\eqno{(A.II.7)}$$
where some marked symbols are
$$\begin{array}{rcl}
I_1&=&[\Gamma^\alpha_{\lambda\sigma}
-\Gamma^\lambda_{\beta\beta}\delta_{\alpha\sigma}
+\varepsilon^{-1}\Theta_\beta(2\Gamma^\lambda_{3\beta}\delta_{\alpha\sigma}
-\Gamma^\alpha_{3\sigma}\delta_{\beta\lambda})
-(r\varepsilon)^{-2}a\Gamma^\lambda_{33}\delta_{\alpha\sigma}]\\[5pt]
&=&-r\delta_{2\alpha}\Theta_\lambda\Theta_\sigma
+r\delta_{2\lambda}|\widetilde{\nabla}\Theta|^2\delta_{\alpha\sigma}+\varepsilon^{-1}\Theta_\beta
(-r\varepsilon)(2\delta_{2\lambda}\Theta_\beta\delta_{\alpha\sigma}
-\delta_{2\alpha}\Theta_\sigma\delta_{\beta\lambda})\\[5pt]
&&-(r\varepsilon)^{-2}a(-r\varepsilon^2)\delta_{2\lambda}\delta_{\alpha\sigma}
=r^{-1}\delta_{2\lambda}\delta_{\alpha\sigma},\\[5pt]
I_2&=&-\Gamma^3_{\beta\beta}\delta_{\alpha\lambda}+\varepsilon^{-1}\Theta_\beta(2\Gamma^3_{3\beta}\delta_{\alpha\lambda}
-\Gamma^\alpha_{\beta\lambda})
+(r\varepsilon)^{-2}a(\Gamma^\alpha_{3\lambda}-\Gamma^3_{33}\delta_{\alpha\lambda})\\[5pt]
&=&-[(r\varepsilon)^{-1}(a_{2\beta}\Theta_\beta+\Theta_2)
+\varepsilon^{-1}\widetilde{\Delta}\Theta]\delta_{\alpha\lambda}\\[5pt]
&&+\varepsilon^{-1}\Theta_\beta(2r^{-1}a_{2\beta}\delta_{\alpha\lambda}+r\delta_{2\alpha}\Theta_\lambda\Theta_\beta)
+(r\varepsilon)^{-2}a(-r\varepsilon\delta_{2\alpha}\Theta_\lambda-r\varepsilon\Theta_2\delta_{\alpha\lambda})\\[5pt]
&=&-(r\varepsilon)^{-1}(\delta_{\alpha\lambda}\Theta_2+\delta_{2\alpha}\Theta_\lambda)
-\varepsilon^{-1}\delta_{\alpha\lambda}\widetilde{\Delta}\Theta,\\[5pt]
I_3&=&\Gamma^\alpha_{3\lambda}-\varepsilon^{-1}\Theta_\lambda\Gamma^\alpha_{33}
=-r\varepsilon\delta_{2\alpha}\Theta_\lambda
-\varepsilon^{-1}\Theta_\lambda(-\varepsilon^2 r\delta_{2\alpha})
=-2r\varepsilon\delta_{2\alpha}\Theta_\lambda,\\[5pt]
I_4&=&(r\varepsilon)^{-2}a\Gamma^\alpha_{33}-\varepsilon^{-1}\Theta_\beta\Gamma^\alpha_{3\beta}=
(r\varepsilon)^{-2}a(-r\varepsilon^2\delta_{2\alpha})-\varepsilon^{-1}\Theta_\beta
(-r\varepsilon\delta_{2\alpha}\Theta_\beta)=-r^{-1}\delta_{2\alpha},\end{array}$$
Therefore
$$\begin{array}{ll}
-r\delta_{2\sigma}\Theta_\lambda
I_1-r\varepsilon\delta_{2\lambda}I_2+(r\varepsilon)^{-1}a_{2\lambda}I_3+r\Theta_2I_4=(r\widetilde{\Delta}\Theta-2a\Theta_2)\delta_{2\alpha},\\
\varepsilon^{-1}(\Theta_{\lambda\sigma}-r^{-1}\Theta_\lambda\delta_{2\sigma})I_3+r^{-1}\delta_{2\sigma}I_4=[r^{-2}(2a-3)\delta_{2\sigma}-2r\Theta_\lambda\Theta_{\lambda\sigma}]\delta_{2\alpha}.
\end{array}$$
Substituting the above expression into (A.II.7), then
$$\begin{array}{rl}
\Delta w^\alpha
=&\widetilde{\Delta}w^\alpha+\frac{\partial}{\partial\xi}[(r\varepsilon)^{-2}a\frac{\partial
w^\alpha}{\partial\xi}-(r\varepsilon)^{-1}\delta_{2\alpha}\Pi(w,\Theta)-2\varepsilon^{-1}\Theta_\beta\frac{\partial
w^\alpha}{\partial
x^\beta}+I_2w^\lambda+I_4w^3]\\
&+r^{-1}\frac{\partial w^\alpha}{\partial
r}-2r\varepsilon\delta_{2\alpha}\Theta_\lambda\frac{\partial
w^3}{\partial x^\lambda}+[(r\widetilde{\Delta}\Theta-2a\Theta_2)\delta_{2\alpha}]\Pi(w,\Theta)\\
&+[r^{-2}(2a-3)\delta_{2\sigma}-2r\Theta_\lambda\Theta_{\lambda\sigma}]\delta_{2\alpha}w^\sigma.
\end{array}\eqno{(A.II.8)}$$
(A.II.8) can be rewritten in a splitting form, that is
$$\begin{array}{ll}
\Delta w^\alpha &=\widetilde{\Delta}w^\alpha
+\frac{\partial}{\partial\xi}l^\alpha_\xi(\bm w,\Theta)+l^\alpha(\bm w,\Theta),\\
\end{array}\eqno{(A.II.9)}$$
where
$l^\alpha(w,\Theta),l^\alpha_\xi(w,\Theta),B^\alpha_\sigma(\Theta)$
are formulated in (A.II.3).
%$$\left\{\begin{array}{ll}
%l^\alpha(w,\Theta)&=-2r\varepsilon\delta_{2\alpha}\Theta_\lambda\frac{\partial
%w^3}{\partial x^\lambda}+\frac1r\frac{\partial w^\alpha}{\partial
%r}-\frac{w^2}{r^2}\delta_{2\alpha}\\
%&+\delta_{2\alpha}\varepsilon(r\widetilde{\Delta}\Theta-2a\Theta_2)w^3
%+B^\alpha_\sigma(\Theta)w^\sigma,\\
%l^\alpha_\xi(w,\Theta)&:=(r\varepsilon)^{-2}a\frac{\partial
%w^\alpha}{\partial\xi}-2\varepsilon^{-1}\Theta_\beta\frac{\partial
%w^\alpha}{\partial x^\beta}\\
%&-[(r\varepsilon)^{-1}(\delta_{\alpha\lambda}\Theta_2+2\delta_{2\alpha}\Theta_\lambda)
%+\varepsilon^{-1}\delta_{\alpha\lambda}\widetilde{\Delta}\Theta]w^\lambda
%-2r^{-1}\delta_{2\alpha}w^3,\\
%B^\alpha_\sigma(\Theta)&=\delta_{2\alpha}[2(\delta_{2\sigma}|\widetilde{\nabla}
%\Theta|^2-r\Theta_\lambda\Theta_{\lambda\sigma})+
%(r\widetilde{\Delta}\Theta-2a\Theta_2)\Theta_\sigma],\\
%\end{array}\right.\eqno{(A.II.10)}$$

Our task is now to prove the second equality of (A.II.6). Indeed,
$$\Delta w^3=\nabla_\beta\nabla_\beta
w^3-\varepsilon^{-1}\Theta_\beta(\nabla_\beta\nabla_3
w^3+\nabla_3\nabla_\beta w^3)+(r\varepsilon)^{-2}a
\nabla_3\nabla_3w^3.$$ An argument similar to the one used in proof
of the first equality of (A.II.6) shows that
$$\begin{array}{rcl}
\Delta w^3&=&\partial_\beta(\nabla_\beta
w^3)+\Gamma^3_{\beta\lambda}\nabla_\beta
w^\lambda+(\Gamma^3_{3\beta}-\Gamma^\beta_{\lambda\lambda})\nabla_\beta
w^3-\Gamma^3_{\beta\beta}\nabla_3w^3\\
&&-\varepsilon^{-1}\Theta_\beta[\partial_\beta\nabla_3w^3
+\Gamma^3_{\beta\lambda}\nabla_3w^\lambda+\Gamma^3_{3\beta}
\nabla_3w^3-\Gamma^\lambda_{3\beta}\nabla_\lambda
w^3-\Gamma^3_{3\beta}\nabla_3w^3\\
&&+\partial_\xi\nabla_\beta w^3+\Gamma^3_{3\lambda}\nabla_\beta
w^\lambda+(\Gamma^3_{33}\delta_{\beta\lambda}-\Gamma^\lambda_{3\beta})\nabla_\lambda
w^3-\Gamma^3_{3\beta}\nabla_3w^3]\\
&&+(r\varepsilon)^{-2}a[\partial_\xi\nabla_3w^3+\Gamma^3_{3\lambda}\nabla_3w^\lambda+\Gamma^3_{33}
\nabla_3w^3-\Gamma^\lambda_{33}\nabla_\lambda
w^3-\Gamma^3_{33}\nabla_3w^3],
\end{array}$$
and
$$\begin{array}{rcl}
\Delta w^3&=&\partial_\beta(\nabla_\beta
w^3)-\varepsilon^{-1}\Theta_\beta\partial_\beta\nabla_3w^3+\frac{\partial}{\partial\xi}[(r\varepsilon)^{-2}a
\nabla_3w^3 -\varepsilon^{-1}\Theta_\beta\nabla_\beta
w^3]\\
&&+J_1\nabla_\beta w^\lambda+J_2\nabla_\beta
w^3+J_3\nabla_3w^\lambda+J_4\nabla_3w^3,\\
\end{array}\eqno{(A.II.10)}$$
where (by (A.I.23))
$$\left\{\begin{array}{ll}
J_1=(\Gamma^3_{\beta\lambda}-\varepsilon^{-1}\Theta_\beta\Gamma^3_{3\lambda})=(r\varepsilon)^{-1}
(\delta_{2\beta}\Theta_\lambda+r\Theta_{\beta\lambda}),\\
J_2=\Gamma^3_{3\beta}-\Gamma^\beta_{\lambda\lambda}-(r\varepsilon)^{-2}a\Gamma^\beta_{33}
+\varepsilon^{-1}\Theta_\lambda(-\Gamma^3_{33}\delta_{\beta\lambda}+2\Gamma^\beta_{3\lambda})
=2r^{-1}\delta_{2\beta},\\
J_3=(r\varepsilon)^{-2}a\Gamma^3_{3\lambda}-\varepsilon^{-1}\Theta_\beta\Gamma^3_{\beta\lambda}
=\varepsilon^{-2}r^{-3}(\delta_{2\lambda}-r^3\Theta_{\beta}\Theta_{\beta\lambda}),\\
J_4=-\Gamma^3_{\beta\beta}+2\varepsilon^{-1}\Theta_\beta\Gamma^3_{3\beta}=(\varepsilon)^{-1}(r\Theta_2
|\widetilde{\nabla}\Theta|^2-\widetilde{\Delta}\Theta)
\end{array}\right.\eqno{(A.II.11)}$$
Some further calculations show that
$$\begin{array}{ll}
\partial_\beta(\nabla_\beta
w^3)-\varepsilon^{-1}\Theta_\beta\partial_\beta\nabla_3w^3+\frac{\partial}{\partial\xi}[(r\varepsilon)^{-2}a
\nabla_3w^3 -\varepsilon^{-1}\Theta_\beta\nabla_\beta
w^3]\\
\quad=\partial_\beta(\frac{\partial w^3}{\partial
x^\beta}+(r\varepsilon)^{-1}[(\delta_{\alpha2}\Theta_\beta+r\Theta_{\alpha\beta})w^\alpha
+a_{2\beta}\Pi(w,\Theta)])\\
\quad\quad-\varepsilon^{-1}\Theta_\beta\partial_\beta(\frac{\partial
w^3}{\partial\xi}+\frac{w^2}{r}+r\Theta_2\Pi(w,\Theta))
+\frac{\partial}{\partial\xi}[(r\varepsilon)^{-2}a \nabla_3w^3
-\varepsilon^{-1}\Theta_\beta\nabla_\beta
w^3]\\
\quad=\widetilde{\Delta}
w^3+\frac{\partial}{\partial\xi}[(r\varepsilon)^{-2}a \nabla_3w^3
-\varepsilon^{-1}\Theta_\beta\nabla_\beta
w^3-\varepsilon^{-1}\Theta_\beta\frac{\partial w^3}{\partial
x^\beta}]\\
\quad\quad+\frac{\partial}{\partial
x^\beta}((r\varepsilon)^{-1}[(\delta_{\alpha2}\Theta_\beta+r\Theta_{\alpha\beta})w^\alpha
+a_{2\beta}\Pi(w,\Theta)]-\varepsilon^{-1}\Theta_\beta(\frac{w^2}{r}+r\Theta_2\Pi(w,\Theta)))\\
\quad\quad+\varepsilon^{-1}\widetilde{\Delta}\Theta(\frac{w^2}{r}+r\Theta_2\Pi(w,\Theta)).
\end{array}$$
Combining like terms, we get
$$\begin{array}{ll}
\partial_\beta(\nabla_\beta
w^3)-\varepsilon^{-1}\Theta_\beta\partial_\beta\nabla_3w^3+\frac{\partial}{\partial\xi}[(r\varepsilon)^{-2}a
\nabla_3w^3 -\varepsilon^{-1}\Theta_\beta\nabla_\beta
w^3]\\
\quad=\widetilde{\Delta}
w^3+\frac{\partial}{\partial\xi}[(r\varepsilon)^{-2}a \nabla_3w^3
-\varepsilon^{-1}\Theta_\beta\nabla_\beta
w^3-\varepsilon^{-1}\Theta_\beta\frac{\partial w^3}{\partial
x^\beta}]\\
\quad\quad+\frac{\partial}{\partial
x^\beta}(\varepsilon^{-1}\Theta_{\alpha\beta}w^\alpha)
+\varepsilon^{-1}\widetilde{\Delta}\Theta(\frac{w^2}{r}+r\Theta_2\Pi(w,\Theta)).
\end{array}$$
On the other hand, by using (A.I.24), then
$$\begin{array}{ll}
J_1\nabla_\beta w^\lambda+J_2\nabla_\beta
w^3+J_3\nabla_3w^\lambda+J_4\nabla_3w^3 =J_1\partial_\beta
w^\lambda+J_2\partial_\beta
w^3+\partial_\xi(J_3 w^\lambda+J_4 w^3)\\
\quad\quad+[-r\delta_{2\lambda}\Theta_\beta
J_1+(r\varepsilon)^{-1}a_{2\beta}J_2-r\varepsilon\delta_{2\lambda}J_3
 +r\Theta_2J_4]\Pi(\bm w,\Theta)+(r\varepsilon)^{-1}(2\Theta_{2\alpha}-\widetilde{\Delta}\Theta\delta_{2\alpha})w^\alpha\\
 \quad=J_1\partial_\beta
w^\lambda+J_2\partial_\beta
w^3+\partial_\xi(J_3 w^\lambda+J_4 w^3)\\
\quad\quad+[\varepsilon^{-1}r^{-2}+\varepsilon^{-1}a\Theta_2\Theta_2-r\varepsilon^{-1}\Theta_2\widetilde{\Delta}\Theta]
 \Pi(\bm w,\Theta)+(r\varepsilon)^{-1}(2\Theta_{2\alpha}-\widetilde{\Delta}\Theta\delta_{2\alpha})w^\alpha.
\end{array}$$
Summing up the above conclusions, (A.II.11) becomes
$$\begin{array}{ll}\Delta w^3=&\widetilde{\Delta}
w^3+\frac{\partial}{\partial\xi}[(r\varepsilon)^{-2}a \nabla_3w^3
-\varepsilon^{-1}\Theta_\beta\nabla_\beta
w^3-\varepsilon^{-1}\Theta_\beta\frac{\partial w^3}{\partial
x^\beta}+J_3 w^\lambda+J_4 w^3]\\
&+\frac{\partial}{\partial
x^\beta}(\varepsilon^{-1}\Theta_{\alpha\beta}w^\alpha)
+\varepsilon^{-1}\widetilde{\Delta}\Theta(\frac{w^2}{r}+r\Theta_2\Pi(\bm
w,\Theta)) +J_1\partial_\beta w^\lambda+J_2\partial_\beta
w^3\\
 &+[\varepsilon^{-1}r^{-2}+\varepsilon^{-1}a\Theta_2\Theta_2-r\varepsilon^{-1}\Theta_2\widetilde{\Delta}\Theta]
 \Pi(w,\Theta)+(r\varepsilon)^{-1}(2\Theta_{2\alpha}-\widetilde{\Delta}\Theta\delta_{2\alpha})w^\alpha.
\end{array}\eqno{(A.II.12)}$$
Thanks to the expanded formula
$$J_1\partial_\beta w^\lambda+J_2\partial_\beta w^3=(r\varepsilon)^{-1}(\delta_{2\beta}\Theta_\lambda+r\Theta_{\beta\lambda})\frac{\partial
w^\lambda}{\partial x^\beta}+\frac2r\frac{\partial w^3}{\partial
r}.$$ Hence, similarly,
$$\begin{array}{ll}\Delta w^3&=\widetilde{\Delta}
w^3+\frac{\partial}{\partial\xi}l^3_\xi(\bm w,\Theta)+l^3(\bm w,\Theta),\\
\end{array}\eqno{(A.II.13)}$$where $l^3(w,\Theta),l^3_\xi(w,\Theta)$
are expressed in (A.II.4).
This is the second expression of (A.II.6).

Our goal now is to consider the terms of the pressure. Actually, we
have
$$ g^{\alpha\beta}\partial_\beta p+g^{\alpha3}\partial_\xi
p=\delta^{\alpha\beta}\partial_\beta
p-\varepsilon^{-1}\Theta_\alpha\partial_\xi p,\quad
g^{3\alpha}\partial_\alpha p+g^{33}\partial_\xi
p=-\varepsilon^{-1}\Theta_\alpha\partial_\alpha
p+(r\varepsilon)^{-2}a\partial_\xi p.\eqno{(A.II.14)}$$

  By Using (A.I.24), the nonlinear terms are formulated as
  $$\left\{\begin{array}{ll}
  w^j\nabla_jw^\alpha&=w^\beta\nabla_\beta
  w^\alpha+w^3\nabla_3w^\alpha
  =w^\beta\frac{\partial
  w^\alpha}{\partial x^\beta}+w^3\frac{\partial
  w^\alpha}{\partial\xi}-r\delta_{2\alpha}\Pi(\bm w,\Theta)\Pi(\bm w,\Theta)\\
  w^j\nabla_jw^3&=w^\beta\frac{\partial
  w^3}{\partial x^\beta}+w^3\frac{\partial
  w^3}{\partial\xi}+\varepsilon^{-1}w^\beta
  w^\lambda\Theta_{\beta\lambda}\\
  &+(r\varepsilon)^{-1}\Pi(\bm w,\Theta)[2w^2+(x^2)^2\Theta_2\Pi(\bm w,\Theta)],\\
\end{array}\right.\eqno{(A.II.15)}$$
Combing  (A.II.6) with (A.II.14)-(A.II.15) obtains (A.II.1). The
proof is completed.
\end{proof}

\subsection{The equations for the G\^{a}teaux derivative of
the solutions of NSEs}

In this section we consider the G\^{a}teaux  derivatives of the
solution of NSE with respective to the two dimensional manifold
$\Im$ which is a portion of the solid boundary of the flow passage
in impeller.

\begin{proposition} Assume that there exists a G\^{a}teaux
derivatives ($\widehat{\bm w}:=\frac{{\cal D }\bm w}{{\cal
D}\Theta},\widehat{p}:=\frac{{\cal D }p}{{\cal D}\Theta}$) of the
solutions $(\bm w,p)$ of the Navier-Stokes equations (A.II.13) with
boundary conditions
$$\bm w|_{\Gamma_s}=\bm 0,\quad [\nu\frac{\partial \bm w}{\partial
n}-p\bm n]|_{\Gamma_1}=\bm h,\eqno{(A.III.1)}$$ where $\Gamma_s$ and
$\Gamma_1$ are boundary $\partial \Omega=\Gamma_s\cup \Gamma_1$,
$\Omega=D\times [-1,1]$. Then $(\widehat{\bm w},\widehat{p})$
satisfy the following linearized Navier-Stokes equations with the
corresponding homogenous boundary conditions, that is,
$$\left\{\begin{array}{ll}
\frac{\partial \widehat{w}^\alpha}{\partial x^\alpha}+\frac{\partial
\widehat{w}^3}{\partial\xi}+\frac{\widehat{w}^2}{r}=0,\\
-\nu\widetilde{\Delta} \widehat{w}^\alpha-\nu l^\alpha(\widehat{\bm
w},\Theta)+C^\alpha(\widehat{\bm w},\bm \omega)+\nabla_\alpha
\widehat{p}-\frac{\partial}{\partial\xi}(\nu
l^\alpha_\xi(\widehat{\bm w},\Theta)+\varepsilon^{-1} \Theta_\alpha
\widehat{p})+N_x^\alpha(\bm w,\widehat{\bm w})\\
\quad\quad+N_x^\alpha(\widehat{\bm w},\bm
w)+N_\xi^\alpha(w,\widehat{\bm w})+N_\xi^\alpha(\widehat{\bm w},\bm
w)
+R^\alpha(\bm w,p,\Theta)=0,\\
 -\nu\widetilde{\Delta} \widehat{w}^3 -\nu l^3(\widehat{\bm w},\Theta)+C^3(\widehat{\bm w},\bm \omega)
 -\varepsilon^{-1}\Theta_\beta\partial_\beta \widehat{p}+\frac{\partial}{\partial\xi}(-\nu
l^\alpha_\xi(\widehat{\bm w},\Theta)+(r\varepsilon)^{-2} a
\widehat{p}\\
\quad\quad+N_x^3(\bm w,\widehat{\bm w}) +N_x^3(\widehat{\bm w},\bm
w)+N_\xi^3(\bm w,\widehat{\bm w})+N_\xi^3(\widehat{\bm w},\bm w)
 +R^3(\bm w,p,\Theta)=0,
\end{array}\right.\eqno{(A.III.2)}$$
and
$$\left\{\begin{array}{ll}
\widehat{\bm w}=\bm 0,\quad \mbox{on}\quad\Gamma_s\cap \xi=\xi_k,\\
\nu \frac{\partial \widehat{\bm w}}{\partial\,\bm n}-\widehat{p}\bm
n=0,\quad \mbox{on}\quad \Gamma_{in}\cap\xi=\xi_k,\quad
 \nu \frac{\partial \widehat{\bm w}}{\partial\,\bm n}-\widehat{p}\bm n=0,\quad
 \mbox{on}\quad\Gamma_{out}\cap\xi=\xi_k,
\end{array}\right.\eqno{(A.III.3)}$$
where
$$\left\{\begin{array}{rl}
R^\alpha(\bm w,p)\bm \eta=&-\nu\frac{{\cal D}l^\alpha}{{\cal
D}\Theta}\bm \eta-\frac{\partial }{\partial \xi}[\frac{{\cal
D}l^\alpha_\xi}{{\cal D}\Theta}\bm
\eta+\varepsilon^{-1}p\eta_\alpha]+\frac{\partial
C^\alpha}{\partial\Theta}\bm \eta
-2r\delta_{2\alpha}\Theta_\lambda\Pi(\bm w,\Theta)\eta_\lambda,\\
R^3(\bm w,p)\bm \eta=&-\nu\frac{{\cal D}l^3}{{\cal D}\Theta}\bm
\eta+\frac{\partial }{\partial \xi}[-\frac{{\cal D}l^3_\xi}{{\cal
D}\Theta}\bm \eta+(r\varepsilon)^{-2}2r^2\Theta_\lambda
p\eta_\lambda]+\frac{\partial C^3}{\partial\Theta}\bm
\eta-\varepsilon^{-1}\eta_\alpha\partial_\alpha
p\\
&+\varepsilon^{-1}w^\lambda
w^\sigma\eta_{\lambda\sigma}+(r\varepsilon)^{-1}[w^\lambda(2w^2+r^2\Theta_2\Pi(\bm w,\Theta))\\
&+\Pi(\bm w,\Theta)(\delta_{2\lambda}r^2\Pi(\bm
w,\Theta)+r^2\Theta_2w^\lambda)]\eta_\lambda,
\end{array}\right.\eqno{(A.III.4)}$$
\end{proposition}

\begin{proof}The Navier-Stokes equations (A.II.1) can be rewritten as

$$\left\{\begin{array}{ll}
\frac{\partial w^\alpha}{\partial
x^\alpha}+\frac{w^2}{r}+\frac{\partial w^3}{\partial x^3}=0, \\
{\cal N}^\alpha(\bm w,p,\Theta) \bm e_\alpha+{\cal N}^3(\bm
w,p,\Theta)\bm e_3=f^\alpha \bm e_\alpha+f^3 \bm e_3.
\end{array}\right.\eqno{(A.III.5)}$$
Set G\^{a}teaux derivative with respect with $\Theta$ along any
director
$$\bm \eta\in\,{\cal W}:=H^3(D)\cap \{H^1(D), \mbox{with}\quad \bm w|_{\Gamma_s}=\bm 0,\quad [\nu\frac{\partial
\bm w}{\partial \bm n}-p\bm n]|_{\Gamma_1}=\bm 0\}$$ denoted by
$\frac{{\cal D}}{{\cal D}\Theta}\bm \eta$. Then from (A.III.5) we
assert
$$\begin{array}{ll}
\frac{{\cal D}}{{\cal D}\Theta}{\cal N}^\alpha(\bm w,p,\Theta) \bm
e_\alpha\bm\eta +\frac{{\cal D}}{{\cal D}\Theta}{\cal N}^3(\bm
w,p,\Theta) \bm e_3\bm \eta+{\cal N}^\alpha(\bm
w,p,\Theta)\frac{{\cal D}\bm e_\alpha}{{\cal D}\Theta}\bm \eta+{\cal
N}^3(\bm w,p,\Theta)\frac{{\cal
D}\bm e_3}{{\cal D}\Theta}\bm \eta\\[10pt]
\quad\quad=f^\alpha \frac{{\cal D}\bm e_\alpha}{{\cal D}\Theta}\bm
e_\alpha\bm\eta+f^3
\frac{{\cal D}\bm e_3}{{\cal D}\Theta}\bm e_3\bm\eta,\\[10pt]
\frac{{\cal D}}{{\cal D}\Theta}{\cal N}^\alpha(\bm w,p,\Theta) \bm
e_\alpha +\frac{{\cal D}}{{\cal D}\Theta}{\cal N}^3(\bm w,p,\Theta)
\bm e_3+[{\cal N}^\alpha(\bm w,p,\Theta)-f^\alpha]\frac{{\cal D}\bm
e_\alpha}{{\cal D}\Theta}+[{\cal N}^3(\bm
w,p,\Theta)-f^3]\frac{{\cal D}\bm e_3}{{\cal D}\Theta}\bm e_3=0.
\end{array}$$
Since Navier-Stokes equation (A.III.2), it yields
$$\left\{\begin{array}{ll}
\frac{{\cal D}}{{\cal D}\Theta}{\cal N}^\alpha(\bm
w,p,\Theta)\bm\eta= \frac{\partial }{\partial\Theta}{\cal
N}^\alpha(\bm w,p,\Theta)\bm \eta+ \frac{\partial }{\partial \bm
w}{\cal N}^\alpha(\bm w,p,\Theta)\widehat{\bm w}\bm \eta
+\frac{\partial }{\partial p}{\cal N}^\alpha(\bm w,p,\Theta)\widehat{p}\bm \eta=0,\\
\frac{{\cal D}}{{\cal D}\Theta}{\cal N}^3(\bm w,p,\Theta)\bm \eta=
\frac{\partial }{\partial\Theta}{\cal N}^3(\bm w,p,\Theta)\bm \eta+
\frac{\partial }{\partial \bm w}{\cal N}^3(\bm
w,p,\Theta)\widehat{\bm w}\bm \eta
+\frac{\partial }{\partial p}{\cal N}^\alpha(\bm w,p,\Theta)\widehat{p}\bm \eta  =0,\\
\end{array}\right.\eqno{(A.III.8)}$$ It is obvious that
$$R^\alpha(\bm w,\Theta)\bm\eta=\frac{{\cal D}}{{\cal D}\Theta}{\cal
N}^\alpha(\bm w,p,\Theta)\bm \eta,\quad R^3(\bm
w,\Theta)\bm\eta=\frac{{\cal D}}{{\cal D}\Theta}{\cal N}^3(\bm
w,p,\Theta)\bm \eta.$$
\end{proof}

\section*{\leftline{\large\bf Reference}}

{\parindent=0pt

\def\toto#1#2{\centerline{\hbox to 0.7cm{#1\hss}
\parbox[t]{13cm}{#2}}\vspace{\baselineskip}}

\toto{[1]}{Li Kaitai and Huang Aixiang 1983 Mathematical Aspect of
the Stream-Function Equations of Compressible Turbomachinery Flows
and Their Finite Element Approximation Using optimal Control Comp.
Meth. Appl. Mech. and Eng. 41 175-194.}

\toto{[2]} {Li Kaitai and Huang Aixiang 2000 Tensor Analysis and Its
Applications(in Chinese) Chinese Scientific Press Beijing.}

\toto{[3]} {Li Kaitai, Huang Aixiang and Zhang Wenling A Dimension
Split Method for the 3-D Compressible Navier-Stokes Equations in
Turbomachine Comm. Numer. Meth. Eng. 18(1) 1-14.}

\toto{[4]} { Li Kaitai, Zhang Wenling and Huang Aixiang 2006 An
Asymptotic Analysis Method for the Linearly Shell Theory Science in
China Series A-Mathematics 49(8) 1009-1047.}

\toto{[5]}  {Li Kaitai and Shen Xiaoqin 2007 A Dimensional Splitting
Method for Linearly Elastic Shell Int. J. Comput. Math. 84(6)
807-824.}

\toto{[6]}  {Li Kaitai and Jia Huilian 2008 The Navier-Stokes
Equations in Stream Layer or On Stream Surface and a Dimension Split
Method(in Chinese) Acta Math. Sci. 28(2) 264-283.}

 \toto{[7]} {Temam R and
Ziane M 1996 The Navier-Stokes Equations in Three-Dimensinal Thin
Domains with Various Boundary Conditions Adv. Differential Equations
1(4) 499-546.}

\toto{[8]} {Ciarlet P G 20000 Mathematical Elasticity Vol. III:
Theory of Shells North-Holland.}

\toto{[9]}{ Ciarlet P G 2005 An Introduction to Differential
Geometry with Applications to Elasticity Springer Netherland.}

\toto{[10]} {Il’in A A 1991 The Navier-Stokes and Euler Equations
on Two-Dimensional Closed Manifolds Math. USSR Sbornik 69(2)
559-579.}

\toto{[11]}{ Ebin D G and Marsden J 1970 Groups of diffeomorphisms
and the motion of an incompressible fluid Ann. of Math. 92(2)
102-163.}

\toto{[12]}{ Raugel G and Sell G 1993 Navier-Stokes Equations On
Thin 3D Domains I: Global Attractors and Global Regularity of
Solutions J. Am. Math. Soc. 6 503-568.}

\toto{[13]} {Raugel G and Sell G 1992 Navier-Stokes Equations On
Thin 3D Domains II: Global Regularity Of Spatially Periodic
Conditions College de France Proceedings Pitman Res. Notes Math.
Ser. Pitman New York London.}

\toto{[14]} {Hale J K and Ruagel G 1992 Reaction Diffusion Equations
On Thin Domains J. Math. Pure Appl. 71 33-95.}

\toto{[15]}{ Hale J K and Ruagel G 1992 A Damped Hyperbolic
Equations On Thin Domains Tran. Amer. Math. Soc. 329 185-219.}

\toto{[16]} {Chen Wenyi and Jost J 2002 A Riemannian Version of
Korn’s Inequality Calc. Var. 14 517-530.}

\toto{[17]}{Girault V and Raviart P A 1986 Finite Element Methods
for Navier-Stokes Equations: Theory and Algorithms Springer-Verlag
Berlin Heidelberg.}

\toto{[18]}{ Li Kaitai and Hou Yanren 2001 An AIM and One Step
Newton Method for the Navier-Stokes Equations Comput. Methods Appl.
Mech. Engrg. 190 198-317.}

\toto{[19]}{Layton W, Maubach J, Rabier P and Sunmonu 1992 A
Parallel finite element methods In: Proceedings of the ISMM
International Conference on Parallel and Distributed Computing and
Systems (ed. Melhem R) 299-304 Pittsburgh Pennsylvania U.S.A. }

\toto{[20]}{John V, Layton W and Sahin N 2004 Derivation and
analysis of near wall models for channel and recirculating flows
Comput. math. appl.48 1135-1151. }

\toto{[21]}{ Layton W and Lenferink W 1995 Two-Level Picard Defect
Correction for the Navier- Stokes Equations Appl. Math. and Comput.
80 1-12.}

\toto{[22]}{Ervin V, Layton W and Maubach J 1996 A posteriori error
estimation for a two level finite element method for the
Navier-Stokes equations Numer. Meth. PDE 12 333-346.}

\toto{[23]}{Layton W, Lee H K and Peterson J 1998 Numerical solution
for the stationary Navier-Stokes equations by a multi level finite
element method SIAM J. Scientific Computing 20 1-12.}

 \toto{[24]}{Tao Lu, Shih T M and  Liem C B 1992 Domain
Decomposition Methods--New Numerical Techniques for Solving PDE
Science Press at Beijing (in Chinese).}

}

\end{document}